\documentclass[aps,prx,reprint,superscriptaddress,longbibliography,nofootinbib]{revtex4-2}

\usepackage{cancel}
\usepackage[normalem]{ulem}

\makeatletter
\newif\ifeqcontrib@this
\newif\ifeqcontrib@any
\eqcontrib@thisfalse
\eqcontrib@anyfalse

\newcommand{\eqcontrib}{\global\eqcontrib@thistrue\global\eqcontrib@anytrue}

\newcommand{\eqcontribmark}{\textsuperscript{\ensuremath{,\S}}}

\newcommand{\eqcontrib@maybe}{%
  \ifeqcontrib@this
    \global\eqcontrib@thisfalse
    \eqcontribmark
  \fi
}

\newcommand{\printEqContrib}{%
  \ifeqcontrib@any
    \begingroup
      \renewcommand{\thefootnote}{\ensuremath{\S}}%
      \footnotetext{These authors contributed equally to this work.}%
    \endgroup
  \fi
}

\def\doauthor#1#2#3{%
  \ignorespaces#1\unskip\@listcomma
  \begingroup
    #3%
  \@if@empty{#2}{\endgroup{}{}}{\endgroup{\comma@space}{}\frontmatter@footnote{#2}}%
  \eqcontrib@maybe
  \space \@listand
}%
\makeatother

\usepackage{comment}
\usepackage[T1]{fontenc}
\usepackage{lmodern}
\usepackage{microtype}
\usepackage{amsmath,amssymb,mathtools,bm}
\usepackage{amsthm}
\usepackage{graphicx}
\usepackage{booktabs}
\usepackage{xcolor}
\usepackage{hyperref}
\usepackage[nameinlink,capitalise]{cleveref}
\usepackage[switch]{lineno}

\nolinenumbers
\hypersetup{
    colorlinks=true,
    linkcolor=blue,
    citecolor=blue,
    urlcolor=blue
}

\newtheorem{theorem}{Theorem}
\newtheorem{proposition}[theorem]{Proposition}
\newtheorem{corollary}[theorem]{Corollary}
\newtheorem{lemma}[theorem]{Lemma}
\newtheorem{definition}[theorem]{Definition}

\newtheorem{example}[theorem]{Example}

\makeatletter
\newcommand{\supplementtableofcontents}{%
    \section*{Contents}%
    \@starttoc{stoc}%
}
\makeatother
\newcommand{\suppsection}[1]{%
    \section{#1}%
    \addcontentsline{stoc}{section}{\protect\numberline{\thesection}#1}%
}

\usepackage{physics}

\newcommand{\cP}{\mathcal P}
\newcommand{\cL}{\mathcal L}
\newcommand{\cD}{\mathcal D}
\renewcommand{\Tr}{\operatorname{Tr}}
\newcommand{\supp}{\operatorname{supp}}
\newcommand{\poly}{\operatorname{poly}}
\newcommand{\e}{\mathrm e}
\renewcommand{\norm}[1]{\left\lVert #1\right\rVert}
\renewcommand{\abs}[1]{\left\lvert #1\right\rvert}

\begin{document}
\title{Dissipation accelerates quantum and classical simulation of open-system dynamics}

\author{Armando Angrisani\eqcontrib}
\affiliation{Institute of Physics, Ecole Polytechnique F\'ed\'erale de Lausanne, 1015 Lausanne, Switzerland}
\affiliation{Centre for Quantum Science and Engineering, Ecole Polytechnique F\'{e}d\'{e}rale de Lausanne (EPFL),   Lausanne, Switzerland}
\affiliation{Department of Mathematical Sciences, University of Copenhagen, Universitetsparken 5, 2100 Denmark}
\author{Ricard Puig\eqcontrib}
\affiliation{Institute of Physics, Ecole Polytechnique F\'ed\'erale de Lausanne, 1015 Lausanne, Switzerland}
\affiliation{Centre for Quantum Science and Engineering, Ecole Polytechnique F\'{e}d\'{e}rale de Lausanne (EPFL),   Lausanne, Switzerland}
\author{Yanting Teng}
\affiliation{Institute of Physics, Ecole Polytechnique F\'ed\'erale de Lausanne, 1015 Lausanne, Switzerland}
\affiliation{Centre for Quantum Science and Engineering, Ecole Polytechnique F\'{e}d\'{e}rale de Lausanne (EPFL),   Lausanne, Switzerland}
\author{Zo\"{e} Holmes}
\affiliation{Institute of Physics, Ecole Polytechnique F\'ed\'erale de Lausanne, 1015 Lausanne, Switzerland}
\affiliation{Centre for Quantum Science and Engineering, Ecole Polytechnique F\'{e}d\'{e}rale de Lausanne (EPFL),   Lausanne, Switzerland}

\date{\today}

\begin{abstract}
Simulating open quantum systems reveals how environmental coupling shapes relaxation, excitation transport, and the dynamics of quantum correlations. On quantum hardware, dissipative channels add operations and might seem to increase cost. However, we show that a broad class of Pauli noise, including depolarization, can ease quantum and classical simulation by exponentially suppressing high-weight components of Heisenberg-evolved observables. For bounded-degree Lindblad dynamics, this permits system-size-independent Trotter steps when estimating local observables. Combining this compression with Richardson extrapolation, we provide classical and quantum error bounds with respect to the 2-norm of the propagated observable. These give high-probability expectation-value guarantees for random input states or random Hamiltonians acting on a fixed state and control (out-of-time-order) correlators. Concretely, we show that expectation values of local observables at time $t$ can be estimated to accuracy $\varepsilon$ using circuits of maximum depth $\mathcal{O}\!\left(\left[1+\left(t/\gamma\right)^{3/2}\right]\log^2(1/\varepsilon)\right)$, independent of the system size $n$. Classically, we show that sparse Pauli propagation runs in time polynomial in $n$, $t$ and $1/\varepsilon$ for every fixed $\gamma>0$, with polynomial degree $\mathcal O(1/\gamma)$ in the weak dissipation limit. The gap between our quantum and classical upper bounds leaves room for a substantial polynomial quantum advantage as dissipation weakens, an observation further supported by our numerical results.
\end{abstract}

\maketitle
\printEqContrib

\section{Introduction}

Simulating open quantum systems is essential for understanding phenomena in which coupling to an environment shapes the observed dynamics. Such simulations connect NMR relaxation to molecular tumbling and internal motion~\cite{redfield1957theory,chen2018ab, seetharamDigitalQuantumSimulation2023}, describe excitation transport through molecular and biological complexes~\cite{mohseni2008environment,ishizaki2009theoretical}, and characterize transport and nonequilibrium steady states~\cite{prosen2011exact}. They also help explain decoherence and the emergence of classical behavior~\cite{joos1985emergence,habib1998decoherence}, while quantifying how noise limits quantum sensing~\cite{smirne2016ultimate,puig2024dynamical}.

On quantum hardware, environmental coupling appears to add cost because open-system evolution requires dissipative channels alongside coherent operations. However, the resulting dissipation can also simplify the observables whose dynamics we seek to approximate. To see why this might help quantum simulation, consider a product-formula approximation built from local coherent and dissipative operations. Coherent interactions generally spread an initially local observable through the system. As the observable spreads, more commutators can contribute to its product-formula error in Hamiltonian dynamics~\cite{childs2021theory}. Here we show that dissipation can counter this growth to permit larger time steps and thereby reduce simulation depths. 

Classically, direct evolution of an $n$-qubit density operator requires up to $4^n$ coefficients, compared with up to $2^n$ amplitudes for a state vector~\cite{donvil2022quantum,sander2025large}. Several methods reduce this cost by exploiting additional structure. Quantum-trajectory methods replace density-matrix evolution with stochastic sampling of state-vector trajectories~\cite{pocklington2025efficient,cao2026dynamically, chen2024optimized}, while low-rank and tensor-network methods compress the dynamics when the required rank or bond dimension remains controlled~\cite{dalibard1992wave,zwolak2004mixed,sander2025large, werner2016positive,cichy2026classical}. Their efficiency therefore depends on the sampling cost or the compressibility of the evolving states and operators.

Spatial locality provides a complementary route to reducing both quantum and classical simulation costs. Locality of Lindbladians can sharpen product-formula error bounds \cite{wang2026lindbladian,barthel2012quasilocality}, while light-cone bounds allow classical simulations of local observables to be restricted to a finite spatial neighborhood~\cite{barthel2012quasilocality,wild2023classical}. The cost of this spatial truncation grows with the neighborhood's volume, particularly in higher dimensions and at longer times. Our approach exploits a different structure: dissipation suppresses Pauli weight, which can remain controlled even as the observable's spatial support expands.

Pauli propagation provides a natural representation for exploiting this structure classically. It expands a Heisenberg-evolved observable in the Pauli basis and tracks its coefficients on a classical computer \cite{rudolph2025pauli,beguvsic2023fast}. Under interacting dynamics, an initially local observable can become a sum of exponentially many Pauli operators, making exact tracking expensive. Noisy-circuit algorithms show that local noise can suppress this growth and, in some regimes, permit polynomial-time simulation~\cite{gao2018efficient,noh2020efficient,fontana2023classical,aharonov2022polynomial,schuster2024polynomial,gonzalez2024pauli,angrisani2025simulating, martinez2025efficient}. However, these guarantees assume noise of fixed strength per layer or gate. A finer product-formula approximation contains more layers with proportionally weaker noise in each layer, so these arguments do not usefully apply to simulating continuous-time dynamics. The challenge is to obtain a guarantee that remains valid as the time step tends to zero.

In this work, we show that dissipation can ease both quantum and classical open quantum system simulation by compressing Heisenberg-evolved observables. For a class of local Pauli-diagonal noise that includes depolarization, high-weight components of initially local observables are exponentially suppressed, and only a controlled number of Pauli coefficients remain appreciable. We call this \emph{dissipation-induced compression}. Our bounds compare the continuous coherent and dissipative rates directly, so they remain valid as the product-formula step shrinks. 

On quantum hardware, observable compression allows larger Trotter steps than error bounds that ignore dissipation would suggest. For local observables in bounded-degree systems, the step-size guarantee is independent of system size and holds with high probability over input states drawn from a 1-design. Thus, the improvement lies in the error bound for a standard Trotter simulation. We then apply standard Richardson extrapolation, combining simulations at several step sizes to cancel leading errors. Theorem~\ref{thm:main-quantum-simulation} establishes that for fixed evolution time and other local parameters, the maximum circuit depth scales as $\mathcal{O}\left(\gamma^{-3/2}\log^2(1/\varepsilon)\right)$ in the weak-damping regime, where $\gamma$ is the minimum local damping rate. 

Classically, the observable compression allows Pauli propagation to discard small coefficients even as the product formula is refined. With Richardson extrapolation, we obtain in Theorem~\ref{thm:main-classical-simulation} an error $\varepsilon$ in time $\widetilde{\mathcal O}\left[        n+\left(            (\e+Jt(1+J/\gamma))/\varepsilon        \right)^{\mathcal O(1+J/\gamma)}    \right]$, for evolution time $t$ and fixed local parameters, where $J$ bounds the local coherent strength. We use the standard notation where $\widetilde{\mathcal{O}}$ hides the logarithmic factors. The runtime is polynomial in the system size and inverse accuracy for every fixed $\gamma>0$. For gate-based local depolarizing noise, the same framework improves the prior general runtime upper bound from quasipolynomial to polynomial \cite{schuster2024polynomial}.

Our results follow from a bound on the normalized Hilbert--Schmidt error between exact and simulated Heisenberg observables. While Theorems~\ref{thm:main-quantum-simulation} and~\ref{thm:main-classical-simulation} state high-probability expectation-value guarantees for input states drawn from a 1-design, our results also apply directly to infinite-temperature correlations, out-of-time ordered correlators (OTOCs)~\cite{roberts2015diagnosing,maldacena2016bound, roberts2017chaos, google2025observation, barron2026observable}, and one clean qubit model (DQC1)~\cite{knill1998power} expectation values without changing the analysis.

The two bounds leave room for a widening quantum advantage as dissipation weakens: the quantum depth bound grows through a polynomial prefactor in $1/\gamma$, while the classical runtime bound pays an increasing power of $1/\varepsilon$, proportional to $J/\gamma$. Our numerics show a complementary separation in how the two resources grow with system size. In an all-to-all Heisenberg benchmark at fixed error, increasing the system from $3$ to $12$ qubits raises the estimated ideal circuit depth by only a factor of about $5$--$10$, while the measured Pauli-propagation time rises by roughly four to five orders of magnitude (Fig.~\ref{fig:relative_resource}). This difference in relative growth is generally greatest at weaker noise. Although the curves do not compare absolute architecture-dependent end-to-end runtimes, they suggest weakly dissipative dynamics is a promising setting for achieving quantum advantage with shallow circuits.

\section{Framework}
\label{sec:framework}

\subsection{Model}
We consider classical and quantum algorithms for simulating the dynamics of an $n$-qubit system governed by a time-independent local Lindbladian. To analyse our algorithms it is convenient to view both simulation tasks in the Heisenberg picture, where an observable $O$ evolves under the generator $\cL$ as
\begin{equation}
    O(t)=\e^{t\cL}(O).
\end{equation}
For an initial state $\rho$, this evolution determines expectation values of the form $\Tr[\rho O(t)]$. We will also consider correlation functions built directly from evolved observables, including OTOCs.

We assume that the Lindbladian decomposes into local blocks,
\begin{align}
    \cL&=\sum_{a=1}^{m}\cL_a,\\
    \cL_a&=\mathcal H_a+\cD_a.
    \label{eq:main-local-block-decomposition}
\end{align}
Here the index $a$ labels a local block containing a coherent generator $\mathcal H_a$ and a dissipative contribution $\cD_a$. To make the coherent local structure explicit, we expand the Hamiltonian in the phase-free Pauli basis
\[
    \cP_n\coloneqq\{I,X,Y,Z\}^{\otimes n}.
\]
For $P\in\cP_n$, its support $\supp(P)$ is the set of qubits on which it acts nontrivially, and its Pauli weight is
\[
    \abs{P}\coloneqq\abs{\supp(P)}.
\]

The coherent part of the dynamics is specified by the Hamiltonian operator
\begin{equation}
    H=\sum_{a=1}^{m}H_a,
    \qquad
    H_a\coloneqq J_aP_a \, ,
\end{equation}
where $P_a\in\cP_n$ and $J_a\in\mathbb R$, with $\abs{P_a}\leq k$ and $\abs{J_a}\leq J$. In the main text, we assume that the Pauli strings in the Hamiltonian have a constant weight, i.e.,  $k \in \mathcal{O}(1)$, to simplify our bounds but provide bounds for the general case in the appendices.
We further assume that each qubit lies in $\supp(P_a)$ for at most $d$ values of $a$. Thus, $k$ bounds the size of each interaction, $J$ bounds its strength, and $d$ is the \emph{local degree} of the Lindbladian, which bounds the number of interactions meeting at any qubit. 

Here $H_a$ is a local Hermitian operator on the system Hilbert space. In the Heisenberg picture, it induces the linear generator
\begin{equation}
    \mathcal H_a(O)
    \coloneqq i[H_a,O]
    =iJ_a[P_a,O] 
\end{equation}
on observables, as appears in the coherent part of Eq.~\eqref{eq:main-local-block-decomposition} and acts on observables as $e^{t\mathcal H}(O)=e^{itH}Oe^{-itH}$.

Each coherent interaction $\mathcal H_a$ is paired with a local Markovian noise process $\cD_a$ acting on the same region. Writing
\[
    S_a\coloneqq\supp(P_a),
\]
we describe this process in the Heisenberg picture by a valid Lindblad generator $\cD_a$ supported on $S_a$. We assume that $\cD_a$ is diagonal in the Pauli basis. That is, every Pauli operator is an eigenoperator:
\begin{equation*}
    \cD_a(P)=-\lambda_a(P)P,
    \qquad
    P\in\cP_n.
\end{equation*}
Pauli diagonality means that the noise does not mix different Pauli operators. Instead, it damps the coefficient of each Pauli operator at a rate $\lambda_a(P)$. We assume that such rate is non-zero for every non-identity Pauli with support in $S_a$ and $a \in [m]$ , and suppose 
\begin{align}
   \gamma \leq  \lambda_a(P) \leq \bar \gamma 
\end{align}
for two positive values $\gamma, \bar \gamma$.

The rate $\bar\gamma$ bounds the local dissipative strength and enters our product-formula and truncation error bounds. To simplify these bounds, we assume $\bar\gamma \in \mathcal{O}(J)$ in the main text, meaning that the noise strength is at most of the same order as the interaction strength. A reader who is more familiar with noise channels used in circuit models might find it helpful to consider an evolution interval of duration $\delta$. The Heisenberg action of the corresponding local noise channel is
\begin{equation*}
     \mathcal N_{a,\delta}
    \coloneqq \e^{\delta\cD_a}.
\end{equation*}
Since the Pauli operators are eigenoperators of $\cD_a$, the channel acts as
\begin{equation*}
    \mathcal N_{a,\delta}(P)
    =
    \e^{-\delta\lambda_a(P)}P
    =
    \left(1-\delta\lambda_a(P)+O(\delta^2)\right)P.
\end{equation*}
Thus, $\mathcal N_{a,\delta}$ has the Pauli-diagonal form familiar from circuit `Pauli noise' models, but $\lambda_a(P)$ is a continuous-time decay rate, not the probability of a Pauli error in one step. Over a short step, the coefficient of $P$ is linearly in $\delta$ suppressed to leading order.

Local depolarizing noise is the canonical example of this setting. More generally, the setting encompasses stochastic Pauli processes, including spatially correlated Pauli noise, provided that every nonidentity Pauli mode is damped. Our results do not cover Pauli noise with undamped directions, such as single-axis dephasing and bit-flip noise. They also exclude amplitude damping, which is not Pauli diagonal. 

\subsection{Product-formula approximation}
\label{sec:product-formula}

To approximate the continuous-time evolution by products of local evolutions, we group local blocks into parallel color layers~\cite{tranter2019ordering}. Specifically, we color the labels $a \in [m]$ so that blocks with overlapping supports receive different colors. This yields color classes $C_1,\ldots,C_\chi$, with $\chi \leq k(d-1)+1$, and we define
\begin{equation}
    \mathcal H^{(c)}
    \coloneqq \sum_{a\in C_c}\mathcal H_a,
    \qquad
    \cD^{(c)}
    \coloneqq \sum_{a\in C_c}\cD_a.
\end{equation}
Within each color class, the local maps have disjoint supports and can therefore be applied in parallel. In particular, for distinct $a,b \in C_c$, the Hamiltonian terms $H_a$ and $H_b$ act on disjoint sets of qubits, allowing the unitaries $e^{-iH_a\delta}$ and $e^{-iH_b\delta}$ to be applied simultaneously. For example, the nearest-neighbor Hamiltonian on an open chain, $H = \sum_{i=1}^{n-1} Z_i Z_{i+1}$, admits a two-color decomposition: one class contains the terms with odd $i$, and the other contains those with even $i$. The color layers sum to the full generator, $\cL=\sum_{c=1}^{\chi}(\mathcal H^{(c)}+\cD^{(c)})$.

The familiar first-order Trotter formula approximates the Lindbladian evolution by applying the coherent and dissipative parts of each color layer in succession. For a step of duration $h$, let
\begin{align}
    T_{c,h}
    &\coloneqq
    \e^{h\cD^{(c)}}\e^{h\mathcal H^{(c)}},
    \\
    \Phi_h
    &\coloneqq
    T_{\chi,h}\cdots T_{1,h}.
    \label{eq:main-first-order-product-formula}
\end{align}
Thus $T_{c,h}$ first evolves the observable coherently and then dissipatively within color $c$.  Then if we apply $N$ steps while setting $h=t/N$, the repeated approximation $\Phi_h^N$ converges to $\e^{t\cL}$ as $N\to\infty$. Taking shorter steps reduces the ``Trotter error'', but reaching the same evolution time then requires more layers.

To reduce the Trotter error, and thereby simulation costs, both algorithms below will use Richardson extrapolation~\cite{wang2026lindbladian,low2019well, rendon2024improved, watson2024randomly, watson2025exponentially} (as detailed in Section~\ref{sec:methods}). When doing so, it is advantageous to use a symmetric product formula, rather than the first-order one, because this accelerates the convergence of Richardson extrapolation. Namely, we symmetrize the coherent and dissipative evolution within each color and arrange the color layers in a palindromic order. Define the update for color $c$ and the full step by
\begin{align}
    Q_{c,h}
    &\coloneqq
    \e^{h\cD^{(c)}/2}
    \e^{h\mathcal H^{(c)}}
    \e^{h\cD^{(c)}/2},\\
    \mathsf{S}_h
    &\coloneqq
    Q_{1,h/2}\cdots Q_{\chi-1,h/2}
    Q_{\chi,h}
    Q_{\chi-1,h/2}\cdots Q_{1,h/2}.
    \label{eq:main-symmetric-product-formula}
\end{align}

For $h=t/N$, the approximation $\mathsf{S}_h^N$ converges to $\e^{t\cL}$ as $N\to\infty$. On quantum hardware this is the usual symmetric product-formula circuit. The Heisenberg-picture form is simply the representation used in our analysis.

\subsection{Pauli propagation}
\label{sec:pauli-propagation}

Pauli propagation~\cite{rudolph2025pauli, beguvsic2023fast, rudolph2023classical, shao2023simulating, angrisani2024classically, lerch2024efficient, fontana2023classical, beguvsic2024real, fuller2025improved, rall2019simulation, beguvsic2023simulating, rudolph2026thermal, aharonov2022polynomial, schuster2024polynomial, gonzalez2024pauli, cirstoiu2024fourier, angrisani2025simulating, teng2025leveraging, xu2026classical, martinez2025efficient} is a classical simulation method for simulating quantum circuits. The algorithm works in the Pauli basis in the Heisenberg picture by writing an observable as \(    O=\sum_{P\in\cP_n}c_P P \) and storing its Pauli strings and coefficients. The action of a circuit is then computed via applying each gate as local update rules in the Pauli basis. Coherent evolution leaves $P$ unchanged when it commutes with $P_a$. When it anticommutes, it rotates $P$ into a second Pauli direction,
\begin{equation}
    \e^{s\mathcal H_a}(P)
    =
    \cos(2J_as)P+\sin(2J_as)Q_{a,P},
    \label{eq:main-pauli-propagation-rotation}
\end{equation}
where $Q_{a,P}\coloneqq iP_aP$ is a signed Pauli operator. Dissipative evolution rescales each Pauli operator,
\begin{equation}
    \e^{s\cD_a}(P)=\e^{-s\lambda_a(P)}P.
    \label{eq:main-pauli-propagation-noise}
\end{equation}
Thus, noise rescales Pauli coefficients, whereas coherent evolution can create new ones.

The difficulty is that repeated coherent updates can make the number of Pauli terms grow exponentially and therefore truncation schemes are essential. Contributions to the same Pauli operator are first combined (i.e., merged), after which we discard coefficients smaller than a threshold $\tau$ and Pauli operators of weight greater than a cutoff $w$. The coefficient cutoff controls the number of stored terms. The weight cutoff controls the cost of propagating each retained term through the next layer.

Previous guarantees for Pauli propagation in noisy circuits rely on every circuit layer reducing the magnitude of nonidentity Pauli coefficients by a fixed factor. As these coefficients become small, they can be discarded without introducing much error. This prevents the number of coefficients that must be tracked from growing uncontrollably.

This argument does not directly apply to a product-formula approximation of continuous-time dynamics. If $h=t/N$, each dissipative sub-step has duration $s=O(h)$. An active Pauli coefficient is then multiplied by $\e^{-s\lambda_a(P)}\leq\e^{-\gamma s}$, so the fraction by which it is guaranteed to decrease is only
\begin{equation}
    1-\e^{-\gamma s} \in O(\gamma h) \in O(\gamma t/N) \, .
\end{equation}
Thus we face two interrelated barriers: (i) each layer suppresses the coefficients less and less as the Trotter error is reduced by taking $N\to\infty$ and (ii) the number of layers at which truncation is performed grows as $N$. Substituting this vanishing noise strength, and growing depths, into existing noisy-circuit bounds gives estimates that deteriorate as the Trotter error is reduced. 

We resolve these issues via the observation that the same small step also limits how much amplitude the coherent evolution can transfer into new Pauli coefficients. During a step of duration $h$, a local Hamiltonian term changes the coefficients it couples by $O(Jh)$, while the paired dissipator suppresses each active coefficient by a relative amount  $\order{\overline{\gamma}h}$. Thus, refining the product formula weakens coefficient generation and suppression in tandem. Their competition is governed by the continuous rates $J$ and $\gamma$, and in particular by the ratio $J/\gamma$, rather than by the vanishing noise in any one Trotter layer. Our key observation is that this balance limits how much weight can spread into a growing collection of small coefficients. We capture this balance through a Pauli-norm bound that is uniform in $h$. Consequently, small coefficients can be discarded after successive layers without the truncation error deteriorating as the Trotter resolution is increased. Below we state the resulting algorithm complexities with the intuitions and methods that lead to them provided in more detail in Section~\ref{sec:methods}.

\section{Results}
\label{sec:results}

\subsection{Assumptions and performance metric}

We work throughout with the local Lindbladian and damping assumptions of Section~\ref{sec:framework} and the symmetric product formula described in Section~\ref{sec:product-formula}. Let $O$ be supported on at most $w_0$ qubits and normalized so that $\norm{O}_2=(2^{-n}\Tr[O^2])^{1/2}=1$. We estimate
\begin{equation}
    \langle O(t)\rangle_\rho
    \coloneqq
    \Tr\!\left[\rho\,O(t)\right]
    =
    \Tr\!\left[\rho\,\e^{t\cL}(O)\right].
    \label{eq:main-target-expectation}
\end{equation}
We draw $\rho$ from an ensemble $\mathsf E$ satisfying $\mathbb E_{\rho\sim\mathsf E}[\rho]=I/2^n$ and measure the root-mean-square error over this ensemble. For example, $\mathsf E$ may be uniform over any orthonormal product basis. By Markov's inequality, a root-mean-square (RMS) error of $\varepsilon\sqrt{\delta}$ implies an additive error of at most $\varepsilon$ with probability at least $1-\delta$ over $\rho\sim\mathsf E$.

We state the main theorems for expectation values here, but their proofs first bound the error between the exact and simulated Heisenberg observables in normalized Hilbert--Schmidt norm. We state and prove these observable-level bounds in the appendices, then use them to obtain the expectation-value guarantees above. As detailed in Appendix~\ref{supp:operational}, the same bounds also control infinite-temperature correlation functions, OTOCs, and DQC1 output expectations without changing the underlying simulation analysis. They further imply high-probability expectation-value guarantees over input states sampled from a state $1$-design. For a fixed Pauli observable, analogous guarantees hold over Hamiltonians drawn from a suitable Pauli-invariant ensemble, for any input state fixed independently of the Hamiltonian.

For the classical result, we assume that the local terms are specified classically, and that $\Tr(\rho P)$ can be evaluated in time $\mathcal O(\abs{P})$ for every retained Pauli operator $P$. For the quantum result, we assume that $\rho$ can be prepared and that each local coherent or dissipative evolution can be implemented in constant depth for fixed $k$. The case $J=0$ is Pauli diagonal and can be treated exactly, so below we take $J>0$.

\subsection{Quantum simulation}

The quantum algorithm uses a standard symmetric product formula to simulate the local, time-independent Lindblad evolution. Our contribution is a stronger guarantee for this procedure. Noise suppresses the high-weight Pauli terms that would otherwise cause the product-formula error to grow as the observable spreads. Consequently, with high probability over $\rho\sim\mathsf E$, one can use a larger time step than worst-case bounds allow. The required step size is set by local parameters rather than by $n$. We further improve the precision dependence using Richardson extrapolation, an existing error-cancellation technique. The algorithm simulates the same total time $t$ with several different product-formula step sizes and combines the resulting expectation values with signed weights. These weights cancel the leading powers of the step-size error.

\begin{theorem}[Quantum simulation]
\label{thm:main-quantum-simulation}
Under the target problem assumptions and quantum access model above, for every $0<\varepsilon\leq1$, there exists a quantum algorithm that outputs an estimate $\widehat f$ of $\langle O(t)\rangle_\rho$ satisfying
\begin{equation}
    \left(
        \mathbb E_{\rho\sim\mathsf E}
        \abs{\widehat f-\langle O(t)\rangle_\rho}^2
    \right)^{1/2}
    \leq \varepsilon.
\end{equation}
The maximum circuit depth obeys
\begin{equation}
    \begin{aligned}
    D_{\mathrm Q}
    &\in
    \mathcal O\!\Bigg[
        \Bigg(d+d^{5/2}J^{3/2}t^{3/2}\\[-2pt]
    &\hspace{4.2em}{}
        \times\left(1+w_0+\frac{J}{\gamma}\right)^{3/2}
       \Bigg) \log^2\!\left(\frac{\e}{\varepsilon}\right)
    \Bigg].
    \end{aligned}
    \label{eq:main-quantum-depth}
\end{equation}
\end{theorem}

That the step-size is independent of system size is the key new ingredient in this result, compared to other similar algorithms~\cite{wang2026lindbladian, barthel2012quasilocality, wang2026query,chen2026query,chen2025randomized,kato2026exponentially,ding2024simulating,ding2024single,cleve2016efficient}. For fixed $w_0$, its weak-damping dependence scales as $\gamma^{-3/2}$. Richardson extrapolation then gives the polylogarithmic precision dependence using $O(\!\log(1/\varepsilon))$ circuit variants, which only introduces a $\order{{\rm poly}\log 1/\epsilon}$ factor in measurements. Appendix~\ref{sec:quantum_compar} gives a more detailed comparison with state of the art algorithms.

\subsection{Classical simulation}

The classical algorithm applies the sparse Pauli-propagation procedure of Section~\ref{sec:pauli-propagation} to the same symmetric product-formula circuits. After each layer, it combines identical Pauli strings and discards small coefficients and high-weight terms. Our contribution is to show that dissipation-induced compression controls the accumulated truncation error uniformly in the product-formula resolution. However, the retained Pauli expansion still grows as the target error decreases, so the runtime remains polynomial in $1/\varepsilon$. 

\begin{theorem}[Classical simulation]
\label{thm:main-classical-simulation}
Under the target problem assumptions and classical access model above, for every $0<\varepsilon\leq1$, there exists a classical algorithm that outputs an estimate $\widehat f$ of $\langle O(t)\rangle_\rho$ satisfying
\begin{equation}
    \left(
        \mathbb E_{\rho\sim\mathsf E}
        \abs{\widehat f-\langle O(t)\rangle_\rho}^2
    \right)^{1/2}
    \leq \varepsilon.
\end{equation}
Suppressing multiplicative prefactors that depend only on the fixed  $w_0$ and physical parameters, its total runtime obeys
\begin{equation}
    T_{\mathrm C}
    \in
    \widetilde{\mathcal O}\!\left[
        n+(1+t^{3/2})^{1+\beta_\star}
        \varepsilon^{-\beta_\star}
        \log^{3+5\beta_\star}\!\left(\frac{\e}{\varepsilon}\right)
    \right],
    \label{eq:main-classical-runtime}
\end{equation}
where 
\begin{equation}
    \beta_\star
    \coloneqq
    \max\!\left\{2,\frac{4J}{\gamma}\right\}
    +\frac{k\log 2}
           {\operatorname{arsinh}(\gamma/(2J))}
    =\mathcal O\!\left(1+\frac{J}{\gamma}\right).
    \label{eq:main-classical-exponent}
\end{equation}
In the weak-damping regime, this scaling simplifies to 
\begin{align}
T_{\mathrm C}\in\widetilde{\mathcal O}\!\left[
n+\left(
\frac{J^2 t}{\gamma\varepsilon}
\right)^{\mathcal O(J/\gamma)}
\right].
\end{align}

\end{theorem}

We recall that $\widetilde{\mathcal{O}}$ hides the logarithmic factors. The new compression bound controls the cost of the truncated propagation. Applying the same Richardson extrapolation then reduces the number of propagated layers from $O(\varepsilon^{-1/2})$ to $O(\log^2(1/\varepsilon))$ and improves the precision dependence from $\varepsilon^{-(3\beta_\star+1)/2}$ to $\varepsilon^{-\beta_\star}$, up to logarithmic factors. The polynomial degree therefore grows as the damping weakens. Compared to other continuous-time classical algorithms this means that there is no exponential dependence in time or system~\cite{xu2026classical, wild2023classical, barthel2012quasilocality}. A more in-depth comparison is in Appendix~\ref{sec:classical_compar}.

\medskip

\noindent\emph{Improved classical simulation of noisy circuits.} Our classical simulation guarantees extend to discrete-time dynamics (e.g., \textit{Floquet} dynamics), yielding improved runtime bounds for quantum circuits with local depolarizing noise. For estimating expectation values on randomly sampled input states, Ref.~\cite{schuster2024polynomial} establishes a polynomial runtime under uniform noise and a quasipolynomial runtime under gate-based noise. Gate-based noise acts only on the qubits participating in each gate and therefore requires no noise on idle qubits. Uniform noise, by contrast, acts on every qubit after each unitary layer, including idle qubits.

Our framework improves the runtime for gate-based noise from quasipolynomial to polynomial time. Moreover, our runtime bound explicitly captures the competition between unitary evolution and noise through the ratio $\delta/\gamma$ defined below. While we state the result for gate-based noise, we remark that uniform noise is recovered as a special case by inserting noisy identity gates on the idle qubits of each layer.

\begin{theorem}[Classical simulation with gate-based noise]
\label{thm:main-gate-noise}
Let $\mathcal C$ be a noisy $n$-qubit circuit of $\poly(n)$ gates of the form $U_j=\e^{-i\delta H_j}$, with $\delta\geq0$ and $\norm{H_j}_{\mathrm{op}}\leq1$. Each gate acts on at most $k=O(1)$ qubits and is followed by independent local depolarizing noise on those qubits, with $\mathcal N_\gamma(P)=\e^{-\gamma}P$ for $P\in\{X,Y,Z\}$ and $\gamma>0$. Let $O$ be as in the target problem assumptions. Under the stated input-ensemble assumptions, for every $0<\varepsilon\leq1$, there exists a classical algorithm that outputs an estimate $\widehat f$ of $\langle \mathcal{C}(O)\rangle_\rho$ satisfying
\begin{equation}
 \left(
 \mathbb E_{\rho\sim\mathsf E}
 \left|\widehat f- \langle \mathcal{C}(O)\rangle_\rho\right|^2
 \right)^{1/2}
 \leq\varepsilon
 \label{eq:main-gate-rms}
\end{equation}
in time
\begin{equation}
T_C \in \left(\frac{n}{\varepsilon}\right)^{
 \mathcal O(1+\delta/\gamma)}.
 \label{eq:main-gate-runtime}
\end{equation}
For bounded $\delta/\gamma$, this is polynomial in $n$ and $\varepsilon^{-1}$.
\end{theorem}

In a product-formula circuit intended to approximate noiseless Hamiltonian dynamics, smaller time steps give smaller rotation angles. If the device is noisy with a fixed noise per gate, this reduces $\delta/\gamma$ and improves the noise-dependent exponent in our classical simulation bound for the resulting noisy circuit.

\section{Methods}
\label{sec:methods}

The two algorithms in Section~\ref{sec:results} use different simulation procedures, but their guarantees rest on a common bound on the evolving observable. We first explain how the Pauli-$2$ error is a relevant quantity. We then show how local dissipation controls both the number and the weight of appreciable Pauli terms. Finally, we use that control to bound the two algorithms and explain the additional gain from Richardson extrapolation.

\subsection{Why Pauli-$2$ error matters}
The normalized Hilbert--Schmidt norm measures error directly in the Pauli coefficients. For $A=\sum_{P\in\cP_n}c_P P$, we define
\begin{equation}
    \norm{A}_2^2
    := \frac{\Tr[A^\dagger A]}{2^n}
    = \sum_{P\in\cP_n}\abs{c_P}^2.
    \label{eq:methods-pauli-two}
\end{equation}
Several operational interpretations of this norm are detailed in Appendix~\ref{supp:operational}. In particular, it controls the root-mean-square error in expectation values over any state $1$-design (i.e. any ensemble with maximally mixed average). To see this, let $X:=O(t)-\widetilde O$ be the difference between the target observable and its approximation. For any state $1$-design $\mathsf E$,
\begin{equation}
    \mathbb E_{\rho\sim\mathsf E}
    \abs{\Tr(\rho X)}^2
    \leq \norm{X}_2^2.
\end{equation}
Thus, a Pauli-$2$ error bound of $\norm{X}_2\leq\varepsilon$ implies a rootmean square error of at most $\varepsilon$.

For classical Pauli propagation, Pauli-$2$ error is precisely the mass discarded during truncation. After combining contributions to the same Pauli string, we discard coefficients smaller than a threshold $\tau>0$. Writing $A=\sum_{P\in\cP_n}c_PP$, the truncated observable is
\[
    \mathcal T_\tau(A)
    \coloneqq
    \sum_{\abs{c_P}\geq\tau}c_PP.
\]
The resulting Pauli-$2$ error is
\begin{equation}
    \norm{A-\mathcal T_\tau(A)}_2^2
    =\sum_{\abs{c_P}<\tau}\abs{c_P}^2.
    \label{eq:methods-coefficient-error}
\end{equation}
The difficulty is to keep this error small after many layers while storing only a manageable number of strings.

For quantum simulation, Pauli weight controls the error of the symmetric product formula in Section~\ref{sec:product-formula}. Its leading error contains nested commutators of local generators. A local term can contribute only if it meets the current Pauli string or a region reached by an earlier commutator. At each fixed order, the number of relevant terms is governed by the string's weight and the local degree $d$, rather than directly by the total number of qubits. However, an ordinary Pauli-$2$ bound does not rule out an observable whose coefficients are concentrated on very high-weight strings. We need a bound that also controls this weight.

\subsection{Dissipation-induced compression}

Two stronger coefficient norms capture the forms of spreading relevant to the algorithms. For $1\leq p\leq2$ and $\theta\geq0$, define
\begin{equation}
    \norm{A}_p^p
    \coloneqq\sum_{P\in\cP_n}\abs{c_P}^p,
    \qquad
    \norm{A}_{2,\theta}^2
    \coloneqq\sum_{P\in\cP_n}
    \e^{2\theta\abs{P}}\abs{c_P}^2.
    \label{eq:methods-suppression-norms}
\end{equation}
Here $\norm{\cdot}_p$ is a norm of the Pauli coefficients, not the Schatten norm, and correspondingly we will refer to it as the Pauli-$p$ norm. A bound with $p<2$ limits how many coefficients can be appreciable. A bound with $\theta>0$ limits the total Pauli-$2$ mass at high weight.

Our key result is that both stronger norms are nonincreasing at all times for suitable values of $p$ and $\theta$, determined by the balance between dissipation and unitary scrambling.

\begin{theorem}[Dissipation-induced compression]
\label{thm:methods-suppression}
Under the local Lindbladian and damping assumptions of Section~\ref{sec:framework}, there are parameters $1\leq p_\star<2$ and $\theta_\star>0$, depending only on $J$, $\gamma$, and $k$, such that for every observable $A$ and $t\geq0$,
\begin{align}
    \norm{\e^{t\cL}(A)}_p
    &\leq\norm{A}_p,
    &&p_\star\leq p\leq2,
    \label{eq:methods-p-contraction}\\
    \norm{\e^{t\cL}(A)}_{2,\theta}
    &\leq\norm{A}_{2,\theta},
    &&0\leq\theta\leq\theta_\star.
    \label{eq:methods-weight-contraction}
\end{align}
Both inequalities also hold with $\e^{t\cL}$ replaced by the symmetric product-formula step $\mathsf{S}_h $, and $\Phi_h$, for every $h\geq0$.
\end{theorem}
\begin{proof}[Proof sketch]
Fix a local block $a$ and consider a discretized step $\e^{h\mathcal H_a}\e^{h\cD_a}$. The dissipative term $\e^{h\cD_a}$ leaves Pauli operators outside $S_a$ unchanged and damps every Pauli operator that meets $S_a$ at rate at least $\gamma$. The coherent term $\e^{h\mathcal H_a}$ couples pairs of Pauli strings through two-dimensional rotations. Each rotation changes the weight by at most $\abs{P_a}\leq k$.

In the weighted Pauli-$2$ norm, a rotation can increase the norm at rate at most $2J\sinh(k\theta)$. The local damping offsets this increase whenever $2J\sinh(k\theta)\leq\gamma$. Thus we may take
\begin{equation}
    \theta_\star
    \coloneqq\frac{1}{k}
    \operatorname{arsinh}\!\left(\frac{\gamma}{2J}\right).
    \label{eq:methods-theta-star}
\end{equation}
For the unweighted Pauli-$p$ norm, the unitary rotation $\e^{h\mathcal H_a}$ increases the Pauli-$1$ norm at a rate at most $2J$, while preserving the Pauli-$2$ norm. The dissipative term $\e^{h\cD_a}$ damps the coefficients mixed by the rotation at a rate at least $\gamma$. Thus, we can interpolate between $p=1$ and $p=2$ to determine a coefficient $p_\star$ such that Pauli-$p_\star$ norm is nonincreasing. In particular, the Riesz-Thorin theorem gives
\begin{equation}
    p_\star
    \coloneqq
    \frac{2}{1+\min\{1,\gamma/(2J)\}}.
    \label{eq:methods-p-star}
\end{equation}

These rate balances also apply to the symmetric local map $\e^{h\cD_a/2}\e^{h\mathcal H_a}\e^{h\cD_a/2}$: its two dissipative half steps supply the full damping over time $h$. Maps on disjoint blocks tensorize within a color layer, and composing the color layers preserves the contraction. This proves the claim for $\mathsf{S}_h$. Taking the product-formula limit gives the claim for $\e^{t\cL}$.
\end{proof}

The theorem gives two complementary forms of compression. For an initially local observable, the Pauli-$p$ bound limits the number of appreciable coefficients, while the weighted Pauli-$2$ bound makes high-weight contributions small. Both bounds hold for exact Lindblad evolution and after any number of symmetric product-formula steps, independently of the step size. We use weight suppression to control quantum Trotter error, then combine both bounds to justify sparse classical Pauli propagation.

\subsection{Consequences for quantum simulation}
\label{sec:methods-quantum-sim}
For quantum simulation, weight suppression allows us to choose a Trotter step size independent of system size at fixed target accuracy. For a step of duration $h$, define the error on an observable $A$ by
\begin{equation}
    \mathcal E_h(A)
    \coloneqq
    \bigl(\mathsf{S}_h-\e^{h\cL}\bigr)(A),
    \label{eq:methods-one-step-product-error}
\end{equation}
where $\mathsf{S}_h$ is defined in Eq.\ \eqref{eq:main-symmetric-product-formula}. The symmetric ordering cancels the error at order $h^2$, thus the leading term has the form 
\begin{equation}
    \mathcal E_h(A)
    =
    h^3\sum_{\mu,\nu,\lambda}
    \kappa_{\mu\nu\lambda}
    [\mathcal G_\mu,[\mathcal G_\nu,\mathcal G_\lambda]](A)
    +O_A(h^4).
    \label{eq:methods-nested-commutator-error}
\end{equation}
Here each $\mathcal G_\mu$ is a local coherent or dissipative generator, $\mathcal H_a$ or $\cD_a$. The brackets are commutators of generators, and the numerical coefficients $\kappa_{\mu\nu\lambda}$ depend on the symmetric ordering.

A nested commutator can act on a Pauli string $P$ only when its local blocks form an overlapping chain that reaches $\supp(P)$. At any fixed order, the number of relevant blocks is therefore governed by the weight $\abs{P}$ and the local degree $d$, rather than directly by $n$. However, as the observable evolves, it develops components of higher Pauli weight. Writing $A=\sum_P c_P P$, the weighted Pauli-$2$ norm controls their total squared coefficient mass:
\begin{equation}
    \sum_{P:\,|P|\geq w}\abs{c_P}^2
    \leq \e^{-2\theta w}\norm{A}_{2,\theta}^2.
    \label{eq:methods-weight-tail}
\end{equation}
Together with the uniform bound on $\norm{A}_{2,\theta}$, this exponential tail controls the weight moments entering the commutator estimates independently of system size.

The finite-step commutator bounds then control the one-step error, including the remainder beyond the leading term in \cref{eq:methods-nested-commutator-error}. Summing the errors over $N=t/h$ steps gives, for sufficiently small $h$ and a suitable $0<\theta\leq\theta_\star$,
\begin{equation}
    \norm{
        \bigl(\mathsf{S}_h^N-\e^{t\cL}\bigr)(O)
    }_2
    \leq
    C_{\mathrm{loc},\theta}\,t h^2\norm{O}_{2,\theta},
    \label{eq:methods-quantum-product-error}
\end{equation}
where $C_{\mathrm{loc},\theta}$ is a constant that depends on local parameters but not on $n$. For an observable initially supported on at most $w_0$ qubits, $\norm{O}_{2,\theta}\leq\e^{\theta w_0}\norm{O}_2$. Thus the step size needed to control Trotter error depends on local parameters and $w_0$, not on the system size.

\subsection{Consequences for classical simulation}

After each symmetric color update, the classical algorithm uses Theorem~\ref{thm:methods-suppression}'s Pauli-$p_\star$ bound to limit the number of retained strings and its weighted Pauli-$2$ bound to limit their weight.

Recall that $\mathcal T_\tau(A)$ discards coefficients smaller than $\tau$. Let $\Pi_{\leq w}$ discard strings of weight greater than $w$, and let $N_\tau(A)$ count the coefficients retained by $\mathcal T_\tau$. For $1\leq p<2$ and $\theta>0$,
\begin{align}
    \norm{A-\mathcal T_\tau(A)}_2
    &\leq \tau^{1-p/2}\norm{A}_p^{p/2},
    \label{eq:methods-coefficient-cutoff}\\
    N_\tau(A)
    &\leq \tau^{-p}\norm{A}_p^p,\\
    \norm{A-\Pi_{\leq w}(A)}_2
    &\leq \e^{-\theta(w+1)}\norm{A}_{2,\theta}.
    \label{eq:methods-weight-cutoff}
\end{align}
The first and third bounds control the errors caused by the two cutoffs. The middle bound controls how many strings remain.

The bounds in Eqs.~(\ref{eq:methods-coefficient-cutoff}-\ref{eq:methods-weight-cutoff}) continue to apply as the algorithm repeatedly truncates. Each symmetric color update cannot increase the Pauli-$p_\star$ or weighted Pauli-$2$ norm (from Theorem~\ref{thm:methods-suppression}), and discarding coefficients cannot increase either norm. Thus the expansion before every truncation has both norms bounded by their initial values on $O$. If the circuit contains $L$ color updates, its accumulated truncation error is therefore at most
\begin{equation}
    E_{\mathrm{tr}}
    \leq L\left[
        \tau^{1-p_\star/2}\norm{O}_{p_\star}^{p_\star/2}
        +\e^{-\theta_\star(w+1)}
         \norm{O}_{2,\theta_\star}
    \right].
    \label{eq:methods-accumulated-truncation}
\end{equation}

To make $E_{\mathrm{tr}}\leq\eta$, we give each cutoff an error budget of $\eta/(2L)$ at every update. The coefficient-truncation threshold can be chosen as
\begin{equation}
    \tau=
    \left(
        \frac{\eta}
             {2L\norm{O}_{p_\star}^{p_\star/2}}
    \right)^{2/(2-p_\star)} \, 
    \label{eq:methods-classical-coefficient-threshold}
\end{equation}
and weight-truncation level $w$ as
\begin{equation}
    w+1\geq
    \frac{1}{\theta_\star}
    \log\!\left(
        \frac{2L\norm{O}_{2,\theta_\star}}{\eta}
    \right).
    \label{eq:methods-classical-weight-cutoff}
\end{equation}
For a normalized $O$ supported on at most $w_0$ qubits, the initial norms do not depend on $n$. The required weight is then $w=O(w_0+\theta_\star^{-1}\log(\e L/\eta))$. Substituting the chosen $\tau$ into the counting bound shows that every retained expansion contains at most
\begin{equation}
    N_\tau
    \leq
    \left(
        \frac{2L\norm{O}_{p_\star}}{\eta}
    \right)^{\alpha_\star},
    \label{eq:methods-retained-count}
\end{equation}
where
\begin{equation}
    \alpha_\star
    \coloneqq\frac{2p_\star}{2-p_\star}
    =\max\!\left\{2,\frac{4J}{\gamma}\right\}.
    \label{eq:methods-sparsity-exponent}
\end{equation}
describes how the retained-string count grows as the target error decreases. This bound becomes weaker as damping decreases.

The cost of updating each string also depends on its weight. A string of weight at most $w$ meets at most $w$ disjoint blocks in one color layer. Each coherent rotation can produce two Pauli terms, while the dissipative half steps only rescale them, so one string produces at most $2^w$ branches. The chosen weight cutoff makes this factor polynomial in $L/\eta$. The resulting propagation cost is
\begin{equation}
     \widetilde{\mathcal O}\!\left(
        L^{1+\beta_\star}\eta^{-\beta_\star}
    \right),
    \qquad
    \beta_\star
    =\alpha_\star+\frac{\log 2}{\theta_\star}.
    \label{eq:methods-classical-cost}
\end{equation}
We recall that $\widetilde{\mathcal{O}}$ hides the logarithmic factors. Here we suppress prefactors depending on the local parameters and $w_0$. The additional term in $\beta_\star$ is the cost of branching, which is why controlling the number of strings alone would not suffice.

\subsection{Richardson extrapolation}
The weight suppression argument in Section~\ref{sec:methods-quantum-sim} already permits a Trotter step size independent of system size at fixed accuracy $\varepsilon$. The resulting depth bound, however, scales as $\mathcal O(\varepsilon^{-1})$ for a first-order trotter step(Appendix~\ref{supp:quantum}) and therefore becomes polynomial in $n$ when inverse-polynomial accuracy is required.

Richardson extrapolation reduces the maximum circuit depth to $\mathcal O(\log^2(\e/\varepsilon))$ at fixed time and local parameters, using interpolation techniques developed for Hamiltonian~\cite{low2019well, rendon2024improved, watson2024randomly, watson2025exponentially}, and Lindbladian simulation~\cite{wang2026lindbladian}. It combines simulations at different step sizes to cancel leading Trotter errors. For our symmetric formula, the fixed-time error expands in even powers of $h=t/N$. For any finite order $q$ and sufficiently small $h$,
\begin{equation}
    \mathsf{S}_{t/N}^{N}
    =\e^{t\cL}
    +\sum_{\ell=1}^{q-1}N^{-2\ell}E_\ell(t)
    +R_{q,N}(t).
    \label{eq:methods-symmetric-expansion}
\end{equation}
This equality is an equality between \textit{maps} acting on observables.

\begin{equation}
    E_\ell(t)
    \coloneqq
    \lim_{N\to\infty}N^{2\ell}
    \left[
        \mathsf{S}_{t/N}^{N}-\e^{t\cL}
        -\sum_{j=1}^{\ell-1}N^{-2j}E_j(t)
    \right] ,
    \label{eq:methods-richardson-corrections}
\end{equation}
where the sum is empty for $\ell=1$. At finite order $q$, the remainder map is
\begin{equation}
    R_{q,N}(t)
    \coloneqq
    \mathsf{S}_{t/N}^{N}-\e^{t\cL}
    -\sum_{\ell=1}^{q-1}N^{-2\ell}E_\ell(t).
    \label{eq:methods-richardson-remainder}
\end{equation}
For fixed $q$, its controlled norm decreases as $N^{-2q}$ under the small-step condition, with a bound independent of the number of qubits.

On quantum hardware, we run $q$ different circuits. Circuit $j$ uses $N_j$ symmetric steps of duration $t/N_j$, so every circuit approximates evolution to the same time $t$. We combine their expectation values in classical postprocessing as
\begin{equation}
    f_{\mathrm{Rich}}(\rho)
    \coloneqq
    \sum_{j=1}^{q}\omega_j
    \Tr\!\left[\rho\,\mathsf{S}_{t/N_j}^{N_j}(O)\right].
    \label{eq:methods-richardson-expectation}
\end{equation}
We choose the step counts and weights to cancel the leading terms in \cref{eq:methods-symmetric-expansion}. The condition $\sum_{j=1}^{q}\omega_j=1$ retains the exact evolution, while
\begin{equation}
    \sum_{j=1}^{q}\omega_jN_j^{-2\ell}=0
    \qquad (1\leq\ell<q)
    \label{eq:methods-richardson-conditions}
\end{equation}
cancels the first $q-1$ error corrections. With $q=O(\log(1/\varepsilon))$, the deepest circuit needs at most $O(\log^2(1/\varepsilon))$ steps at fixed $t$, $w_0$, and local parameters. The signed weights can amplify measurement and implementation errors, which are separate from the ideal circuit-depth bound.

The classical algorithm applies the same signed combination to its truncated Pauli-propagation estimates. The chosen weights satisfy $\sum_j\abs{\omega_j}=O(q^3)$, so it suffices to make each circuit's truncation error $\eta=\Theta(\varepsilon/q^3)$. There are $q$ circuits with at most $L=O(q^2)$ color updates each, up to fixed physical-parameter factors. Substituting these choices into \cref{eq:methods-classical-cost} gives the precision factor
\begin{equation}
    q\,L^{1+\beta_\star}\eta^{-\beta_\star}
    =O\!\left[
        \varepsilon^{-\beta_\star}
        \log^{3+5\beta_\star}\!\left(\frac{\e}{\varepsilon}\right)
    \right],
    \label{eq:methods-classical-richardson-cost}
\end{equation}
as stated in \cref{thm:main-classical-simulation}. Thus, extrapolation improves the precision dependence of the existing propagation procedure.

\section{Numerical Implementations}

\begin{figure}
    \centering
    \includegraphics[width=\linewidth]{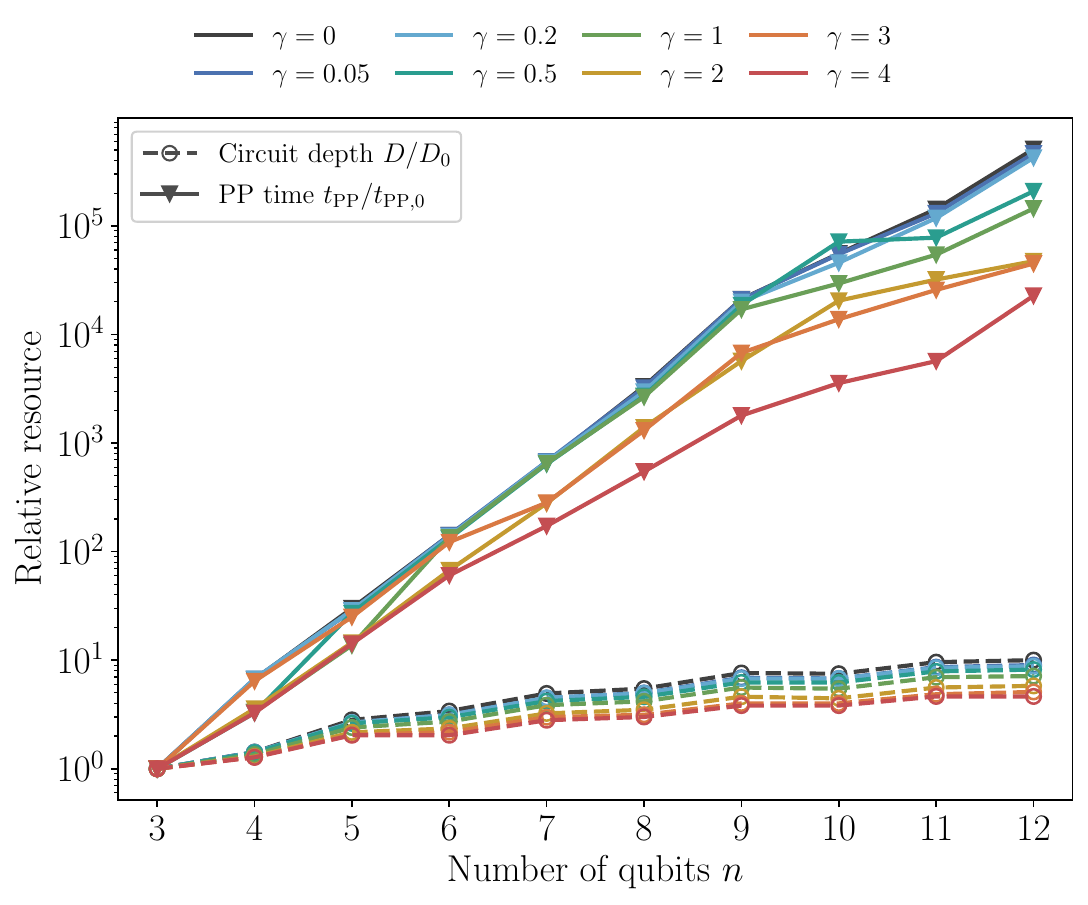}
    \caption{\textit{Relative resource growth for all-to-all Heisenberg dynamics.}    At final evolution time $T=0.3$ and normalized Hilbert--Schmidt    error target $\varepsilon=0.01$, dashed lines with open circles show    circuit depth $D/D_0 :=D(n,\gamma)/D(3,\gamma)$, and solid lines with filled    downward triangles show PP propagation time    $t_{\mathrm{PP}}/t_{\mathrm{PP},0}:=t_{\mathrm{PP}}(n,\gamma)/t_{\mathrm{PP}}(3,\gamma)$.    Colors identify the depolarization rate $\gamma$; each resource is    normalized to its own three-qubit value at the same noise rate.    Depth uses the nearest-integer interpolated Trotter-step estimate    from untruncated simulations. PP time is the fastest measured time    among sampled step counts and cutoffs meeting the total-error target.    Lines guide the eye; the ratios compare relative growth, not absolute    quantum and classical runtimes.}
    \label{fig:relative_resource}
\end{figure}

The gap between our asymptotic scaling analysis leaves room for a quantum advantage that increases as dissipation weakens. However, comparing upper bounds is clearly inconclusive. We therefore proceed to numerically explore the contrasting quantum and classical resource growth suggested by our theoretical analysis. In our long-range Heisenberg numerical benchmark~\cite{Fr_rot_2017}, classical propagation time grows substantially faster with system size than ideal quantum circuit depth. This benchmark uses a collective observable and a more physically relevant fixed input state~\cite{heyl2019quantum}, hence extends beyond the average case prediction in our theorems.

We compare the resources of an ideal quantum product-formula circuit and a classical Pauli-propagation (PP) implementation for systems of $n=3,\ldots,12$ qubits. At these sizes, direct integration of the continuous dynamics using \texttt{QuantumToolbox.jl} provides a numerical reference for the Heisenberg-evolved magnetization operator~\cite{mercurio2025quantumtoolbox}. We simulate the product-formula circuits using \texttt{PauliPropagation.jl}~\cite{rudolph2025pauli}, both without coefficient truncation, to isolate Trotter error, and with truncation, to measure the classical cost at a fixed total error. This benchmark uses a first-order product formula without Richardson extrapolation. In this section, $t$ denotes elapsed evolution time and $T$ the final time of a numerical run.

We consider the anisotropic Heisenberg Hamiltonian~\cite{Fr_rot_2017, muleadyValidatingPhasespaceMethods2023, perlinSpinSqueezingShortRange2020}:
\begin{equation}
    H=-J\sum_{i<j}\frac{X_iX_j+Y_iY_j+(1+J_z)Z_iZ_j}{d_{ij}^{3}},
    \label{eq:numerical-heisenberg-model}
\end{equation}
with $J=1$ and $J_z=-1.8$. The qubits reside on compact patches of the square lattice with open boundaries and unit lattice spacing; $d_{ij}$ denotes their Euclidean separation. Every pair interacts, with no system-size-dependent rescaling of the couplings. Independent single-qubit depolarization acts at rate $\gamma$, with Heisenberg generator $\cD_\gamma(P)=-\gamma\abs{P}P$. We propagate the magnetization operator,
\begin{equation}
    M_x(0)=\frac{1}{\sqrt n}\sum_{i=1}^n X_i.
\end{equation}
For the size scan in \cref{fig:relative_resource}, we evolve $M_x(t)$ to the fixed final time $T=0.3$, for noise strengths ranging from $\gamma=0$ to $\gamma=4$. For an approximation $\widetilde M_{x;N,\tau}(T)$ using $N$ Trotter steps and Pauli coefficient cutoff $\tau$, we measure the total operator error
\begin{equation}
    \begin{aligned}
        \varepsilon_{\mathrm{tot}}(N,\tau)
        &=\norm{\widetilde M_{x;N,\tau}(T)-M_x(T)}_2.
    \end{aligned}
    \label{eq:numerical-total-error}
\end{equation}

The quantum resource underlying \cref{fig:relative_resource} is the circuit depth estimated from untruncated simulations ($\tau=0$) at the target Trotter error. We count coherent layers of disjoint two-qubit gates, assuming all-to-all hardware connectivity, and for $\gamma>0$, one parallel layer of local depolarizing channels per Trotter step. For the classical calculation, PP discards Pauli coefficients with magnitude below $\tau$ after individual gates. At each $(n,\gamma)$, we scan both $N$ and $\tau$ and report the fastest measured PP wall-clock propagation time $t_{\mathrm{PP}}$ among runs satisfying $\varepsilon_{\mathrm{tot}}(N,\tau)\leq\varepsilon$ where we take $\varepsilon=0.01$. The comparison is directly against the continuous reference, so the criterion includes both Trotter and truncation errors. The classical optimum within the sampled grid can use a different step count from the quantum circuit estimate because it is chosen to jointly minimize truncation and Trotter errors. The implementation and interpolation details for the resource comparison are given in Supplementary Information~\ref{supp:numerics}.

Figure~\ref{fig:relative_resource} shows a substantially faster relative growth of classical propagation time $t_{\mathrm{PP}}(n,\gamma)$ than of circuit depth $D(n,\gamma)$ over the sampled sizes. To compare their growth despite their different units, we plot $D(n,\gamma)/D(3,\gamma)$ and $t_{\mathrm{PP}}(n,\gamma)/t_{\mathrm{PP}}(3,\gamma)$. From $n=3$ to $n=12$, the depth grows by approximately $5-10$ times across the sampled noise strengths, whereas PP time grows by roughly $10^4-10^5$ times. Increasing dissipation reduces the observed relative growth of PP time, consistent with more efficient classical simulation through suppression of Pauli coefficients. The connecting lines are guides to the eye.

Over $n=3,\ldots,12$, power-law fits to the data in \cref{fig:relative_resource} give $D\propto n^{1.15\text{--}1.73}$ and $t_{\mathrm{PP}}\propto n^{7.31\text{--}9.70}$ across the plotted noise strengths (Supplementary Information~\ref{supp:numerics}), hinting towards a substantial polynomial speed-up on quantum hardware. Although additional quantum hardware runtime considerations (in particular measurement repetitions and circuit reset times) would likely significantly reduce any gap in practice, we are cautiously optimistic that for sufficiently large $n$ that this seeming polynomial gap could leave room for advantage on near-term hardware. Architecture-specific resource estimates are needed to determine whether this potential speedup persists once fault-tolerant overheads are included.

Finally, we extend the PP calculation to a $7\times7$ lattice of 49 qubits, using the same Hamiltonian, open boundaries, and depolarizing noise as above. For the initial product state $\ket{+}^{\otimes n}$, we monitor the magnetization per spin, \(    m_x(t)=\langle M_x(t)\rangle/\sqrt n \). Figure~\ref{fig:heisenberg-7x7-dynamics} compares four noise strengths and three coefficient cutoffs, using the fixed step size $\delta=3/320$. We show the saved trajectories in $0\leq t\leq0.9$. The same normalization of $M_x$ and the same coefficient-truncation convention are used as in the size scan. At higher noise strength, the magnetization converges faster with a higher truncation threshold, whereas with lower noise strength, the predictions converge more slowly with smaller truncation $\tau$. This is consistent with our expectation that noise enables classical simulation to converge faster. 

\begin{figure}
    \centering
    \includegraphics[width=\linewidth]{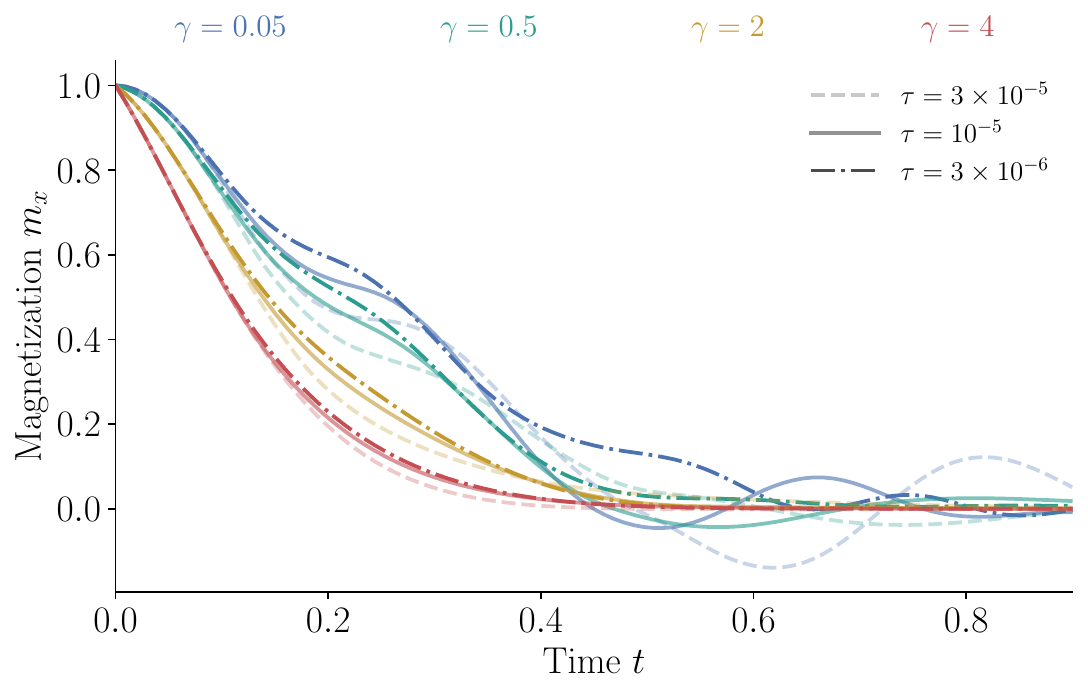}
    \caption{\textit{Magnetization $m_x(t)$ on a $7\times7$ lattice.}    Colors indicate $\gamma=0.05,0.5,2,4$ for the Hamiltonian in    \cref{eq:numerical-heisenberg-model}. Dashed, solid, and dash-dotted curves use coefficient cutoffs $\tau=3\times10^{-5}$, $10^{-5}$, and $3\times10^{-6}$. Coarser cutoffs are drawn more faintly for visualization.}
    \label{fig:heisenberg-7x7-dynamics}
\end{figure}

\section{Discussion}

We have shown that dissipation can make open-system dynamics easier to simulate on both quantum and classical hardware. By suppressing high-weight Pauli components of evolving observables, it permits larger Trotter steps when estimating local observables. This in turn yields circuit-depth bounds independent of system size for typical input states. The same compression makes Pauli propagation polynomial in system size and inverse accuracy for fixed local parameters and any positive damping rate. Our upper bounds leave room for a substantial quantum advantage as dissipation weakens, a possibility also suggested by our numerics. Determining whether this separation translates into a practical advantage will require end-to-end resource analyses for specific hardware architectures, which we leave for future work.

The main limitation of our results is their restriction to a particular class of dissipation. Namely, the local dissipator must damp every nonidentity Pauli direction on the region of its paired interaction. Although depolarization satisfies these conditions, many physically relevant environments do not. For example, $Z$-dephasing leaves $Z$ unchanged, correlated reservoirs can protect collective modes, and amplitude damping describes energy relaxation outside the Pauli-diagonal setting. Since these mechanisms shape relaxation, excitation transport, and reservoir-driven steady states, understanding whether they also compress evolving observables would substantially broaden the physical scope of our results.

A reason to expect such compression is provided by Schuster and Yao~\cite{schuster2023operator}, who conjectured that, in ergodic systems with weak local noise, operator growth can be described effectively by replacing the physical dissipation by local depolarization. This suggests that the mechanism underlying our results may persist beyond the noise models covered by our proofs. Turning this effective description into rigorous simulation guarantees for more general dissipation is an important next step.

Fermionic dynamics offer another interesting extension. Majorana propagation represents observables as sums of Majorana monomials and truncates them by length~\cite{miller2025simulation,d2025majorana}. 
In work that appeared shortly after the first version of this manuscript, Facelli et al.~\cite{facelli2026fastconvergence} also establish conditions under which dissipation suppresses high-degree components and controls the convergence of classical propagation algorithms in spin and fermionic systems. A natural next step is to extend our system-size-independent quantum depth bounds and polynomial-time classical guarantees to fermionic dynamics.
Beyond fermions, one could study time-dependent random Hamiltonians (e.g. Brownian dynamics~\cite{lashkari2013towards, bentsen2024complexity}) and continuous-variable systems~\cite{guseynov2026coherent,upreti2025quantum}. Compression may also aid the inverse task of learning Lindblad generators from data~\cite{romanov2026learning,ivashkov2026ansatz, lewis2026learning,arad2026near}.

\begin{acknowledgments}
AA thanks Daniel Stilck França for useful discussions and for suggesting the use of $p$-norms to quantify sparsity in the Pauli basis. AA and ZH thank Kushal Seetharam and Alexander Schuckert for interesting discussions. AA and ZH thank Nathaniel Leitao for helpful discussions and for pointing out the connection between our work and Schuster and Yao's operator growth hypothesis. AA and ZH acknowledge support from the Sandoz Family Foundation--Monique de Meuron program for Academic Promotion.
AA acknowledges support from the European Union’s Horizon Europe research and innovation programme under the Marie Skłodowska-Curie Actions (PROQSY, Grant No.~101282902).
RP acknowledges support from the SNF Quantum Flagship Replacement Scheme (Grant No.~215933). YT acknowledges support from NCCR Spin funded by the SNF (Grant No.~565785).
\end{acknowledgments}

\section*{AI Usage Statement}
The authors thank Claude and GPT who worked like diligent and unusually responsive students. The human authors supplied the core ideas, intuitions, methods, and proof strategies. Through much back-and-forth, Claude and GPT then fleshed out the details, ran the numerics, and drafted parts of the manuscript. As with student supervision, we sometimes wondered whether this saved time, but hope the training was useful. The human authors carefully checked and refined the final manuscript.

\bibliographystyle{unsrt}
\bibliography{quantum,temporal}

\begin{thebibliography}{10}

\bibitem{redfield1957theory}
Alfred~G Redfield.
\newblock On the theory of relaxation processes.
\newblock {\em IBM Journal of Research and Development}, 1(1):19--31, 1957.

\bibitem{chen2018ab}
Po-chia Chen, Maggy Hologne, Olivier Walker, and Janosch Hennig.
\newblock Ab initio prediction of nmr spin relaxation parameters from molecular dynamics simulations.
\newblock {\em Journal of Chemical Theory and Computation}, 14(2):1009--1019, 2018.

\bibitem{seetharamDigitalQuantumSimulation2023}
Kushal Seetharam, Debopriyo Biswas, Crystal Noel, Andrew Risinger, Daiwei Zhu, Or~Katz, Sambuddha Chattopadhyay, Marko Cetina, Christopher Monroe, Eugene Demler, and Dries Sels.
\newblock Digital quantum simulation of {{NMR}} experiments.
\newblock {\em Science Advances}, 9(46):eadh2594, November 2023.

\bibitem{mohseni2008environment}
Masoud Mohseni, Patrick Rebentrost, Seth Lloyd, and Alan Aspuru-Guzik.
\newblock Environment-assisted quantum walks in photosynthetic energy transfer.
\newblock {\em The Journal of chemical physics}, 129(17), 2008.

\bibitem{ishizaki2009theoretical}
Akihito Ishizaki and Graham~R Fleming.
\newblock Theoretical examination of quantum coherence in a photosynthetic system at physiological temperature.
\newblock {\em Proceedings of the National Academy of Sciences}, 106(41):17255--17260, 2009.

\bibitem{prosen2011exact}
Toma{\v{z}} Prosen.
\newblock Exact nonequilibrium steady state of a strongly driven open xxz chain.
\newblock {\em Physical review letters}, 107(13):137201, 2011.

\bibitem{joos1985emergence}
Eric Joos and H~Dieter Zeh.
\newblock The emergence of classical properties through interaction with the environment.
\newblock {\em Zeitschrift f{\"u}r Physik B Condensed Matter}, 59(2):223--243, 1985.

\bibitem{habib1998decoherence}
Salman Habib, Kosuke Shizume, and Wojciech~Hubert Zurek.
\newblock Decoherence, chaos, and the correspondence principle.
\newblock {\em Physical review letters}, 80(20):4361, 1998.

\bibitem{smirne2016ultimate}
Andrea Smirne, Jan Ko{\l}ody{\'n}ski, Susana~F Huelga, and Rafa{\l} Demkowicz-Dobrza{\'n}ski.
\newblock Ultimate precision limits for noisy frequency estimation.
\newblock {\em Physical review letters}, 116(12):120801, 2016.

\bibitem{puig2024dynamical}
Ricard Puig, Pavel Sekatski, Paolo~Andrea Erdman, Paolo Abiuso, John Calsamiglia, and Mart{\'\i} Perarnau-Llobet.
\newblock From dynamical to steady-state many-body metrology: Precision limits and their attainability with two-body interactions.
\newblock {\em PRX Quantum}, 6(3):030309, 2025.

\bibitem{childs2021theory}
Andrew~M Childs, Yuan Su, Minh~C Tran, Nathan Wiebe, and Shuchen Zhu.
\newblock Theory of trotter error with commutator scaling.
\newblock {\em Physical Review X}, 11(1):011020, 2021.

\bibitem{donvil2022quantum}
Brecht Donvil and Paolo Muratore-Ginanneschi.
\newblock Quantum trajectory framework for general time-local master equations.
\newblock {\em Nature communications}, 13(1):4140, 2022.

\bibitem{sander2025large}
Aaron Sander, Maximilian Fr{\"o}hlich, Martin Eigel, Jens Eisert, Patrick Gel{\ss}, Michael Hinterm{\"u}ller, Richard~M Milbradt, Robert Wille, and Christian~B Mendl.
\newblock Large-scale stochastic simulation of open quantum systems.
\newblock {\em Nature Communications}, 16(1):11074, 2025.

\bibitem{pocklington2025efficient}
Andrew Pocklington and Aashish~A Clerk.
\newblock Efficient simulation of nontrivial dissipative spin chains via stochastic unraveling.
\newblock {\em PRX Quantum}, 6(3):030349, 2025.

\bibitem{cao2026dynamically}
Yu~Cao, Mingfeng He, and Xiantao Li.
\newblock Dynamically optimal unraveling schemes for simulating lindblad equations.
\newblock {\em Journal of Physics A: Mathematical and Theoretical}, 59(16):165301, 2026.

\bibitem{chen2024optimized}
Zhuo Chen, Yimu Bao, and Soonwon Choi.
\newblock Optimized trajectory unraveling for classical simulation of noisy quantum dynamics.
\newblock {\em Physical Review Letters}, 133(23):230403, 2024.

\bibitem{dalibard1992wave}
Jean Dalibard, Yvan Castin, and Klaus M\o{}lmer.
\newblock Wave-function approach to dissipative processes in quantum optics.
\newblock {\em Phys. Rev. Lett.}, 68:580--583, Feb 1992.

\bibitem{zwolak2004mixed}
Michael Zwolak and Guifr{\'e} Vidal.
\newblock Mixed-state dynamics in one-dimensional quantum lattice systems: A time-dependent superoperator renormalization algorithm.
\newblock {\em Physical review letters}, 93(20):207205, 2004.

\bibitem{werner2016positive}
A.~H. Werner, D.~Jaschke, P.~Silvi, M.~Kliesch, T.~Calarco, J.~Eisert, and S.~Montangero.
\newblock Positive tensor network approach for simulating open quantum many-body systems.
\newblock {\em Phys. Rev. Lett.}, 116:237201, Jun 2016.

\bibitem{cichy2026classical}
Simon Cichy, Paul~K Faehrmann, Lennart Bittel, Jens Eisert, and Hakop Pashayan.
\newblock Classical simulation of noisy quantum circuits via locally entanglement-optimal unravelings.
\newblock {\em Physical Review A}, 114(2):022430, 2026.

\bibitem{wang2026lindbladian}
Xinzhao Wang, Shuo Zhou, Xiaoyang Wang, Yi-Cong Zheng, Shengyu Zhang, and Tongyang Li.
\newblock Lindbladian simulation with commutator bounds.
\newblock {\em arXiv preprint arXiv:2603.28602}, 2026.

\bibitem{barthel2012quasilocality}
Thomas Barthel and Martin Kliesch.
\newblock Quasilocality and efficient simulation of markovian quantum dynamics.
\newblock {\em Physical review letters}, 108(23):230504, 2012.

\bibitem{wild2023classical}
Dominik~S Wild and {\'A}lvaro~M Alhambra.
\newblock Classical simulation of short-time quantum dynamics.
\newblock {\em PRX Quantum}, 4(2):020340, 2023.

\bibitem{rudolph2025pauli}
Manuel~S Rudolph, Tyson Jones, Yanting Teng, Armando Angrisani, and Zo{\"e} Holmes.
\newblock Pauli propagation: A computational framework for simulating quantum systems.
\newblock {\em arXiv preprint arXiv:2505.21606}, 2025.

\bibitem{beguvsic2023fast}
Tomislav Begu{\v{s}}i{\'c}, Johnnie Gray, and Garnet Kin-Lic Chan.
\newblock Fast and converged classical simulations of evidence for the utility of quantum computing before fault tolerance.
\newblock {\em Science Advances}, 10(3), 2024.

\bibitem{gao2018efficient}
Xun Gao and Luming Duan.
\newblock Efficient classical simulation of noisy quantum computation.
\newblock {\em arXiv preprint arXiv:1810.03176}, 2018.

\bibitem{noh2020efficient}
Kyungjoo Noh, Liang Jiang, and Bill Fefferman.
\newblock Efficient classical simulation of noisy random quantum circuits in one dimension.
\newblock {\em Quantum}, 4:318, 2020.

\bibitem{fontana2023classical}
Enrico Fontana, Manuel~S Rudolph, Ross Duncan, Ivan Rungger, and Cristina C{\^\i}rstoiu.
\newblock Classical simulations of noisy variational quantum circuits.
\newblock {\em npj Quantum Information}, 11(1):1--12, 2025.

\bibitem{aharonov2022polynomial}
Dorit Aharonov, Xun Gao, Zeph Landau, Yunchao Liu, and Umesh Vazirani.
\newblock A polynomial-time classical algorithm for noisy random circuit sampling.
\newblock {\em Proceedings of the 55th Annual ACM Symposium on Theory of Computing}, page 945, 2023.

\bibitem{schuster2024polynomial}
Thomas Schuster, Chao Yin, Xun Gao, and Norman~Y Yao.
\newblock A polynomial-time classical algorithm for noisy quantum circuits.
\newblock {\em Physical Review X}, 15(4):041018, 2025.

\bibitem{gonzalez2024pauli}
Guillermo Gonz{\'a}lez-Garc{\'\i}a, J~Ignacio Cirac, and Rahul Trivedi.
\newblock Pauli path simulations of noisy quantum circuits beyond average case.
\newblock {\em Quantum}, 9:1730, 2025.

\bibitem{angrisani2025simulating}
Armando Angrisani, Antonio~A Mele, Manuel~S Rudolph, M~Cerezo, and Zoe Holmes.
\newblock Simulating quantum circuits with arbitrary local noise using pauli propagation.
\newblock {\em PRX Quantum}, 7(2):020313, 2026.

\bibitem{martinez2025efficient}
Victor Martinez, Armando Angrisani, Ekaterina Pankovets, Omar Fawzi, and Daniel Stilck~Fran{\c{c}}a.
\newblock Efficient simulation of parametrized quantum circuits under nonunital noise through pauli backpropagation.
\newblock {\em Physical Review Letters}, 134(25):250602, 2025.

\bibitem{roberts2015diagnosing}
Daniel~A Roberts and Douglas Stanford.
\newblock Diagnosing chaos using four-point functions in two-dimensional conformal field theory.
\newblock {\em Physical review letters}, 115(13):131603, 2015.

\bibitem{maldacena2016bound}
Juan Maldacena, Stephen~H Shenker, and Douglas Stanford.
\newblock A bound on chaos.
\newblock {\em Journal of High Energy Physics}, 2016(8):1--17, 2016.

\bibitem{roberts2017chaos}
Daniel~A Roberts and Beni Yoshida.
\newblock Chaos and complexity by design.
\newblock {\em Journal of High Energy Physics}, 2017(4):121, 2017.

\bibitem{google2025observation}
Observation of constructive interference at the edge of quantum ergodicity.
\newblock {\em Nature}, 646(8086):825--830, 2025.

\bibitem{barron2026observable}
Samantha~V Barron, Bradley Mitchell, Vinay Tripathi, Francesco Grieco, Ilan Rosen, Francesca Pietracaprina, Davide Materia, Alireza Seif, Darvin Wanisch, Ram{\'o}n~L Panad{\'e}s-Barrueta, et~al.
\newblock Observable estimation in the absence of classical verification.
\newblock {\em arXiv preprint arXiv:2607.25998}, 2026.

\bibitem{knill1998power}
Emanuel Knill and Raymond Laflamme.
\newblock Power of one bit of quantum information.
\newblock {\em Physical Review Letters}, 81(25):5672, 1998.

\bibitem{tranter2019ordering}
Andrew Tranter, Peter~J Love, Florian Mintert, Nathan Wiebe, and Peter~V Coveney.
\newblock Ordering of trotterization: Impact on errors in quantum simulation of electronic structure.
\newblock {\em Entropy}, 21(12):1218, 2019.

\bibitem{low2019well}
Guang~Hao Low, Vadym Kliuchnikov, and Nathan Wiebe.
\newblock Well-conditioned multiproduct hamiltonian simulation.
\newblock {\em arXiv preprint arXiv:1907.11679}, 2019.

\bibitem{rendon2024improved}
Gumaro Rendon, Jacob Watkins, and Nathan Wiebe.
\newblock Improved accuracy for trotter simulations using chebyshev interpolation.
\newblock {\em Quantum}, 8:1266, 2024.

\bibitem{watson2024randomly}
James~D Watson.
\newblock Randomly compiled quantum simulation with exponentially reduced circuit depths.
\newblock {\em arXiv preprint arXiv:2411.04240}, 2024.

\bibitem{watson2025exponentially}
James~D Watson and Jacob Watkins.
\newblock Exponentially reduced circuit depths using trotter error mitigation.
\newblock {\em PRX quantum}, 6(3):030325, 2025.

\bibitem{rudolph2023classical}
Manuel~S Rudolph, Enrico Fontana, Zo{\"e} Holmes, and Lukasz Cincio.
\newblock Classical surrogate simulation of quantum systems with {LOWESA}.
\newblock {\em arXiv preprint arXiv:2308.09109}, 2023.

\bibitem{shao2023simulating}
Yuguo Shao, Fuchuan Wei, Song Cheng, and Zhengwei Liu.
\newblock Simulating noisy variational quantum algorithms: A polynomial approach.
\newblock {\em Physical Review Letters}, 133(12):120603, 2024.

\bibitem{angrisani2024classically}
Armando Angrisani, Alexander Schmidhuber, Manuel~S. Rudolph, M.~Cerezo, Zo\"e Holmes, and Hsin-Yuan Huang.
\newblock {Classically Estimating Observables of Noiseless Quantum Circuits}.
\newblock {\em Phys. Rev. Lett.}, 135:170602, Oct 2025.

\bibitem{lerch2024efficient}
Sacha Lerch, Ricard Puig, Manuel~S Rudolph, Armando Angrisani, Tyson Jones, Marco Cerezo, Supanut Thanasilp, and Zo{\"e} Holmes.
\newblock {Efficient Quantum-Enhanced Classical Simulation for Patches of Quantum Landscapes}.
\newblock {\em PRX Quantum}, 7(2):020359, 2026.

\bibitem{beguvsic2024real}
Tomislav Begu{\v{s}}i{\'c} and Garnet Kin-Lic Chan.
\newblock Real-time operator evolution in two and three dimensions via sparse pauli dynamics.
\newblock {\em PRX Quantum}, 6(2):020302, 2025.

\bibitem{fuller2025improved}
Bryce Fuller, Minh~C Tran, Danylo Lykov, Caleb Johnson, Max Rossmannek, Ken~Xuan Wei, Andre He, Youngseok Kim, DinhDuy Vu, Kunal Sharma, et~al.
\newblock Improved quantum computation using operator backpropagation.
\newblock {\em arXiv preprint arXiv:2502.01897}, 2025.

\bibitem{rall2019simulation}
Patrick Rall, Daniel Liang, Jeremy Cook, and William Kretschmer.
\newblock Simulation of qubit quantum circuits via pauli propagation.
\newblock {\em Physical Review A}, 99(6):062337, 2019.

\bibitem{beguvsic2023simulating}
Tomislav Begu{\v{s}}i{\'c}, Kasra Hejazi, and Garnet~Kin Chan.
\newblock Simulating quantum circuit expectation values by {C}lifford perturbation theory.
\newblock {\em The Journal of Chemical Physics}, 162(15), 2025.

\bibitem{rudolph2026thermal}
Manuel~S Rudolph, Armando Angrisani, Andrew Wright, Iwo Sanderski, Ricard Puig, and Zo{\"e} Holmes.
\newblock {Thermal State Simulation with Pauli and Majorana Propagation}.
\newblock {\em arXiv preprint arXiv:2602.04878}, 2026.

\bibitem{cirstoiu2024fourier}
Cristina Cirstoiu.
\newblock A fourier analysis framework for approximate classical simulations of quantum circuits.
\newblock {\em arXiv preprint arXiv:2410.13856}, 2024.

\bibitem{teng2025leveraging}
Yanting Teng, Su~Yeon Chang, Manuel~S Rudolph, and Zo{\"e} Holmes.
\newblock Leveraging symmetry merging in pauli propagation.
\newblock {\em arXiv preprint arXiv:2512.12094}, 2025.

\bibitem{xu2026classical}
Jue Xu, Chu Zhao, Xiangran Zhang, Shuchen Zhu, and Qi~Zhao.
\newblock Classical simulation of noiseless quantum dynamics without randomness.
\newblock {\em arXiv preprint arXiv:2601.15770}, 2026.

\bibitem{wang2026query}
Chunhao Wang and Christopher Ye.
\newblock Query-optimal quantum simulation of lindblad evolution.
\newblock {\em arXiv preprint arXiv:2609.17490}, 2026.

\bibitem{chen2026query}
Boyang Chen, Minbo Gao, Xinzhao Wang, and Shuo Zhou.
\newblock Query-optimal and gate-efficient lindbladian simulation.
\newblock {\em arXiv preprint arXiv:2609.18757}, 2026.

\bibitem{chen2025randomized}
Hongrui Chen, Bowen Li, Jianfeng Lu, and Lexing Ying.
\newblock A randomized method for simulating lindblad equations and thermal state preparation.
\newblock {\em Quantum}, 9:1917, 2025.

\bibitem{kato2026exponentially}
Jumpei Kato, Kaito Wada, Kosuke Ito, and Naoki Yamamoto.
\newblock Exponentially accurate open quantum simulation via randomized dissipation with minimal ancilla.
\newblock {\em PRX Quantum}, 7(3):033046, 2026.

\bibitem{ding2024simulating}
Zhiyan Ding, Xiantao Li, and Lin Lin.
\newblock Simulating open quantum systems using hamiltonian simulations.
\newblock {\em PRX quantum}, 5(2):020332, 2024.

\bibitem{ding2024single}
Zhiyan Ding, Chi-Fang Chen, and Lin Lin.
\newblock Single-ancilla ground state preparation via lindbladians.
\newblock {\em Physical Review Research}, 6(3):033147, 2024.

\bibitem{cleve2016efficient}
Richard Cleve and Chunhao Wang.
\newblock Efficient quantum algorithms for simulating lindblad evolution.
\newblock {\em arXiv preprint arXiv:1612.09512}, 2016.

\bibitem{Fr_rot_2017}
Irénée Frérot, Piero Naldesi, and Tommaso Roscilde.
\newblock Entanglement and fluctuations in the xxz model with power-law interactions.
\newblock {\em Physical Review B}, 95(24), 2017.

\bibitem{heyl2019quantum}
Markus Heyl, Philipp Hauke, and Peter Zoller.
\newblock Quantum localization bounds trotter errors in digital quantum simulation.
\newblock {\em Science Advances}, 5(4):eaau8342, 2019.

\bibitem{mercurio2025quantumtoolbox}
Alberto Mercurio, Yi-Te Huang, Li-Xun Cai, Yueh-Nan Chen, Vincenzo Savona, and Franco Nori.
\newblock Quantumtoolbox. jl: An efficient julia framework for simulating open quantum systems.
\newblock {\em Quantum}, 9:1866, 2025.

\bibitem{muleadyValidatingPhasespaceMethods2023}
Sean~R. Muleady, Mingru Yang, Steven~R. White, and Ana~Maria Rey.
\newblock Validating phase-space methods with tensor networks in two-dimensional spin models with power-law interactions.
\newblock {\em Physical Review Letters}, 131(15), October 2023.

\bibitem{perlinSpinSqueezingShortRange2020}
Michael~A. Perlin, Chunlei Qu, and Ana~Maria Rey.
\newblock Spin {{Squeezing}} with {{Short-Range Spin-Exchange Interactions}}.
\newblock {\em Physical Review Letters}, 125(22):223401, November 2020.

\bibitem{schuster2023operator}
Thomas Schuster and Norman~Y Yao.
\newblock Operator growth in open quantum systems.
\newblock {\em Physical Review Letters}, 131(16):160402, 2023.

\bibitem{miller2025simulation}
Aaron Miller, Joachim Favre, Zo{\"e} Holmes, {\"O}zlem Salehi, Rahul Chakraborty, Anton Nyk{\"a}nen, Zoltan Zimboras, Adam Glos, and Guillermo Garc{\'\i}a-P{\'e}rez.
\newblock {Simulation of Fermionic circuits using Majorana Propagation}.
\newblock {\em arXiv preprint arXiv:2503.18939}, 2025.

\bibitem{d2025majorana}
Matteo D'Anna, Jannes Nys, and Juan Carrasquilla.
\newblock Majorana string simulation of nonequilibrium dynamics in two-dimensional lattice fermion systems.
\newblock {\em arXiv:2511.02809}, 2025.

\bibitem{facelli2026fastconvergence}
Giorgio Facelli, Hamza Fawzi, Omar Fawzi, Cambyse Rouz{\'e}, Sam Slezak, and Daniel~Stilck Fran{\c{c}}a.
\newblock Fast convergence of propagation algorithms for open quantum systems.
\newblock {\em arXiv preprint arXiv:2610.06769}, 2026.

\bibitem{lashkari2013towards}
Nima Lashkari, Douglas Stanford, Matthew Hastings, Tobias Osborne, and Patrick Hayden.
\newblock Towards the fast scrambling conjecture.
\newblock {\em Journal of High Energy Physics}, 2013(4):1--33, 2013.

\bibitem{bentsen2024complexity}
Gregory Bentsen, Bill Fefferman, Soumik Ghosh, Michael~J Gullans, and Yinchen Liu.
\newblock On the complexity of sampling from shallow brownian circuits.
\newblock {\em arXiv preprint arXiv:2411.04169}, 2024.

\bibitem{guseynov2026coherent}
Nikita Guseynov, Zo{\"e} Holmes, and Armando Angrisani.
\newblock Coherent-state propagation: A computational framework for simulating bosonic quantum systems.
\newblock {\em arXiv preprint arXiv:2604.19625}, 2026.

\bibitem{upreti2025quantum}
Varun Upreti, Ulysse Chabaud, Zo{\"e} Holmes, and Armando Angrisani.
\newblock When quantum resources backfire: Non-gaussianity and symplectic coherence in noisy bosonic circuits.
\newblock {\em arXiv preprint arXiv:2510.07264}, 2025.

\bibitem{romanov2026learning}
Nikita Romanov, Petr Ivashkov, Weiyuan Gong, Ishaan Kannan, Andi Gu, Hong-Ye Hu, and Susanne~F Yelin.
\newblock Learning arbitrary lindbladians with quantum error correction.
\newblock {\em arXiv preprint arXiv:2606.18188}, 2026.

\bibitem{ivashkov2026ansatz}
Petr Ivashkov, Nikita Romanov, Weiyuan Gong, Andi Gu, Hong-Ye Hu, and Susanne~F Yelin.
\newblock Ansatz-free learning of lindbladian dynamics in situ.
\newblock {\em arXiv preprint arXiv:2603.05492}, 2026.

\bibitem{lewis2026learning}
Laura Lewis, Ewin Tang, and John Wright.
\newblock Learning the structure of open quantum systems.
\newblock {\em arXiv preprint arXiv:2606.30358}, 2026.

\bibitem{arad2026near}
Itai Arad, Zhili Chen, Naixu Guo, Patrick Rebentrost, and Zhan Yu.
\newblock Near-optimal learning of local lindbladians.
\newblock {\em arXiv preprint arXiv:2606.20535}, 2026.

\bibitem{hartman2002ordinary}
Philip Hartman.
\newblock {\em Ordinary differential equations}.
\newblock SIAM, 2002.

\bibitem{descombes2010exact}
St{\'e}phane Descombes and Mechthild Thalhammer.
\newblock An exact local error representation of exponential operator splitting methods for evolutionary problems and applications to linear schr{\"o}dinger equations in the semi-classical regime.
\newblock {\em BIT Numerical Mathematics}, 50(4):729--749, 2010.

\bibitem{bergh2012interpolation}
J{\"o}ran Bergh and J{\"o}rgen L{\"o}fstr{\"o}m.
\newblock {\em Interpolation spaces: an introduction}.
\newblock Springer Science \& Business Media, 1976.

\bibitem{tripathi2021suppression}
Vinay Tripathi, Huo Chen, Mostafa Khezri, Ka-Wa Yip, EM~Levenson-Falk, and Daniel~A Lidar.
\newblock Suppression of crosstalk in superconducting qubits using dynamical decoupling.
\newblock {\em arXiv preprint arXiv:2108.04530}, 2021.

\bibitem{mcclean2017openfermion}
Jarrod~R. McClean, Kevin~J. Sung, Ian~D. Kivlichan, Yudong Cao, Chengyu Dai, E.~Schuyler Fried, Craig Gidney, Brendan Gimby, Pranav Gokhale, Thomas Häner, Tarini Hardikar, Vojtěch Havlíček, Oscar Higgott, Cupjin Huang, Josh Izaac, Zhang Jiang, Xinle Liu, Sam McArdle, Matthew Neeley, Thomas O'Brien, Bryan O'Gorman, Isil Ozfidan, Maxwell~D. Radin, Jhonathan Romero, Nicholas Rubin, Nicolas P.~D. Sawaya, Kanav Setia, Sukin Sim, Damian~S. Steiger, Mark Steudtner, Qiming Sun, Wei Sun, Daochen Wang, Fang Zhang, and Ryan Babbush.
\newblock Openfermion: The electronic structure package for quantum computers.
\newblock {\em arXiv preprint arXiv:1710.07629}, 2017.

\bibitem{monoprop}
{Algorithmiq}.
\newblock {monoprop}: Software for {Majorana} and {Pauli} propagation.
\newblock \url{https://github.com/Algorithmiq/monoprop}.
\newblock Accessed 30 September 2026.

\bibitem{dowling2026classical}
Neil Dowling.
\newblock Classical simulability from operator entanglement scaling.
\newblock {\em arXiv preprint arXiv:2603.05656}, 2026.

\end{thebibliography}

\clearpage
\newpage
\onecolumngrid

\setcounter{section}{0}
\setcounter{subsection}{0}
\setcounter{equation}{0}
\setcounter{figure}{0}
\setcounter{table}{0}
\setcounter{theorem}{0}

\renewcommand{\thesection}{S\arabic{section}}
\renewcommand{\thesubsection}{S\arabic{section}.\arabic{subsection}}
\makeatletter
\renewcommand{\p@subsection}{}
\makeatother
\renewcommand{\theequation}{S\arabic{equation}}
\renewcommand{\thefigure}{S\arabic{figure}}
\renewcommand{\thetable}{S\arabic{table}}
\renewcommand{\thetheorem}{S\arabic{theorem}}

\renewcommand{\theHsection}{supp.\arabic{section}}
\renewcommand{\theHsubsection}{supp.\arabic{section}.\arabic{subsection}}
\renewcommand{\theHequation}{supp.\arabic{equation}}
\renewcommand{\theHfigure}{supp.\arabic{figure}}
\renewcommand{\theHtable}{supp.\arabic{table}}
\providecommand{\theHtheorem}{}
\renewcommand{\theHtheorem}{supp.\arabic{theorem}}

\clearpage
\newpage
\onecolumngrid

\begin{center}
    {\Large\bfseries Supplementary Information\par}
    \vspace{0.6em}
    {\large\bfseries  Dissipation accelerates quantum and classical simulation of open-system dynamics \par}
    \vspace{1em}
    {\normalsize Armando Angrisani, Ricard Puig, Yanting Teng, and Zo\"{e} Holmes\par}
\end{center}

\vspace{1em}
\setcounter{tocdepth}{0}
\supplementtableofcontents
\clearpage

\suppsection{Summary of notation}
\label{supp}

The following table summarizes the notation used throughout the paper. We emphasize that, in our convention, \(\norm{O}_p\) denotes the \(\ell_p\)-norm of the coefficients in the Pauli expansion of \(O\), rather than the Schatten \(p\)-norm commonly used in the quantum information literature. For \(p=2\), the two notions coincide up to a normalization factor.

\begin{table}[h]
\centering
\begin{tabular}{ll}
\toprule
Symbol & \qquad Meaning \\
\midrule
$\cP_n$ & set of $n$-qubit Pauli operators \\
$|P|$ & Pauli weight of $P$ \\
$\norm{v}_{\ell_p}$ & Vector $p$-norm \\
$\norm{O}_{p}$ & Unweighted Pauli-$p$ norm (Eq.\ \ref{eq:pauli-p-norm}) \\
$\norm{O}_{2,\theta}$ & Weighted Pauli-$2$ norm (Eq.\ \ref{eq:weighted-pauli-two-norm}) \\
$\supp(P_a)$ & Support of the Pauli interaction $P_a$ \\
$J$ & Uniform bound on the coherent strengths $|J_a|$ \\
$\gamma$ & Uniform lower bound on active-sector dissipative rates, $\mathcal D_a$ \\
$\bar\gamma$ & {Uniform upper bound on local dissipative rates, $\mathcal D_a$} \\
$k$ & {Maximum interaction support size} \\
$d$ & {Maximum number of interaction terms containing any one qubit} \\
$p_\star,\theta_\star$ & Coefficient and weight regularity parameters \\
\bottomrule
\end{tabular}
\end{table}

\suppsection{Useful lemmas}

We collect here some standard analytic tools used throughout the supplemental material. We begin with the following form of Gronwall's inequality.
\begin{lemma}[Gronwall inequality {\cite[Chapter~III,
Theorem~1.1]{hartman2002ordinary}}]
\label{lem:gronwall}
Let \(f:[0,T]\to[0,\infty)\) be continuous, and suppose that there are constants \(a,b\geq 0\) such that, for every \(t\in[0,T]\),
\begin{align}
    f(t)
    \leq
    a
    +
    b\int_0^t f(s)\,ds .
\end{align}
Then
\begin{align}
    f(t)\leq a e^{bt}
    \qquad
    \text{for all } t\in[0,T].
\end{align}
\end{lemma}

We also use the standard exact local-error representation for the first-order Lie--Trotter product formula.

\begin{lemma}[Lie--Trotter local-error identity
{\cite[Eq.~(5.3)]{descombes2010exact}}]
\label{lem:lie-trotter-local-error}
Let \(A\) and \(B\) be linear maps on a finite-dimensional vector space. Then, for every \(t\geq0\),
\begin{align}
    \e^{tB}\e^{tA}-\e^{t(A+B)}
    =
    \int_0^t\int_0^s
    &
    \e^{(t-s)(A+B)}
    \e^{(s-u)B}
    [B,A]
    \e^{uB}
    \e^{sA}
    \,\mathrm du\,\mathrm ds.
    \label{eq:lie-trotter-local-error}
\end{align}
\end{lemma}

We will also need a convenient way to bound an induced \(\ell_2\)-operator norm using only the absolute row and column sums of the corresponding matrix.

\begin{lemma}[Schur's test]
\label{lem:quantum-schur-test}
Let \(M=(M_{Q,P})_{Q,P}\) be a complex matrix. Then
\begin{equation}
    \norm{M}_{\ell_2\to\ell_2}
    \leq
    \sqrt{
        \left(
            \sup_P \sum_Q \abs{M_{Q,P}}
        \right)
        \left(
            \sup_Q \sum_P \abs{M_{Q,P}}
        \right)
    }.
\label{eq:quantum-schur-test}
\end{equation}
\end{lemma}

\begin{proof}
Set
\begin{align}
    C:=\sup_P\sum_Q\abs{M_{Q,P}},
    \qquad
    R:=\sup_Q\sum_P\abs{M_{Q,P}}.
\end{align}
For any \(x=(x_P)_P\), the triangle inequality followed by the Cauchy--Schwarz inequality gives, for every \(Q\),
\begin{align}
    \abs{(Mx)_Q}^2
    &=
    \abs{\sum_P M_{Q,P}x_P}^2
    \leq
    \left(
        \sum_P \abs{M_{Q,P}}\abs{x_P}
    \right)^2
    \\
    &\leq
    \left(
        \sum_P \abs{M_{Q,P}}
    \right)
    \left(
        \sum_P \abs{M_{Q,P}}\abs{x_P}^2
    \right)
    \leq
    R\sum_P\abs{M_{Q,P}}\abs{x_P}^2.
\end{align}
Summing over \(Q\), we obtain
\begin{align}
    \norm{Mx}_{\ell_2}^2
    &\leq
    R\sum_{Q,P}
    \abs{M_{Q,P}}\abs{x_P}^2
    \\
    &=
    R\sum_P
    \left(
        \sum_Q\abs{M_{Q,P}}
    \right)
    \abs{x_P}^2
    \\
    &\leq
    RC\sum_P\abs{x_P}^2
    =
    RC\norm{x}_{\ell_2}^2.
\end{align}
Taking square roots and then the supremum over \(x\neq0\) proves \eqref{eq:quantum-schur-test}.
\end{proof}

Finally, the following interpolation result allows us to transfer operator-norm estimates between different \(\ell_p\) spaces.

\begin{lemma}[Riesz--Thorin interpolation
{\cite[Theorem~1.1.1]{bergh2012interpolation}}]
\label{lem:interpolation}
Let \(1\leq r\leq s\leq t\leq\infty\), and let \(\theta\in[0,1]\) satisfy
\begin{align}
    \frac{1}{s}
    =
    \frac{\theta}{r}
    +
    \frac{1-\theta}{t}.
\end{align}
For a linear map \(T:\mathbb C^N\to\mathbb C^M\), define the induced \(\ell_q\)-operator norm by
\begin{align}
    \norm{T}_{q\to q}
    :=
    \sup_{x\neq0}
    \frac{\norm{Tx}_q}{\norm{x}_q}.
\end{align}
Then
\begin{align}
    \norm{T}_{s\to s}
    \leq
    \norm{T}_{r\to r}^{\theta}
    \norm{T}_{t\to t}^{1-\theta}.
\end{align}
\end{lemma}

These lemmas will be invoked repeatedly to control norm growth, estimate weighted transfer matrices, and interpolate between endpoint bounds.

\suppsection{Pauli notation and norms}
\label{supp:preliminaries}

\subsection{Pauli basis}
\label{supp:pauli-representation}

Let
\begin{align}
    \cP_n:=\{I,X,Y,Z\}^{\otimes n}
\end{align}
denote the set of $n$-qubit Pauli operators. The Pauli operators form an orthonormal basis with respect to the normalized Hilbert--Schmidt inner product
\begin{align}
    \langle A,B\rangle
    :=
    2^{-n}\Tr(A^\dagger B).
\end{align}
Hence every operator $O$ has a unique expansion
\begin{equation}
    O=\sum_{P\in\cP_n}c_P P,
    \qquad
    c_P
    =
    2^{-n}\Tr(PO).
    \label{eq:pauli-expansion}
\end{equation}
For $P=P_1\otimes\cdots\otimes P_n\in\cP_n$, we write
\begin{align}
    \supp(P)
    :=
    \{x\in[n]:P_x\neq I\}
\end{align}
for its support and $|P|:=|\supp(P)|$ for its Pauli weight. Parseval's identity gives the normalized Hilbert--Schmidt norm
\begin{equation}
    \norm{O}_2
    :=
    \left(2^{-n}\Tr(O^\dagger O)\right)^{1/2}
    =
    \left(\sum_{P\in\cP_n}|c_P|^2\right)^{1/2}.
    \label{eq:normalized-HS-norm}
\end{equation}

\subsection{Pauli-$p$ and weighted Pauli-$2$ norms}
\label{supp:pauli-norms}

For $1\leq p<\infty$, define the unweighted Pauli-$p$ norm
\begin{equation}
    \norm{O}_{p}
    :=
    \left(
        \sum_{P\in\cP_n}|c_P|^p
    \right)^{1/p}.
    \label{eq:pauli-p-norm}
\end{equation}
For $\theta\geq0$, define the weighted Pauli-$2$ norm
\begin{equation}
    \norm{O}_{2,\theta}
    :=
    \left(
        \sum_{P\in\cP_n}
        \e^{2\theta|P|}|c_P|^2
    \right)^{1/2}.
    \label{eq:weighted-pauli-two-norm}
\end{equation}
Thus $\norm{O}_{2,0}=\norm{O}_2$. The exponent $p$ and the weight parameter $\theta$ quantify different properties of the Pauli expansion. Control of $\norm{O}_p$ for $p<2$ implies that only a limited number of coefficients can be appreciable, whereas control of $\norm{O}_{2,\theta}$ for $\theta>0$ implies exponential suppression of the high-weight tail. Quantitative statements are given in \cref{supp:truncation-preliminaries}.

We will repeatedly use two monotonicity properties. If $1\leq p\leq q<\infty$, then
\begin{equation}
    \norm{O}_{q}
    \leq
    \norm{O}_{p}.
    \label{eq:p-monotonicity}
\end{equation}
Indeed, the two sides are the $\ell_q$ and $\ell_p$ norms of the same finite coefficient vector $(|c_P|)_P$, and finite-dimensional sequence norms decrease with their exponent. Moreover, if $0\leq\eta\leq\theta$, then
\begin{equation}
    \norm{O}_{2,\eta}
    \leq
    \norm{O}_{2,\theta},
    \label{eq:theta-monotonicity}
\end{equation}
because $\e^{\eta|P|}|c_P|\leq\e^{\theta|P|}|c_P|$ for every $P$.

\subsection{Induced and absolute induced norms}
\label{supp:induced-norms}

Let $\Phi$ be a linear map on operators. For $1\leq p,q<\infty$, define
\begin{equation}
    \norm{\Phi}_{p\to q}
    :=
    \sup_{O\neq0}
    \frac{\norm{\Phi(O)}_q}{\norm{O}_p}.
    \label{eq:pauli-induced-norm}
\end{equation}
For weighted Pauli-$2$ spaces, define
\begin{equation}
    \norm{\Phi}_{2,\theta\to 2,\eta}
    :=
    \sup_{O\neq0}
    \frac{\norm{\Phi(O)}_{2,\eta}}
         {\norm{O}_{2,\theta}}.
    \label{eq:weighted-two-induced-norm}
\end{equation}

Write the Pauli transfer matrix of $\Phi$ as
\begin{equation}
    \Phi(P)
    =
    \sum_{Q\in\cP_n}\Phi_{Q,P}Q,
    \qquad
    \Phi_{Q,P}
    =
    2^{-n}\Tr\!\left(Q\Phi(P)\right).
    \label{eq:pauli-transfer-matrix}
\end{equation}
We define the absolute Pauli transfer map by
\begin{equation}
    \Phi^{\mathrm{abs}}(P)
    :=
    \sum_{Q\in\cP_n}|\Phi_{Q,P}|Q
    \label{eq:absolute-transfer-map}
\end{equation}
and extend it linearly. Its unweighted and weighted induced norms are
\begin{equation}
    \norm{\Phi}_{p\to q}^{\mathrm{abs}}
    :=
    \norm{\Phi^{\mathrm{abs}}}_{p\to q},
    \qquad
    \norm{\Phi}_{2,\theta\to 2,\eta}^{\mathrm{abs}}
    :=
    \norm{\Phi^{\mathrm{abs}}}_{2,\theta\to 2,\eta}.
    \label{eq:absolute-induced-norms}
\end{equation}
The map $\Phi^{\mathrm{abs}}$ is generally not a physical quantum map; it simply removes cancellations between Pauli transitions. If $O=\sum_Pc_PP$, then the coefficient of $Q$ in $\Phi(O)$ satisfies
\begin{equation}
    \left|\sum_P\Phi_{Q,P}c_P\right|
    \leq
    \sum_P|\Phi_{Q,P}|\,|c_P|.
    \label{eq:absolute-coefficient-domination}
\end{equation}
Taking the relevant output sequence norm and using that replacing $c_P$ by $|c_P|$ leaves the input norm unchanged gives
\begin{equation}
    \norm{\Phi}_{p\to q}
    \leq
    \norm{\Phi}_{p\to q}^{\mathrm{abs}},
    \qquad
    \norm{\Phi}_{2,\theta\to 2,\eta}
    \leq
    \norm{\Phi}_{2,\theta\to 2,\eta}^{\mathrm{abs}}.
    \label{eq:ordinary-bounded-by-absolute}
\end{equation}

The ordinary induced norms are submultiplicative:
\begin{equation}
    \norm{\Psi\circ\Phi}_{p\to r}
    \leq
    \norm{\Psi}_{q\to r}\norm{\Phi}_{p\to q},
    \label{eq:unweighted-induced-submultiplicativity}
\end{equation}
and
\begin{equation}
    \norm{\Psi\circ\Phi}_{2,\theta\to 2,\zeta}
    \leq
    \norm{\Psi}_{2,\eta\to 2,\zeta}\norm{\Phi}_{2,\theta\to 2,\eta}.
    \label{eq:weighted-induced-submultiplicativity}
\end{equation}
The absolute induced norms obey the same composition bounds. Indeed,
\begin{equation}
    |(\Psi\circ\Phi)_{R,P}|
    \leq
    \sum_Q|\Psi_{R,Q}|\,|\Phi_{Q,P}|,
    \label{eq:absolute-composition-coefficients}
\end{equation}
so $(\Psi\circ\Phi)^{\mathrm{abs}}$ is entrywise dominated by $\Psi^{\mathrm{abs}}\circ\Phi^{\mathrm{abs}}$.

\subsection{Interpolation of induced Pauli norms}
\label{supp:pauli-interpolation}

We use the standard Riesz--Thorin interpolation theorem. Suppose that
\begin{equation}
    \norm{\Phi}_{p_i\to q_i}
    \leq
    M_i,
    \qquad i\in\{0,1\}.
    \label{eq:riesz-thorin-endpoints}
\end{equation}
For $0\leq\lambda\leq1$, define
\begin{equation}
    \frac{1}{p_\lambda}
    :=
    \frac{1-\lambda}{p_0}+\frac{\lambda}{p_1},
    \qquad
    \frac{1}{q_\lambda}
    :=
    \frac{1-\lambda}{q_0}+\frac{\lambda}{q_1}.
    \label{eq:riesz-thorin-exponents}
\end{equation}
Then
\begin{equation}
    \norm{\Phi}_{p_\lambda\to q_\lambda}
    \leq
    M_0^{1-\lambda}M_1^\lambda.
    \label{eq:riesz-thorin-pauli}
\end{equation}
We will apply this result only to unweighted Pauli coefficient spaces, interpolating between $p=1$ and $p=2$.

The bound above is the general Riesz--Thorin theorem, whereas Lemma~\ref{lem:interpolation} states its $p\to p$ special case. We apply the theorem to the Pauli transfer matrix, since the Pauli norms are the $\ell_p$ norms of the coefficient vector. In the applications below, $p_0=q_0=1$ and $p_1=q_1=2$, so $p_\lambda=q_\lambda=p$. Thus Lemma~\ref{lem:interpolation} suffices, with $r=1$, $s=p$, $t=2$, and $\theta=1-\lambda=2/p-1$.

\subsection{Weight and coefficient truncation}
\label{supp:truncation-preliminaries}

For an integer $w\geq0$, let
\begin{equation}
    \Pi_{\leq w}(O)
    :=
    \sum_{\substack{P\in\cP_n\\|P|\leq w}}c_PP
    \label{eq:weight-projector}
\end{equation}
be the orthogonal projection onto Pauli operators of weight at most $w$.

\begin{lemma}[High-weight tail bound]
\label{lem:weight-tail}
For every $\theta>0$,
\begin{equation}
    \norm{(I-\Pi_{\leq w})O}_2
    \leq
    \e^{-\theta(w+1)}\norm{O}_{2,\theta}.
    \label{eq:L2-weight-tail}
\end{equation}
\end{lemma}

\begin{proof}
Every discarded Pauli operator satisfies $|P|\geq w+1$, and therefore
\begin{align}
    \norm{(I-\Pi_{\leq w})O}_2^2
    &=
    \sum_{|P|\geq w+1}|c_P|^2
    \\
    &\leq
    \e^{-2\theta(w+1)}
    \sum_{P}\e^{2\theta|P|}|c_P|^2.
\end{align}
Taking the square root proves the claim.
\end{proof}

For a threshold $\tau>0$, define the nonlinear pruning map
\begin{equation}
    \mathcal T_\tau(O)
    :=
    \sum_{\substack{P\in\cP_n\\|c_P|\geq\tau}}c_PP.
    \label{eq:coefficient-pruning-map}
\end{equation}

\begin{lemma}[Coefficient pruning from Pauli-$p$ control]
\label{lem:coefficient-pruning}
Let $1\leq p<2$. Then
\begin{equation}
    \norm{O-\mathcal T_\tau(O)}_2
    \leq
    \tau^{1-p/2}\norm{O}_{p}^{p/2}.
    \label{eq:coefficient-pruning-error}
\end{equation}
Moreover, the retained Pauli coefficients satisfy
\begin{equation}
    \left|
        \left\{P\in\cP_n:|c_P|\geq\tau\right\}
    \right|
    \leq
    \tau^{-p}\norm{O}_{p}^{p}.
    \label{eq:number-retained-coefficients}
\end{equation}
\end{lemma}

\begin{proof}
For every discarded coefficient, $|c_P|<\tau$, and hence
\begin{align}
    |c_P|^2
    =
    |c_P|^p|c_P|^{2-p}
    \leq
    \tau^{2-p}|c_P|^p.
\end{align}
Summing over the discarded coefficients and taking the square root gives \cref{eq:coefficient-pruning-error}. For the counting bound, if $s$ coefficients are retained, then each contributes at least $\tau^p$ to $\|O\|_p^p=\sum_P|c_P|^p$. Thus, $s \tau^p \leq \norm{O}_{p}^{p}$, which gives the claimed bound.
\end{proof}

Define the joint weight-and-coefficient truncation
\begin{equation}
    \mathcal C_{w,\tau}(O)
    :=
    \Pi_{\leq w}\mathcal T_\tau(O)
    =
    \sum_{\substack{P\in\cP_n\\|P|\leq w,\ |c_P|\geq\tau}}c_PP.
    \label{eq:joint-truncation-map}
\end{equation}

\begin{corollary}[Joint weight and coefficient truncation]
\label{cor:joint-truncation}
Let $1\leq p<2$ and $\theta>0$. Then
\begin{equation}
    \norm{O-\mathcal C_{w,\tau}(O)}_2
    \leq
    \e^{-\theta(w+1)}\norm{O}_{2,\theta}
    +
    \tau^{1-p/2}\norm{O}_{p}^{p/2}.
    \label{eq:joint-truncation-error}
\end{equation}
The truncated operator contains only Pauli operators of weight at most $w$, and their number is bounded by
\begin{equation}
    \left|
        \left\{P:\mathcal C_{w,\tau}(O)_P\neq0\right\}
    \right|
    \leq
    \tau^{-p}\norm{O}_{p}^{p}.
    \label{eq:joint-truncation-support-bound}
\end{equation}
\end{corollary}

\begin{proof}
Using
\begin{align}
    O-\mathcal C_{w,\tau}(O)
    =
    (I-\Pi_{\leq w})O
    +
    \Pi_{\leq w}\bigl(O-\mathcal T_\tau(O)\bigr),
\end{align}
the triangle inequality, Lemma~\ref{lem:weight-tail}, and contractivity of $\Pi_{\leq w}$ in the normalized Hilbert--Schmidt norm give \cref{eq:joint-truncation-error}. The support-size bound follows from \cref{eq:number-retained-coefficients}.
\end{proof}

\suppsection{Local open-system dynamics}
\label{supp:model}

We consider the Heisenberg evolution
\begin{equation}
    O(t)
    =
    \e^{t\mathcal L}(O).
    \label{eq:heisenberg-evolution}
\end{equation}

\subsection{Local Lindbladian terms}
\label{supp:local-noise}

We decompose the full Lindbladian directly into local Lindbladian terms,
\begin{equation}
    \mathcal L
    =
    \sum_{a=1}^{m}\mathcal L_a,
    \qquad
    \mathcal L_a
    :=
    \mathcal H_a+\mathcal D_a.
    \label{eq:full-lindbladian}
\end{equation}
For every label $a$, the coherent part is generated by a Pauli interaction $P_a$ supported on at most $k$ qubits,
\begin{equation}
    |\supp(P_a)|\leq k.
    \label{eq:local-block-size}
\end{equation}
The corresponding Hamiltonian generator is
\begin{equation}
    \mathcal H_a(O)
    :=
    iJ_a[P_a,O],
    \qquad
    |J_a|\leq J.
    \label{eq:local-hamiltonian-generator}
\end{equation}
Moreover, we have that
\begin{equation}
    H
    =
    \sum_{a=1}^{m}J_aP_a,
    \qquad
    \mathcal H
    =
    \sum_{a=1}^{m}\mathcal H_a.
    \label{eq:full-hamiltonian-generator}
\end{equation}

The dissipative contribution $\mathcal D_a$ is supported on $\supp(P_a)$ and is assumed to be a Pauli-diagonal Lindblad generator. We therefore define its rates directly on the $n$-qubit Pauli basis:
\begin{equation}
    \mathcal D_a(P)
    =
    -\lambda_a(P)P,
    \qquad
    P\in\cP_n.
    \label{eq:local-pauli-diagonal-noise}
\end{equation}
Because $\mathcal D_a$ is supported on $\supp(P_a)$, $\lambda_a(P)=0$ whenever $\supp(P)\cap\supp(P_a)=\varnothing$. We assume the uniform bounds
\begin{equation}
    \gamma
    \leq
    \lambda_a(P)
    \leq
    \bar\gamma,
    \qquad
    P\in\cP_n,
    \quad
    \supp(P)\cap\supp(P_a)\neq\varnothing,
    \label{eq:local-noise-window}
\end{equation}
for some $0<\gamma\leq\bar\gamma$. Hence every Pauli operator whose support intersects $\supp(P_a)$ is damped by $\mathcal D_a$ at rate at least $\gamma$. The upper bound $\bar\gamma$ is used only for mixed coherent--dissipative commutators in the product-formula analysis and in the Richardson extrapolation; the regularity arguments require only the lower bound $\gamma$.

For bookkeeping, the dissipative part of the full generator may be denoted by
\begin{equation}
    \mathcal D
    :=
    \sum_{a=1}^{m}\mathcal D_a.
\end{equation}
It is derived from the primary decomposition $\mathcal L=\sum_a(\mathcal H_a+\mathcal D_a)$ rather than specified independently and then distributed among the local contributions. It is again Pauli diagonal,
\begin{equation}
    \mathcal D(P)
    =
    -\lambda_P P,
    \qquad
    \lambda_P
    =
    \sum_{a=1}^{m}\lambda_a(P).
    \label{eq:pauli-diagonal-noise}
\end{equation}
The class includes sums of single-site Pauli generators satisfying \cref{eq:local-noise-window}, as well as genuine few-qubit dissipation, such as the two-qubit depolarizing dissipation~\cite{tripathi2021suppression}. Indeed, a two-qubit depolarizing dissipator with $|\supp(P_a)|=2$ satisfies $\lambda_a(P)=\gamma_a>0$ for every $P\in\cP_n$ such that $\supp(P)\cap\supp(P_a)\neq\varnothing$.

\subsection{Interaction hypergraph and coloring}
\label{supp:interaction-coloring}

The collection $\{\supp(P_a)\}_{a=1}^{m}$ defines a hypergraph on the qubit set $[n]$. For each qubit $x$, let
\begin{equation}
    d_x
    :=
    \abs{\{a\in[m]:x\in \supp(P_a)\}},
    \qquad
    d
    :=
    \max_{x\in[n]}d_x,
    \label{eq:local-interaction-degree}
\end{equation}
such that $d$ is the local degree of the Lindbladian dynamics.

We associate with this hypergraph an interaction graph $G$ whose vertices are the labels $a\in[m]$, with $a$ and $b$ adjacent whenever
\begin{equation}
    \supp(P_a)\cap \supp(P_b)\neq\varnothing.
    \label{eq:interaction-graph-adjacency}
\end{equation}

\begin{lemma}[Interaction-graph coloring]
\label{lem:bounded-degree-coloring}
Write $\Delta(G):=\max_{a\in[m]}\deg_G(a)$ for the maximum
degree of the interaction graph, then $\Delta(G)$ satisfies
\begin{equation}
    \Delta(G)
    \leq
    k(d-1).
    \label{eq:interaction-graph-degree}
\end{equation}
Consequently, the labels admit a partition
\begin{equation}
    [m]
    =
    \mathcal C_1\sqcup\cdots\sqcup\mathcal C_\chi
    \label{eq:color-partition}
\end{equation}
into at most
\begin{equation}
    \chi
    \leq
    k(d-1)+1
    \label{eq:number-colors}
\end{equation}
colors such that the sets $\{\supp(P_a):a\in\mathcal C_c\}$ are pairwise disjoint for every fixed color $c$.
\end{lemma}

\begin{proof}
Fix $a\in[m]$. For each $x\in \supp(P_a)$, at most $d_x-1$ other terms contain $x$. Therefore
\begin{equation}
    \deg_G(a)
    \leq
    \sum_{x\in \supp(P_a)}(d_x-1)
    \leq
    k(d-1).
\end{equation}
A greedy coloring uses at most $\Delta(G)+1$ colors.
\end{proof}

For every color $c$, define
\begin{equation}
    \mathcal H^{(c)}
    :=
    \sum_{a\in\mathcal C_c}\mathcal H_a,
    \qquad
    \mathcal D^{(c)}
    :=
    \sum_{a\in\mathcal C_c}\mathcal D_a,
    \qquad
    \mathcal L^{(c)}
    :=
    \sum_{a\in\mathcal C_c}\mathcal L_a.
    \label{eq:color-open-system-generator}
\end{equation}
Since the terms assigned the same color have disjoint supports,
\begin{equation}
    \e^{t\mathcal H^{(c)}}
    =
    \prod_{a\in\mathcal C_c}\e^{t\mathcal H_a},
    \qquad
    \e^{t\mathcal L^{(c)}}
    =
    \prod_{a\in\mathcal C_c}\e^{t\mathcal L_a}.
    \label{eq:exact-color-factorization}
\end{equation}
Local generators with distinct labels in the same color have disjoint supports and therefore commute. The factorizations above are exact and allow each color evolution to be implemented as parallel local evolutions. Since the colors partition all local terms, their generators sum to the original generator:
\begin{equation}
    \mathcal L
    =
    \sum_{c=1}^{\chi}\mathcal L^{(c)}.
    \label{eq:full-lindbladian-decomposition}
\end{equation}
We use this decomposition to build a product formula from color layers. Splitting different colors can introduce an error, as can splitting the coherent and dissipative parts within a local update.

\subsection{Finite-step decompositions}
\label{supp:finite-step-decompositions}

For a step size $\delta\geq0$, define the split local evolution
\begin{equation}
    \mathcal T_{a,\delta}
    :=
    \e^{\delta\mathcal D_a}
    \e^{\delta\mathcal H_a}.
    \label{eq:split-local-map}
\end{equation}
Products of maps are read from right to left, so $\mathcal T_{a,\delta}$ first applies the coherent rotation and then the dissipative map associated with the same label $a$. By contrast, the exact local evolution is $\e^{\delta\mathcal L_a}$.

For every color $c$, let
\begin{equation}
    \mathcal D^{(c)}
    :=
    \sum_{a\in\mathcal C_c}\mathcal D_a.
    \label{eq:color-noise-generator}
\end{equation}
The split color evolution is
\begin{equation}
    \mathcal T_{\delta}^{(c)}
    :=
    \e^{\delta\mathcal D^{(c)}}
    \e^{\delta\mathcal H^{(c)}}
    =
    \prod_{a\in\mathcal C_c}
    \mathcal T_{a,\delta},
    \label{eq:split-color-map}
\end{equation}
where the last equality follows from the disjointness of the supports within each color. The corresponding exact color evolution is $\e^{\delta\mathcal L^{(c)}}$, which factorizes over the local terms by \cref{eq:exact-color-factorization}.

The full Trotterised step is the ordered product of split color maps,
\begin{equation}
    \Phi_\delta^{\mathrm{Tr}}
    :=
    \mathcal T_{\delta}^{(\chi)}
    \cdots
    \mathcal T_{\delta}^{(1)}.
    \label{eq:implemented-color-step}
\end{equation}
For comparison, the color-product step is the product of exact color evolutions
\begin{equation}
    \Phi_\delta^{\mathrm{col}}
    :=
    \e^{\delta\mathcal L^{(\chi)}}
    \cdots
    \e^{\delta\mathcal L^{(1)}}.
    \label{eq:exact-color-product}
\end{equation}
The finite-step error then decomposes as
\begin{equation}
\begin{aligned}
    \Phi_\delta^{\mathrm{Tr}}-\e^{\delta\mathcal L}
    ={}&
    \left(
        \Phi_\delta^{\mathrm{Tr}}
        -
        \Phi_\delta^{\mathrm{col}}
    \right)
    \\
    &+
    \left(
        \Phi_\delta^{\mathrm{col}}
        -
        \e^{\delta\mathcal L}
    \right).
\end{aligned}
    \label{eq:two-stage-color-decomposition}
\end{equation}
The first term is the error incurred by splitting the coherent and dissipative parts within each color. The second is the first-order product-formula error between distinct color generators.

All transformations within one color can be performed in parallel. The sequential depth of one implemented step is therefore at most
\begin{equation}
    \chi
    \leq
    k(d-1)+1,
    \label{eq:color-step-depth}
\end{equation}
independently of the total number of local generators and qubits.

\subsection{Passage to continuous time}
\label{supp:continuous-time-limit}

Recall that
\begin{align}
    \mathcal T_{\delta}^{(c)}
    =
    \e^{\delta\mathcal D^{(c)}}
    \e^{\delta\mathcal H^{(c)}},
    \qquad
    \Phi_\delta^{\mathrm{Tr}}
    =
    \mathcal T_{\delta}^{(\chi)}\cdots\mathcal T_{\delta}^{(1)}.
\end{align}
Since the operator space is finite dimensional, the Lie--Trotter product formula converges in every induced norm. For every $t\geq0$,
\begin{equation}
    \e^{t\mathcal L}
    =
    \lim_{r\to\infty}
    \left(
        \e^{(t/r)\mathcal L^{(\chi)}}
        \cdots
        \e^{(t/r)\mathcal L^{(1)}}
    \right)^r.
    \label{eq:continuous-time-color-limit}
\end{equation}
Moreover,
\begin{equation}
    \mathcal T_{\delta}^{(c)}
    =
    I+\delta\mathcal L^{(c)}+O(\delta^2)
    \label{eq:split-color-first-order-expansion}
\end{equation}
as $\delta\to0$, and therefore
\begin{equation}
    \e^{t\mathcal L}
    =
    \lim_{r\to\infty}
    \left(\Phi_{t/r}^{\mathrm{Tr}}\right)^r.
    \label{eq:continuous-time-split-limit}
\end{equation}

\begin{lemma}[Contractivity from finite steps]
\label{lem:finite-step-to-continuous-contractivity}
Fix $1\leq p<\infty$ and $\theta\geq0$. If $\norm{\Phi_\delta^{\mathrm{Tr}}}_{p\to p}\leq1$ for every sufficiently small $\delta>0$, then
\begin{equation}
    \norm{\e^{t\mathcal L}}_{p\to p}
    \leq1
    \label{eq:continuous-time-p-contractivity}
\end{equation}
for every $t\geq0$. Likewise, if $\norm{\Phi_\delta^{\mathrm{Tr}}}_{2,\theta\to 2,\theta}\leq1$ for every sufficiently small $\delta>0$, then
\begin{equation}
    \norm{\e^{t\mathcal L}}_{2,\theta\to 2,\theta}
    \leq1
    \label{eq:continuous-time-weighted-two-contractivity}
\end{equation}
for every $t\geq0$.
\end{lemma}

\begin{proof}
Choose $r$ sufficiently large that $t/r$ lies in the corresponding small-step regime. Submultiplicativity gives
\begin{align}
    \norm{(\Phi_{t/r}^{\mathrm{Tr}})^r}_{p\to p}\leq1
\end{align}
in the first case and
\begin{align}
    \norm{(\Phi_{t/r}^{\mathrm{Tr}})^r}_{2,\theta\to 2,\theta}\leq1
\end{align}
in the second. Taking $r\to\infty$ and using \cref{eq:continuous-time-split-limit} proves both statements.
\end{proof}

The same argument applies when contractivity is established for each exact color evolution $\e^{\delta\mathcal L^{(c)}}$ rather than for the split step.

\suppsection{Pauli coefficient and weight regularity}
\label{supp:regularity}

The main result of this section is given by Theorem~\ref{thm:supp-regularity}, which proves contractivity in two complementary norms. The weighted Pauli-$2$ bound controls the high-weight tail, while the Pauli-$p$ bound for $p<2$ controls coefficient truncation error and the number of retained coefficients. We prove the local bounds in \cref{supp:weighted-two-contraction,supp:coefficient-regularity} and extend them to full product-formula steps and continuous evolution in \cref{supp:composition-over-colors}.

Recall that the split local evolution is
\begin{align}
    \mathcal T_{a,t}
    =
    \e^{t\mathcal D_a}\e^{t\mathcal H_a},
\end{align}
whereas the exact local evolution is $\e^{t\mathcal L_a}$.

Fix a Hamiltonian label $a$. Define the active Pauli subspace
\begin{equation}
    \mathsf K_a
    :=
    \operatorname{span}
    \left\{
        P\in\cP_n:
        \mathrm{supp}(P)\cap \supp(P_a)\neq\emptyset
    \right\}.
    \label{eq:active-pauli-subspace}
\end{equation}
Its Pauli-orthogonal complement is
\begin{equation}
    \mathsf K_a^\perp
    =
    \operatorname{span}
    \left\{
        P\in\cP_n:
        \mathrm{supp}(P)\cap \supp(P_a)=\emptyset
    \right\}.
    \label{eq:inactive-pauli-subspace}
\end{equation}
Both subspaces are invariant under $\mathcal H_a$ and $\mathcal D_a$, and both generators vanish on $\mathsf K_a^\perp$. Thus $\e^{t\mathcal L_a}$ and $\mathcal T_{a,t}$ act as the identity on $\mathsf K_a^\perp$.

The action of $\mathcal H_a$ on $\mathsf K_a$ decomposes into invariant one- and two-dimensional Pauli subspaces. Every Pauli operator commuting with $P_a$ is fixed. If $P_a$ anticommutes with $P$, let $Q$ be the unique phase-free Pauli operator proportional to $P_aP$. There exists $\xi_{a,P}\in\{-1,+1\}$ such that
\begin{equation}
    \mathcal H_a(P)
    =
    2J_a\xi_{a,P}Q,
    \qquad
    \mathcal H_a(Q)
    =
    -2J_a\xi_{a,P}P.
    \label{eq:local-pauli-pair-generator}
\end{equation}
Since multiplication by $P_a$ changes tensor factors only inside $\supp(P_a)$,
\begin{equation}
    \bigl||Q|-|P|\bigr|
    \leq
    |P_a|
    \leq k.
    \label{eq:local-weight-change}
\end{equation}

\subsection{Weighted Pauli-$2$ regularity}
\label{supp:weighted-two-contraction}

To control the weighted Pauli-$2$ norm, we first bound Hamiltonian growth. Adding noise then determines the range $0\leq\theta\leq\theta_\star$ in which the complete local evolution is contractive.

\begin{lemma}[Weighted growth of one local Pauli rotation]
\label{lem:weighted-local-H-growth}
For every $a$, every $\theta\geq0$, and every $t\geq0$,
\begin{equation}
    \norm{\e^{t\mathcal H_a}}_{2,\theta\to 2,\theta}
    \leq
    \exp\!\left(
        2|J_a|t\sinh(k\theta)
    \right).
    \label{eq:weighted-local-H-growth}
\end{equation}
\end{lemma}

\begin{proof}
The Pauli basis decomposes into invariant one- and two-dimensional subspaces under $\mathcal H_a$. Pauli operators commuting with $P_a$ span one-dimensional invariant subspaces on which $\mathcal H_a$ vanishes. It therefore suffices to consider a two-dimensional invariant subspace spanned by anticommuting partners $P$ and $Q$.

Write
\begin{equation}
    O(t)
    =
    \e^{t\mathcal H_a}(O(0))
    =
    u(t)P+v(t)Q.
\end{equation}
By \cref{eq:local-pauli-pair-generator},
\begin{equation}
    \dot u(t)
    =
    -2J_a\xi_{a,P}v(t),
    \qquad
    \dot v(t)
    =
    2J_a\xi_{a,P}u(t).
    \label{eq:local-pair-coefficient-equations}
\end{equation}
Define the weighted coefficients
\begin{equation}
    y_P(t)
    :=
    \e^{\theta|P|}u(t),
    \qquad
    y_Q(t)
    :=
    \e^{\theta|Q|}v(t),
\end{equation}
and let
\begin{align}
    \Delta:=|Q|-|P|.
\end{align}
Then
\begin{equation}
    \dot y_P(t)
    =
    -2J_a\xi_{a,P}\e^{-\theta\Delta}y_Q(t),
    \qquad
    \dot y_Q(t)
    =
    2J_a\xi_{a,P}\e^{\theta\Delta}y_P(t).
    \label{eq:weighted-local-pair-equations}
\end{equation}

Define
\begin{equation}
    F(t)
    :=
    |y_P(t)|^2+|y_Q(t)|^2,
\end{equation}
so that $F(t)=\norm{O(t)}_{2,\theta}^2$ on this invariant pair. Using \cref{eq:weighted-local-pair-equations}, we obtain
\begin{align}
    F'(t)
    &=
    2\operatorname{Re}\!\left(
        \overline{y_P(t)}\dot y_P(t)
        +
        \overline{y_Q(t)}\dot y_Q(t)
    \right)
    \nonumber\\
    &=
    8J_a\xi_{a,P}\sinh(\theta\Delta)
    \operatorname{Re}\!\left(
        \overline{y_P(t)}y_Q(t)
    \right)
    \nonumber\\
    &\leq
    8|J_a|
    |\sinh(\theta\Delta)|
    |y_P(t)|\,|y_Q(t)|
    \nonumber\\
    &\leq
    4|J_a|\sinh(k\theta)F(t),
    \label{eq:weighted-local-pair-differential-bound}
\end{align}
where we used
\begin{align}
    |\Delta|
    \leq
    |P_a|
    \leq
    k
\end{align}
and
\begin{align}
    2|y_P(t)y_Q(t)|
    \leq
    |y_P(t)|^2+|y_Q(t)|^2.
\end{align}

Integrating \cref{eq:weighted-local-pair-differential-bound} from $0$ to $t$ gives
\begin{align}
    F(t)
    &=
    F(0)
    +
    \int_0^t F'(s)\,\mathrm ds
    \nonumber\\
    &\leq
    F(0)
    +
    4|J_a|\sinh(k\theta)
    \int_0^t F(s)\,\mathrm ds.
    \label{eq:weighted-local-pair-integral-bound}
\end{align}
We may therefore apply Gronwall's inequality in Lemma~\ref{lem:gronwall} with
\begin{equation}
    f(t)=F(t),
    \qquad
    a=F(0),
    \qquad
    b=4|J_a|\sinh(k\theta),
\end{equation}
which yields
\begin{equation}
    F(t)
    \leq
    F(0)
    \exp\!\left(
        4|J_a|t\sinh(k\theta)
    \right).
    \label{eq:weighted-local-pair-gronwall-bound}
\end{equation}

The same inequality holds on every one-dimensional commuting sector, since the corresponding weighted coefficient is constant. The invariant sectors are spanned by disjoint subsets of the Pauli basis and are therefore orthogonal also in the weighted inner product. Their squared weighted norms add without cross terms. Summing \cref{eq:weighted-local-pair-gronwall-bound} over these sectors gives
\begin{equation}
    \norm{\e^{t\mathcal H_a}(O)}_{2,\theta}^2
    \leq
    \exp\!\left(
        4|J_a|t\sinh(k\theta)
    \right)
    \norm{O}_{2,\theta}^2.
\end{equation}
Taking square roots and then the supremum over $O\neq0$ proves \cref{eq:weighted-local-H-growth}.
\end{proof}

We compare the Hamiltonian growth with local damping. For every $P\in\cP_n$ with $\supp(P)\cap\supp(P_a)\neq\varnothing$, the local dissipator damps the coefficient of $P$ at rate at least $\gamma$:
\begin{equation}
    \lambda_a(P)
    \geq
    \gamma.
    \label{eq:minimum-block-damping}
\end{equation}
This is precisely the local spectral-gap assumption in \cref{eq:local-noise-window}.

\begin{lemma}[Weighted-$2$ bound for a local noisy evolution]
\label{lem:strict-weighted-two-local}
For every $a$, every $\theta\geq0$, and every $t\geq0$,
\begin{align}
    \norm{\left.\e^{t\mathcal L_a}\right|_{\mathsf K_a}}_{2,\theta\to 2,\theta}
    &\leq
    \exp\!\left[-t\left(
        \gamma-2|J_a|\sinh(k\theta)
    \right)\right],
    \label{eq:exact-local-weighted-two-bound}
    \\
    \norm{\left.\mathcal T_{a,t}\right|_{\mathsf K_a}}_{2,\theta\to 2,\theta}
    &\leq
    \exp\!\left[-t\left(
        \gamma-2|J_a|\sinh(k\theta)
    \right)\right].
    \label{eq:split-local-weighted-two-bound}
\end{align}
\end{lemma}

\begin{proof}
For the exact evolution, the proof of Lemma~\ref{lem:weighted-local-H-growth} bounds the Hamiltonian contribution to the derivative of the squared weighted norm by $4|J_a|\sinh(k\theta)\norm{O(t)}_{2,\theta}^2$. By \cref{eq:minimum-block-damping}, the noise contribution is at most $-2\gamma\norm{O(t)}_{2,\theta}^2$. Gronwall's inequality proves the first estimate. For the split evolution, combine \cref{eq:weighted-local-H-growth} with
\begin{align}
    \norm{\left.\e^{t\mathcal D_a}\right|_{\mathsf K_a}}_{2,\theta\to 2,\theta}
    \leq\e^{-\gamma t}.
\end{align}
\end{proof}

To control the operator weight growth, we thus need the exponent $\gamma-2|J_a|\sinh(k\theta_\star)$ to be positive. Therefore define the threshold:
\begin{equation}
    \theta_\star
    :=
    \frac{1}{k}\operatorname{arsinh}\!\left(
        \frac{\gamma}{2J}
    \right).
    \label{eq:theta-star-definition}
\end{equation}

\begin{corollary}[Local weighted-$2$ contractivity]
\label{cor:local-weighted-two-contractivity}
For every $0\leq\theta\leq\theta_\star$, every $a$, and every $t\geq0$,
\begin{equation}
    \norm{\e^{t\mathcal L_a}}_{2,\theta\to 2,\theta}
    \leq1,
    \qquad
    \norm{\mathcal T_{a,t}}_{2,\theta\to 2,\theta}
    \leq1.
    \label{eq:local-weighted-two-contractivity}
\end{equation}
\end{corollary}

\begin{proof}
On $\mathsf K_a$, the exponent in \cref{eq:exact-local-weighted-two-bound,eq:split-local-weighted-two-bound} is nonpositive because $2|J_a|\sinh(k\theta)\leq2J\sinh(k\theta_\star)=\gamma$. On $\mathsf K_a^\perp$, both maps are the identity.
\end{proof}

\subsection{Pauli-$p$ regularity}
\label{supp:coefficient-regularity}

The Pauli-$p$ bound controls coefficient truncation through Lemma~\ref{lem:coefficient-pruning}. Interpolating the $p=1$ growth bound with the $p=2$ damping bound gives contractivity for $p_\star\leq p\leq2$.

\begin{lemma}[Pauli-$1$ growth of one local rotation]
\label{lem:local-pauli-one-growth}
For every $a$ and every $t\geq0$,
\begin{equation}
    \norm{\e^{t\mathcal H_a}}_{1\to1}
    \leq
    \e^{2|J_a|t}.
    \label{eq:local-pauli-one-growth}
\end{equation}
\end{lemma}

\begin{proof}
On a two-dimensional invariant subspace spanned by anticommuting partners $P,Q$,
\begin{align}
    \norm{\e^{t\mathcal H_a}(uP+vQ)}_1
    &\leq
    \left(|\cos(2J_at)|+|\sin(2J_at)|\right)(|u|+|v|)
    \\
    &\leq
    \e^{2|J_a|t}(|u|+|v|).
\end{align}
Commuting Pauli operators are fixed, and taking the maximum over all invariant subspaces proves the claim.
\end{proof}

\begin{lemma}[Pauli-$1$ bound for a local noisy evolution]
\label{lem:noisy-local-pauli-one}
For every $a$ and every $t\geq0$,
\begin{align}
    \norm{\left.\e^{t\mathcal L_a}\right|_{\mathsf K_a}}_{1\to1}
    &\leq
    \exp\!\left[t\left(2|J_a|-\gamma\right)\right],
    \label{eq:exact-local-pauli-one-bound}
    \\
    \norm{\left.\mathcal T_{a,t}\right|_{\mathsf K_a}}_{1\to1}
    &\leq
    \exp\!\left[t\left(2|J_a|-\gamma\right)\right].
    \label{eq:split-local-pauli-one-bound}
\end{align}
\end{lemma}

\begin{proof}
For the exact evolution, the Hamiltonian contribution to the upper right derivative of the coefficient $\ell_1$ norm on every active two-dimensional sector is at most $2|J_a|$ times that norm, while the noise contribution is at most $-\gamma$ times the norm. Gronwall's inequality gives the first bound. For the split evolution, combine Lemma~\ref{lem:local-pauli-one-growth} with $\norm{\left.\e^{t\mathcal D_a}\right|_{\mathsf K_a}}_{1\to1} \leq\e^{-\gamma t}$.
\end{proof}

At $\theta=0$, Lemma~\ref{lem:strict-weighted-two-local} gives
\begin{equation}
    \norm{\left.\e^{t\mathcal L_a}\right|_{\mathsf K_a}}_{2\to2}
    \leq\e^{-\gamma t},
    \qquad
    \norm{\left.\mathcal T_{a,t}\right|_{\mathsf K_a}}_{2\to2}
    \leq\e^{-\gamma t}.
    \label{eq:local-pauli-two-damping}
\end{equation}
Interpolating these estimates with \cref{eq:exact-local-pauli-one-bound,eq:split-local-pauli-one-bound} gives the following bound.

\begin{lemma}[Unweighted Pauli-$p$ bound]
\label{lem:unweighted-local-pauli-p}
Let $1\leq p\leq2$. Then
\begin{align}
    \norm{\left.\e^{t\mathcal L_a}\right|_{\mathsf K_a}}_{p\to p}
    &\leq
    \exp\!\left\{t\left[
        2|J_a|\left(\frac{2}{p}-1\right)-\gamma
    \right]\right\},
    \label{eq:exact-local-pauli-p-bound}
    \\
    \norm{\left.\mathcal T_{a,t}\right|_{\mathsf K_a}}_{p\to p}
    &\leq
    \exp\!\left\{t\left[
        2|J_a|\left(\frac{2}{p}-1\right)-\gamma
    \right]\right\}.
    \label{eq:split-local-pauli-p-bound}
\end{align}
\end{lemma}

\begin{proof}
Set $\lambda=2(1-1/p)$, so that $1/p=(1-\lambda)+\lambda/2$ and $1-\lambda=2/p-1$. The Riesz--Thorin theorem applied to the $p=1$ and $p=2$ bounds gives the result.
\end{proof}

For contractivity uniformly in $a$, we require $2J(2/p-1)\leq\gamma$. The smallest admissible exponent in $[1,2]$ is
\begin{equation}
    p_\star
    :=
    \frac{2}{1+\min\left\{1,\frac{\gamma}{2J}\right\}}.
    \label{eq:p-star-definition}
\end{equation}

\begin{corollary}[Local Pauli-$p$ contractivity]
\label{cor:unweighted-local-pauli-p}
For every $p_\star\leq p\leq2$, every $a$, and every $t\geq0$,
\begin{equation}
    \norm{\e^{t\mathcal L_a}}_{p\to p}
    \leq1,
    \qquad
    \norm{\mathcal T_{a,t}}_{p\to p}
    \leq1.
    \label{eq:unweighted-local-pauli-p-contractivity}
\end{equation}
\end{corollary}

\begin{proof}
The restrictions to $\mathsf K_a$ are contractive because $2J(2/p-1)\leq\gamma$. Both maps are the identity on $\mathsf K_a^\perp$, and the Pauli coefficient sets supporting the two subspaces are disjoint.
\end{proof}

\subsection{Regularity of the full evolution}
\label{supp:composition-over-colors}

The local Pauli-$p$ bound in \cref{eq:unweighted-local-pauli-p-contractivity} and the weighted-$2$ bound in \cref{eq:local-weighted-two-contractivity} extend first to each color and then to a full step. Submultiplicativity gives the second extension without requiring different colors to commute.

Recall that
\begin{align}
    \mathcal T_{t}^{(c)}
    =
    \e^{t\mathcal D^{(c)}}\e^{t\mathcal H^{(c)}}
    =
    \prod_{a\in\mathcal C_c}\mathcal T_{a,t},
\end{align}
while the exact color evolution is $\e^{t\mathcal L^{(c)}}$.

\begin{proposition}[Contractivity of one interaction color]
\label{prop:color-regularity}
For every color $c$ and every $t\geq0$,
\begin{align}
    \norm{\e^{t\mathcal L^{(c)}}}_{p\to p}
    &\leq1,
    &
    \norm{\mathcal T_{t}^{(c)}}_{p\to p}
    &\leq1,
    && p_\star\leq p\leq2,
    \label{eq:color-p-contractivity}
    \\
    \norm{\e^{t\mathcal L^{(c)}}}_{2,\theta\to 2,\theta}
    &\leq1,
    &
    \norm{\mathcal T_{t}^{(c)}}_{2,\theta\to 2,\theta}
    &\leq1,
    && 0\leq\theta\leq\theta_\star.
    \label{eq:color-weight-contractivity}
\end{align}
\end{proposition}

\begin{proof}
The supports $\supp(P_a)$ with $a\in\mathcal C_c$ are pairwise disjoint, so the exact and split color maps are tensor products of the corresponding local maps. The unweighted Pauli coefficient space is an $\ell_p$ space over this product basis, and its induced tensor-product norm is the product of the local induced norms. The weighted Pauli-$2$ space also factorizes because Pauli weight is additive over disjoint subsystems. The result follows from Corollaries~\ref{cor:unweighted-local-pauli-p} and~\ref{cor:local-weighted-two-contractivity}.
\end{proof}

Recall that the full step composes all color layers:
\begin{align}
    \Phi_\delta^{\mathrm{Tr}}
    =
    \mathcal T_{\delta}^{(\chi)}\cdots\mathcal T_{\delta}^{(1)},
    \qquad
    \Phi_\delta^{\mathrm{col}}
    =
    \e^{\delta\mathcal L^{(\chi)}}\cdots
    \e^{\delta\mathcal L^{(1)}}.
\end{align}

\begin{corollary}[Contractivity of one full finite step]
\label{cor:full-finite-step-contractivity}
For every $\delta\geq0$,
\begin{align}
    \norm{\Phi_\delta^{\mathrm{col}}}_{p\to p}
    &\leq1,
    &
    \norm{\Phi_\delta^{\mathrm{Tr}}}_{p\to p}
    &\leq1,
    && p_\star\leq p\leq2,
    \label{eq:full-step-p-contractivity}
    \\
    \norm{\Phi_\delta^{\mathrm{col}}}_{2,\theta\to 2,\theta}
    &\leq1,
    &
    \norm{\Phi_\delta^{\mathrm{Tr}}}_{2,\theta\to 2,\theta}
    &\leq1,
    && 0\leq\theta\leq\theta_\star.
    \label{eq:full-step-weight-contractivity}
\end{align}
\end{corollary}

\begin{proof}
Both full steps are compositions of the $\chi$ contractive color maps.
\end{proof}

Note that the same balance applies to symmetric updates by combining the two dissipative half steps.

\begin{corollary}[Contractivity of symmetric colour updates]
\label{cor:symmetric-colour-regularity}
Let
\[
    \mathcal Q_{a,h}:=
    \e^{h\mathcal D_a/2}\e^{h\mathcal H_a}\e^{h\mathcal D_a/2}.
\]
For every $h\geq0$, $p_\star\leq p\leq2$, and $0\leq\theta\leq\theta_\star$, the maps $\mathcal Q_{a,h}$, $Q_{c,h}$, and $\mathsf{S}_h$ of \cref{eq:main-symmetric-product-formula} are contractive in both the Pauli-$p$ norm and the weighted Pauli-$2$ norm.
\end{corollary}

\begin{proof}
Interpolation between \cref{eq:local-pauli-one-growth} and the Pauli-$2$ isometry of $\e^{h\mathcal H_a}$ gives
\[
    \norm{\e^{h\mathcal H_a}}_{p\to p}
    \leq\e^{2|J_a|h(2/p-1)}.
\]
On the invariant active sector $\mathsf K_a$, each dissipative half step contracts either norm by at least $\e^{-\gamma h/2}$. Together with \cref{eq:weighted-local-H-growth}, this yields
\begin{align*}
    \norm{\mathcal Q_{a,h}|_{\mathsf K_a}}_{p\to p}
    &\leq\e^{h[2|J_a|(2/p-1)-\gamma]}\leq1,\\
    \norm{\mathcal Q_{a,h}|_{\mathsf K_a}}_{2,\theta\to2,\theta}
    &\leq\e^{h[2|J_a|\sinh(k\theta)-\gamma]}\leq1.
\end{align*}
On $\mathsf K_a^\perp$ the map is the identity, so the full local map is contractive. Since the supports within a colour are disjoint, $Q_{c,h}=\prod_{a\in\mathcal C_c}\mathcal Q_{a,h}$. Submultiplicativity proves contractivity of $Q_{c,h}$ and of the palindromic composition $\mathsf{S}_h$.
\end{proof}

Taking their Lie--Trotter limit therefore gives the same two norm bounds for the exact continuous evolution, as in Lemma~\ref{lem:finite-step-to-continuous-contractivity}.

\begin{theorem}[Pauli coefficient and weight regularity]
\label{thm:supp-regularity}
Assume the model of \cref{supp:model}, with $|\supp(P_a)|\leq k$, $|J_a|\leq J$, and local damping rate at least $\gamma>0$. Let $p_\star$ and $\theta_\star$ be defined in \cref{eq:p-star-definition,eq:theta-star-definition}. Then, for every observable $O$ and every $t\geq0$,
\begin{align}
    \norm{\e^{t\mathcal L}(O)}_{p}
    &\leq
    \norm{O}_{p},
    && p_\star\leq p\leq2,
    \label{eq:supp-main-p-regularity}
    \\
    \norm{\e^{t\mathcal L}(O)}_{2,\theta}
    &\leq
    \norm{O}_{2,\theta},
    && 0\leq\theta\leq\theta_\star.
    \label{eq:supp-main-weight-regularity}
\end{align}
The split finite-step evolution obeys the same two estimates:
\begin{align}
    \norm{\Phi_\delta^{\mathrm{Tr}}(O)}_{p}
    &\leq
    \norm{O}_{p},
    && p_\star\leq p\leq2,
    \\
    \norm{\Phi_\delta^{\mathrm{Tr}}(O)}_{2,\theta}
    &\leq
    \norm{O}_{2,\theta},
    && 0\leq\theta\leq\theta_\star.
\end{align}
\end{theorem}

\begin{proof}
The split-step estimates follow from Corollary~\ref{cor:full-finite-step-contractivity}. For the exact evolution, use the color Lie--Trotter formula
\begin{align}
    \e^{t\mathcal L}
    =
    \lim_{r\to\infty}
    \left(
        \e^{(t/r)\mathcal L^{(\chi)}}\cdots\e^{(t/r)\mathcal L^{(1)}}
    \right)^r.
\end{align}
For every finite $r$, Proposition~\ref{prop:color-regularity} bounds the product by one in each of the two norms. Taking the limit proves both claims.
\end{proof}

Lemmas~\ref{lem:weight-tail} and~\ref{lem:coefficient-pruning} control weight and coefficient truncation using norms of the initial observable to establish regularity of each evolution. The following sections use this regularity to bound simulation errors.

\suppsection{Digital quantum simulation of open-system dynamics}
\label{supp:quantum}

By \cref{thm:supp-regularity}, the weighted norm at every intermediate time is at most its initial value. This lets us bound the Pauli-$2$ error within and between colors, sum it over $r$ steps, and determine the circuit depth for initially low-weight observables.
We now use the general model of \cref{supp:model}. The lower rate $\gamma$ enters the regularity scale $\theta_\star$, while the upper local rate $\bar\gamma$ controls commutators between noise and Hamiltonian evolution. The product formula is the split color step $\Phi_\delta^{\mathrm{Tr}}$ defined in \cref{eq:implemented-color-step},
\begin{equation}
    \Phi_\delta^{\mathrm{Tr}}
    =
    \mathcal T_{\delta}^{(\chi)}\cdots\mathcal T_{\delta}^{(1)},
    \qquad
    \mathcal T_{\delta}^{(c)}
    =
    \prod_{a\in\mathcal C_c}
    \e^{\delta\mathcal D_a}\e^{\delta\mathcal H_a}.
\end{equation}
Thus each color layer is a parallel product of local noisy Pauli rotations. Throughout this section, $C_k>0$ denotes a constant depending only on $k$ and may change from line to line.

\begin{theorem}[Observable-dependent digital simulation]
\label{thm:supp-quantum-simulation}
Let $0<\theta\leq\min\{1,\theta_\star\}$, let $r\in\mathbb N$, and set $\delta:=t/r$. Suppose that
\begin{equation}
    2J\delta\e^{k\theta}
    \leq
    \frac{\theta}{3}.
    \label{eq:quantum-small-step-condition}
\end{equation}
Then, for every observable $O$,
\begin{equation}
\begin{aligned}
    \norm{
        \left(
            (\Phi_\delta^{\mathrm{Tr}})^r
            -
            \e^{t\mathcal L}
        \right)(O)
    }_2
    \leq{}&
    C_k
    \frac{t^2}{r}
    \frac{\e^{2k\theta}}{\theta}
    \left(
        J^2d^2+J\bar\gamma d^2
    \right)
    \norm{O}_{2,\theta}.
\end{aligned}
    \label{eq:quantum-simulation-weighted-error}
\end{equation}
If $O$ contains only Pauli operators of weight at most $w$, then
\begin{equation}
\begin{aligned}
    \norm{
        \left(
            (\Phi_\delta^{\mathrm{Tr}})^r
            -
            \e^{t\mathcal L}
        \right)(O)
    }_2
    \leq{}&
    C_k
    \frac{t^2}{r}
    \frac{\e^{(2k+w)\theta}}{\theta}
    \left(
        J^2d^2+J\bar\gamma d^2
    \right)
    \norm{O}_2.
\end{aligned}
    \label{eq:quantum-simulation-low-weight-error}
\end{equation}
One product-formula step has sequential depth at most $\chi\leq k(d-1)+1$. Consequently, the complete simulation has depth $O_k(dr)$ and uses $O_k(mr)$ local Pauli rotations and local Pauli-diagonal noise channels.
\end{theorem}

\subsection{Warm-up: a local Pauli under unitary evolution}
\label{supp:local-pauli-trotter-warmup}

We first illustrate why the Trotter error for a local observable need not scale with the system size. For this warm-up, we consider only the unitary part of the evolution.

Recall that
\begin{align}
    H=\sum_{a=1}^m J_aP_a,
    \qquad
    \abs{J_a}\leq J,
    \qquad
    \abs{P_a}\leq k,
\end{align}
and define
\begin{align}
    \mathcal H_a(O):=i[J_aP_a,O],
    \qquad
    \mathcal H:=\sum_{a=1}^m\mathcal H_a.
\end{align}
Consider the first-order product formula
\begin{align}
    \Phi_\delta
    :=
    \mathrm e^{\delta\mathcal H_m}
    \cdots
    \mathrm e^{\delta\mathcal H_1}.
\end{align}
Its expansion to second order can be seen directly by writing
\begin{align}
\Phi_\delta
&=
\prod_{a=m}^{1} \mathrm e^{\delta\mathcal H_a}
\nonumber\
\\&=
\mathcal I
+\delta\sum_{a=1}^m \mathcal H_a
+\delta^2
\left(
\frac12\sum_{a=1}^m \mathcal H_a^2
+
\sum_{1\leq a<b\leq m}
\mathcal H_b\mathcal H_a
\right)
+O(\delta^3).
\label{eq:trotter-second-order-expansion}
\end{align}
On the other hand, since $\mathcal H=\sum_a\mathcal H_a$,
\begin{align}
\mathrm e^{\delta\mathcal H}
&=
\mathcal I
+\delta\sum_{a=1}^m\mathcal H_a
+\frac{\delta^2}{2}
\left(\sum_{a=1}^m\mathcal H_a\right)^2
+O(\delta^3)
\nonumber\
\\&=
\mathcal I
+\delta\sum_{a=1}^m\mathcal H_a
+\delta^2
\left[
\frac12\sum_{a=1}^m\mathcal H_a^2
+
\frac12\sum_{1\leq a<b\leq m}
\left(
\mathcal H_a\mathcal H_b
+
\mathcal H_b\mathcal H_a
\right)
\right]
+O(\delta^3).
\label{eq:exact-second-order-expansion}
\end{align}
The two evolutions therefore agree to first order, while at second order their difference is determined by the ordering of distinct local generators:
\begin{align}
\Phi_\delta
&=
\mathrm e^{\delta\mathcal H}
+
\delta^2\mathcal G
+
O(\delta^3),
\label{eq:trotter-local-error-expansion}
\end{align}
where
\begin{align}
\mathcal G
&:=
\frac12
\sum_{1\leq a<b\leq m}
[\mathcal H_b,\mathcal H_a].
\label{eq:trotter-error-generator}
\end{align}
Thus the leading product-formula error is entirely generated by pairs of noncommuting local terms.

Let \(P\in\mathcal P_n\) be a Pauli string of weight
\begin{align}
    \abs{P}=w.
\end{align}
At most \(dw\) Hamiltonian terms intersect \(\operatorname{supp}(P)\). Moreover, each such term overlaps the support of at most \(kd\) other Hamiltonian terms. Since
\begin{align}
    [\mathcal H_b,\mathcal H_a](P)=0
\end{align}
unless the supports of \(P_a\) and \(P_b\) overlap and at least one of them intersects \(\operatorname{supp}(P)\), at most \(C_kwd^2\) commutators contribute.

For every pair \(a,b\),
\begin{align}
    \norm{
        [\mathcal H_b,\mathcal H_a](P)
    }_2
    \leq
    8J^2\norm{P}_2
    =
    8J^2.
\end{align}
It follows that
\begin{align}
    \norm{\mathcal G(P)}_2
    \leq
    C_kJ^2d^2w.
\end{align}
Consequently,
\begin{align}
    \norm{
        \left(
            \Phi_\delta
            -
            \mathrm e^{\delta\mathcal H}
        \right)(P)
    }_2
    \leq
    C_kJ^2d^2w\,\delta^2
    +
    O(\delta^3).
\end{align}

Thus, for a local Pauli observable, \(w=O(1)\), the leading one-step Trotter error is independent of the total number of qubits and Hamiltonian terms. The general proof below extends this intuition to noisy evolution, arbitrary observables, and finite step sizes. 

\subsection{Error for a single interaction color}
\label{supp:single-color-error}

Within one color, splitting Hamiltonian evolution and noise produces an error governed by their commutator. The next two lemmas bound this commutator and the rotation in weighted norms. Weight loss means lowering the norm parameter from $\alpha$ to $\beta$. The factor $\e^{-(\alpha-\beta)|P|}$ offsets the growth of the absolute row and column sums with Pauli weight.
We first record two absolute-norm estimates. We will use
\begin{equation}
    \sup_{s\geq0}s\e^{-\eta s}
    \leq
    \frac{1}{\eta},
    \qquad
    \eta>0.
    \label{eq:weight-loss-elementary-bound}
\end{equation}

\begin{lemma}[Parallel rotation with weight loss]
\label{lem:parallel-color-weight-loss}
Let $0\leq\beta<\alpha$, and set $\eta:=\alpha-\beta$. If
\begin{equation}
    2Jt\e^{k\alpha}
    \leq
    \eta,
    \label{eq:parallel-color-smallness}
\end{equation}
then, for every color $c$,
\begin{equation}
    \norm{
        \e^{t\mathcal H^{(c)}}
    }_{2,\alpha\to 2,\beta}^{\mathrm{abs}}
    \leq
    1.
    \label{eq:parallel-color-weight-loss}
\end{equation}
\end{lemma}

\begin{proof}
Fix an input Pauli operator $P$ of weight $s$. Because the supports inside $\mathcal C_c$ are pairwise disjoint, at most $s$ Hamiltonian terms of color $c$ can act nontrivially on $P$.

For one such term,
\begin{align}
    \e^{t\mathcal H_a}(R)
    =
    \cos(2J_at)R
    +
    \xi\sin(2J_at)\widehat{P_aR}.
\end{align}
The nontrivial branch has absolute coefficient at most $2Jt$ and changes Pauli weight by at most $k$. Expanding the parallel product therefore gives
\begin{align}
    \sum_Q
    \abs{(\e^{t\mathcal H^{(c)}})_{Q,P}}
    \e^{\beta|Q|-\alpha|P|}
    &\leq
    \e^{-\eta s}
    \left(
        1+2Jt\e^{k\beta}
    \right)^s
    \nonumber\\
    &\leq
    \exp\left[
        -\eta s
        +
        2Jt\e^{k\alpha}s
    \right]
    \leq
    1.
    \label{eq:parallel-color-column-bound}
\end{align}

The row estimate is analogous. Fix an output Pauli operator $Q$ of weight $q$. Every color support on which the nontrivial branch is taken must intersect $\supp(Q)$, so at most $q$ color supports can contribute. If $h$ nontrivial branches are taken, the corresponding input has weight at least $q-kh$. Hence
\begin{align}
    \sum_P
    \abs{(\e^{t\mathcal H^{(c)}})_{Q,P}}
    \e^{\beta|Q|-\alpha|P|}
    &\leq
    \e^{-\eta q}
    \left(
        1+2Jt\e^{k\alpha}
    \right)^q
    \nonumber\\
    &\leq
    1.
    \label{eq:parallel-color-row-bound}
\end{align}
Schur's test proves the claim.
\end{proof}

For a fixed color, define
\begin{equation}
    \mathcal G_c
    :=
    [\mathcal D^{(c)},\mathcal H^{(c)}].
    \label{eq:single-color-mixed-commutator}
\end{equation}
Because distinct labels in the same color have disjoint supports,
\begin{equation}
    \mathcal G_c
    =
    \sum_{a\in\mathcal C_c}
    [\mathcal D_a,\mathcal H_a].
    \label{eq:single-color-commutator-sum}
\end{equation}

\begin{lemma}[Single-color noise--Hamiltonian commutator]
\label{lem:single-color-mixed-commutator}
Let $0\leq\beta<\alpha$ and set $\eta:=\alpha-\beta$. Then
\begin{equation}
    \norm{
        [\mathcal D^{(c)},\mathcal H^{(c)}]
    }_{2,\alpha\to 2,\beta}^{\mathrm{abs}}
    \leq
    C_kJ\bar\gamma
    \frac{\e^{k\alpha}}{\eta}.
    \label{eq:single-color-mixed-commutator-bound}
\end{equation}
The bound is independent of the total number $m$ of local blocks and of the size of the color class.
\end{lemma}

\begin{proof}
Suppose that $P_a$ anticommutes with $P$, and let $Q=\widehat{P_aP}$. Then \cref{eq:local-pauli-pair-generator,eq:local-pauli-diagonal-noise} gives
\begin{equation}
    [\mathcal D_a,\mathcal H_a](P)
    =
    2J_a\xi_{a,P}
    \bigl(
        \lambda_a(P)-\lambda_a(Q)
    \bigr)Q.
    \label{eq:local-mixed-commutator-action}
\end{equation}
Since $0\leq\lambda_a(\cdot)\leq\bar\gamma$, every nonzero local transfer coefficient is bounded by $2J\bar\gamma$.

Fix an input Pauli operator $P$ of weight $s$. Only labels for which $\mathcal H_a(P)\neq0$ can contribute. Such a term intersects $\supp(P)$, and the terms within one color are disjoint, so there are at most $s$ contributing labels. Every output has weight at most $s+k$. Therefore
\begin{align}
    \sum_Q
    \abs{(\mathcal G_c)_{Q,P}}
    \e^{\beta|Q|-\alpha|P|}
    &\leq
    C_kJ\bar\gamma
    s\e^{k\beta}\e^{-\eta s}
    \nonumber\\
    &\leq
    C_kJ\bar\gamma
    \frac{\e^{k\alpha}}{\eta}.
    \label{eq:single-color-commutator-column}
\end{align}
The same argument read backwards gives the corresponding row bound. Schur's test completes the proof.
\end{proof}

Lemmas~\ref{lem:parallel-color-weight-loss} and~\ref{lem:single-color-mixed-commutator} bound the rotation and commutator in the Lie--Trotter identity below. Integrating over $0\leq u\leq s\leq\delta$ gives the factor $\delta^2/2$.
Recall that the split color evolution is
\begin{align}
    \mathcal T_{t}^{(c)}
    =
    \e^{t\mathcal D^{(c)}}\e^{t\mathcal H^{(c)}},
\end{align}
whereas the exact color evolution is $\e^{t\mathcal L^{(c)}}$. Applying Lemma~\ref{lem:lie-trotter-local-error} with
\begin{align}
    A=\mathcal H^{(c)},
    \qquad
    B=\mathcal D^{(c)},
\end{align}
gives
\begin{align}
    \mathcal T_{t}^{(c)}-\e^{t\mathcal L^{(c)}}
    =
    \int_0^t\int_0^s
    &
    \e^{(t-s)\mathcal L^{(c)}}
    \e^{(s-u)\mathcal D^{(c)}}
    [\mathcal D^{(c)},\mathcal H^{(c)}]
    \nonumber\\
    &\times
    \e^{u\mathcal D^{(c)}}
    \e^{s\mathcal H^{(c)}}
    \,\mathrm du\,\mathrm ds.
    \label{eq:single-color-lie-trotter-error}
\end{align}

\begin{lemma}[Error within one interaction color]
\label{lem:single-color-splitting-error}
Let $0<\theta\leq\theta_\star$ and suppose that
\begin{equation}
    2J\delta\e^{k\theta}
    \leq
    \frac{\theta}{3}.
    \label{eq:single-color-small-step}
\end{equation}
Then
\begin{equation}
    \norm{
        \mathcal T_{\delta}^{(c)}
        -
        \e^{\delta\mathcal L^{(c)}}
    }_{2,\theta\to 2,0}
    \leq
    C_k\delta^2
    \frac{\e^{k\theta}}{\theta}
    J\bar\gamma.
    \label{eq:single-color-splitting-error}
\end{equation}
\end{lemma}

\begin{proof}
Use the intermediate weights
\begin{align}
    \theta,
    \qquad
    \frac{2\theta}{3},
    \qquad
    \frac{\theta}{3},
    \qquad
    0.
\end{align}
By Lemma~\ref{lem:parallel-color-weight-loss} and the small-step condition,
\begin{equation}
    \norm{
        \e^{s\mathcal H^{(c)}}
    }_{2,\theta\to 2,2\theta/3}^{\mathrm{abs}}
    \leq
    1
\end{equation}
for every $0\leq s\leq\delta$. The pure-noise maps are diagonal contractions at every weight. Moreover, Lemma~\ref{lem:single-color-mixed-commutator} with $\alpha=2\theta/3$ and $\beta=\theta/3$ gives
\begin{equation}
    \norm{
        [\mathcal D^{(c)},\mathcal H^{(c)}]
    }_{2,2\theta/3\to 2,\theta/3}^{\mathrm{abs}}
    \leq
    C_kJ\bar\gamma
    \frac{\e^{k\theta}}{\theta}.
\end{equation}
Finally, weighted contractivity from Proposition~\ref{prop:color-regularity} implies
\begin{equation}
    \norm{
        \e^{(\delta-s)\mathcal L^{(c)}}
    }_{2,\theta/3\to 2,0}
    \leq
    1.
\end{equation}
Every integrand in \cref{eq:single-color-lie-trotter-error} is therefore bounded by $C_kJ\bar\gamma\e^{k\theta}/\theta$. Since
\begin{align}
    \int_0^\delta\int_0^s
    \mathrm du\,\mathrm ds
    =
    \frac{\delta^2}{2},
\end{align}
the result follows.
\end{proof}

\subsection{Telescoping over colors and time steps}
\label{supp:quantum-telescoping}

A complete step has two sources of error. We first compare $\Phi_\delta^{\mathrm{Tr}}$ with $\Phi_\delta^{\mathrm{col}}$ to account for splitting noise and Hamiltonian evolution within colors. We then compare $\Phi_\delta^{\mathrm{col}}$ with $\e^{\delta\mathcal L}$ to account for ordering different colors.
Recall that
\begin{align}
    \Phi_\delta^{\mathrm{Tr}}
    =
    \mathcal T_{\delta}^{(\chi)}\cdots\mathcal T_{\delta}^{(1)},
    \qquad
    \Phi_\delta^{\mathrm{col}}
    =
    \e^{\delta\mathcal L^{(\chi)}}\cdots
    \e^{\delta\mathcal L^{(1)}}.
\end{align}
We first accumulate the splitting error within the colors.

\begin{proposition}[Splitting error over all colors]
\label{prop:all-color-splitting-error}
Under the assumptions of Lemma~\ref{lem:single-color-splitting-error},
\begin{align}
    \norm{
        \Phi_\delta^{\mathrm{Tr}}
        -
        \Phi_\delta^{\mathrm{col}}
    }_{2,\theta\to 2,0}
    \leq
    C_k\delta^2
    \frac{\e^{k\theta}}{\theta}
    J\bar\gamma d.
    \label{eq:all-color-splitting-error}
\end{align}
\end{proposition}

\begin{proof}
The telescoping identity gives
\begin{align}
    \Phi_\delta^{\mathrm{Tr}}
    -
    \Phi_\delta^{\mathrm{col}}
    =
    \sum_{c=1}^{\chi}
    &
    \mathcal T_{\delta}^{(\chi)}
    \cdots
    \mathcal T_{\delta}^{(c+1)}
    \nonumber\\
    &\times
    \left(
        \mathcal T_{\delta}^{(c)}
        -
        \e^{\delta\mathcal L^{(c)}}
    \right)
    \e^{\delta\mathcal L^{(c-1)}}
    \cdots
    \e^{\delta\mathcal L^{(1)}}.
    \label{eq:all-color-splitting-telescope}
\end{align}
The maps to the right of the difference are contractive at weight $\theta$, while those to the left are contractive at weight zero. ApplyingLemma~\ref{lem:single-color-splitting-error} to every summand gives
\begin{align}
    \norm{
        \Phi_\delta^{\mathrm{Tr}}
        -
        \Phi_\delta^{\mathrm{col}}
    }_{2,\theta\to 2,0}
    \leq
    C_k\chi\delta^2
    \frac{\e^{k\theta}}{\theta}
    J\bar\gamma.
\end{align}
The claim follows from $\chi\leq k(d-1)+1\leq C_kd$.
\end{proof}

We next compare the product of exact color evolutions with the exact global evolution.

\begin{lemma}[Commutator of two color generators]
\label{lem:pairwise-color-commutator}
Let $c\neq c'$. For every $\theta>0$,
\begin{equation}
    \norm{
        [\mathcal L^{(c')},\mathcal L^{(c)}]
    }_{2,\theta\to 2,0}^{\mathrm{abs}}
    \leq
    C_k
    \frac{\e^{2k\theta}}{\theta}
    \left(
        J^2+J\bar\gamma
    \right).
    \label{eq:pairwise-color-commutator}
\end{equation}
The estimate is independent of $m$.
\end{lemma}

\begin{proof}
Since all noise generators are diagonal in the Pauli basis,
\begin{align}
    [\mathcal L^{(c')},\mathcal L^{(c)}]
    ={}
    [\mathcal H^{(c')},\mathcal H^{(c)}]
    +
    [\mathcal D^{(c')},\mathcal H^{(c)}]
    +
    [\mathcal H^{(c')},\mathcal D^{(c)}].
    \label{eq:pairwise-color-commutator-decomposition}
\end{align}

For the Hamiltonian--Hamiltonian part, the previous counting argument is unchanged. For an input Pauli of weight $s$, at most $C_ks$ pairs of terms from the two colors can contribute, each with transfer coefficient $O(J^2)$ and weight change at most $2k$. Hence
\begin{equation}
    \norm{
        [\mathcal H^{(c')},\mathcal H^{(c)}]
    }_{2,\theta\to 2,0}^{\mathrm{abs}}
    \leq
    C_k\frac{J^2\e^{2k\theta}}{\theta}.
    \label{eq:pairwise-HH-bound}
\end{equation}

Consider now $[\mathcal D^{(c')},\mathcal H^{(c)}]$. If $\mathcal H_a(P)=2J_a\xi_{a,P}Q$, then
\begin{equation}
    [\mathcal D_b,\mathcal H_a](P)
    =
    2J_a\xi_{a,P}
    \bigl(
        \lambda_b(P)-\lambda_b(Q)
    \bigr)Q.
    \label{eq:cross-color-mixed-action}
\end{equation}
This vanishes whenever $\supp(P_a)\cap\supp(P_b)=\varnothing$: in that case $Q$ differs from $P$ only outside $\supp(P_b)$, while $\mathcal D_b$ is supported on $\supp(P_b)$, and hence $\lambda_b(P)=\lambda_b(Q)$. Otherwise its magnitude is at most $2J\bar\gamma$.

For fixed $a\in\mathcal C_c$, at most $k$ terms $b\in\mathcal C_{c'}$ can intersect $\supp(P_a)$, because the $\supp(P_b)$ are disjoint within color $c'$ and $|\supp(P_a)|\leq k$. Moreover, only labels $a$ with $\mathcal H_a(P)\neq0$ contribute, and there are at most $|P|$ of them because the $\supp(P_a)$ are disjoint within color $c$. Thus the unweighted column sum is bounded by $C_kJ\bar\gamma|P|$. The same reversed counting gives the corresponding row bound, and every output differs in weight by at most $k$. Schur's test therefore yields
\begin{equation}
    \norm{
        [\mathcal D^{(c')},\mathcal H^{(c)}]
    }_{2,\theta\to 2,0}^{\mathrm{abs}}
    \leq
    C_k\frac{J\bar\gamma\e^{k\theta}}{\theta}.
    \label{eq:pairwise-DH-bound}
\end{equation}
The same estimate holds for $[\mathcal H^{(c')},\mathcal D^{(c)}]$. Combining the three terms proves the claim.
\end{proof}

Recall that
\begin{align}
    \Phi_\delta^{\mathrm{col}}
    =
    \e^{\delta\mathcal L^{(\chi)}}\cdots
    \e^{\delta\mathcal L^{(1)}}.
\end{align}

\begin{proposition}[Error between exact color layers]
\label{prop:exact-color-product-error}
Let $0<\theta\leq\theta_\star$. Then
\begin{equation}
    \norm{
        \Phi_\delta^{\mathrm{col}}
        -
        \e^{\delta\mathcal L}
    }_{2,\theta\to 2,0}
    \leq
    C_k\delta^2
    \frac{\e^{2k\theta}}{\theta}
    \left(
        J^2d^2+J\bar\gamma d^2
    \right).
    \label{eq:exact-color-product-error}
\end{equation}
\end{proposition}

\begin{proof}
For this proof only, define the partial generators
\begin{align}
    \mathcal L_{\leq c}
    :=
    \sum_{j=1}^{c}\mathcal L^{(j)},
    \qquad
    c=1,\ldots,\chi.
\end{align}
Since \(\mathcal L_{\leq\chi}=\mathcal L\), telescoping over the number of included colors gives
\begin{align}
    \Phi_\delta^{\mathrm{col}}-\e^{\delta\mathcal L}
    =
    \sum_{c=2}^{\chi}
    &
    \e^{\delta\mathcal L^{(\chi)}}\cdots
    \e^{\delta\mathcal L^{(c+1)}}
    \nonumber\\
    &\times
    \left(
        \e^{\delta\mathcal L^{(c)}}
        \e^{\delta\mathcal L_{\leq c-1}}
        -
        \e^{\delta\mathcal L_{\leq c}}
    \right).
    \label{eq:exact-color-product-telescope}
\end{align}

Apply Lemma~\ref{lem:lie-trotter-local-error} to the term in parentheses with
\begin{align}
    A=\mathcal L_{\leq c-1},
    \qquad
    B=\mathcal L^{(c)}.
\end{align}
Using
\begin{align}
    [\mathcal L^{(c)},\mathcal L_{\leq c-1}]
    =
    \sum_{j=1}^{c-1}
    [\mathcal L^{(c)},\mathcal L^{(j)}],
\end{align}
we obtain
\begin{align}
    &
    \e^{\delta\mathcal L^{(c)}}
    \e^{\delta\mathcal L_{\leq c-1}}
    -
    \e^{\delta\mathcal L_{\leq c}}
    \nonumber\\
    &\quad=
    \sum_{j=1}^{c-1}
    \int_0^\delta\int_0^s
    \e^{(\delta-s)\mathcal L_{\leq c}}
    \e^{(s-u)\mathcal L^{(c)}}
    [\mathcal L^{(c)},\mathcal L^{(j)}]
    \e^{u\mathcal L^{(c)}}
    \e^{s\mathcal L_{\leq c-1}}
    \,\mathrm du\,\mathrm ds.
    \label{eq:partial-color-lie-trotter-error}
\end{align}

The proof of \cref{thm:supp-regularity} applies verbatim to every partial sum \(\mathcal L_{\leq c}\), which is obtained by retaining only the first \(c\) color generators. Hence its exact evolution is contractive at every weight between zero and \(\theta_\star\). In each integrand, all maps to the right of the commutator are therefore contractive at weight \(\theta\), while all maps to its left are contractive at weight zero. Combining \cref{eq:exact-color-product-telescope,eq:partial-color-lie-trotter-error} gives
\begin{align}
    \norm{
        \Phi_\delta^{\mathrm{col}}
        -
        \e^{\delta\mathcal L}
    }_{2,\theta\to 2,0}
    \leq
    \sum_{c=2}^{\chi}
    \sum_{j=1}^{c-1}
    \int_0^\delta\left(\int_0^s
    \norm{
        [\mathcal L^{(c)},\mathcal L^{(j)}]
    }_{2,\theta\to 2,0}^{\mathrm{abs}}
    \,\mathrm du\right)\,\mathrm ds.
\end{align}
By Lemma~\ref{lem:pairwise-color-commutator},
\begin{align}
    \norm{
        \Phi_\delta^{\mathrm{col}}
        -
        \e^{\delta\mathcal L}
    }_{2,\theta\to 2,0}
    \leq
    C_k\delta^2
    \binom{\chi}{2}
    \frac{\e^{2k\theta}}{\theta}
    \left(
        J^2+J\bar\gamma
    \right).
\end{align}
Using \(\chi\leq C_kd\) proves the claim.
\end{proof}

Recall that
\begin{align}
    \Phi_\delta^{\mathrm{Tr}}
    =
    \mathcal T_{\delta}^{(\chi)}\cdots\mathcal T_{\delta}^{(1)}.
\end{align}

Adding the bounds in Propositions~\ref{prop:all-color-splitting-error} and~\ref{prop:exact-color-product-error} gives the error of one complete step.

\begin{proposition}[Finite-step error]
\label{prop:quantum-finite-step-error}
Let $0<\theta\leq\min\{1,\theta_\star\}$ and suppose that \cref{eq:quantum-small-step-condition} holds. Then
\begin{equation}
    \norm{
        \Phi_\delta^{\mathrm{Tr}}
        -
        \e^{\delta\mathcal L}
    }_{2,\theta\to 2,0}
    \leq
    C_k\delta^2
    \frac{\e^{2k\theta}}{\theta}
    \left(
        J^2d^2+J\bar\gamma d^2
    \right).
    \label{eq:quantum-finite-step-error}
\end{equation}
\end{proposition}

\begin{proof}
Insert the product of exact color maps:
\begin{align}
    \Phi_\delta^{\mathrm{Tr}}-\e^{\delta\mathcal L}
    =
    \left(
        \Phi_\delta^{\mathrm{Tr}}
        -
        \Phi_\delta^{\mathrm{col}}
    \right)
    +
    \left(
        \Phi_\delta^{\mathrm{col}}
        -
        \e^{\delta\mathcal L}
    \right).
\end{align}
Apply Propositions~\ref{prop:all-color-splitting-error} and~\ref{prop:exact-color-product-error}. Since $d\geq1$ and $\e^{k\theta}\leq\e^{2k\theta}$, the first contribution is absorbed into the second after changing $C_k$.
\end{proof}

By Theorem~\ref{thm:supp-regularity}, evolution before each one-step error is contractive in the weighted Pauli-$2$ norm, and evolution after it is contractive in the unweighted norm. The errors therefore add to at most $r$ times the one-step bound, giving the factor $r\delta^2=t^2/r$.

\begin{corollary}[Many-step error]
\label{cor:quantum-many-step-error}
Let $t=r\delta$. Under the assumptions of Proposition~\ref{prop:quantum-finite-step-error},
\begin{equation}
    \norm{
        (\Phi_\delta^{\mathrm{Tr}})^r
        -
        \e^{t\mathcal L}
    }_{2,\theta\to 2,0}
    \leq
    C_k t\delta
    \frac{\e^{2k\theta}}{\theta}
    \left(
        J^2d^2+J\bar\gamma d^2
    \right).
    \label{eq:quantum-many-step-error}
\end{equation}
\end{corollary}

\begin{proof}
Use the telescoping identity
\begin{equation}
    (\Phi_\delta^{\mathrm{Tr}})^r
    -
    \e^{r\delta\mathcal L}
    =
    \sum_{j=0}^{r-1}
    (\Phi_\delta^{\mathrm{Tr}})^{r-1-j}
    \left(
        \Phi_\delta^{\mathrm{Tr}}
        -
        \e^{\delta\mathcal L}
    \right)
    \e^{j\delta\mathcal L}.
    \label{eq:quantum-time-telescoping}
\end{equation}
The left factor is contractive in the unweighted Pauli-$2$ norm, while the right factor is contractive in the weighted Pauli-$2$ norm. Thus
\begin{align}
    \norm{
        (\Phi_\delta^{\mathrm{Tr}})^r
        -
        \e^{t\mathcal L}
    }_{2,\theta\to 2,0}
    \leq
    r
    \norm{
        \Phi_\delta^{\mathrm{Tr}}
        -
        \e^{\delta\mathcal L}
    }_{2,\theta\to 2,0}.
\end{align}
Apply Proposition~\ref{prop:quantum-finite-step-error} and use $r\delta=t$.
\end{proof}

\subsection{Low-weight observable-dependent error}
\label{supp:quantum-low-weight}

The induced estimate becomes especially useful for low-weight observables. If
\begin{align}
    \Pi_{\leq w}(O)=O,
\end{align}
then
\begin{equation}
    \norm{O}_{2,\theta}
    \leq
    \e^{\theta w}\norm{O}_2.
    \label{eq:low-weight-weighted-norm}
\end{equation}
Consequently, Corollary~\ref{cor:quantum-many-step-error} gives
\begin{equation}
\begin{aligned}
    \norm{
        \left(
            (\Phi_\delta^{\mathrm{Tr}})^r
            -
            \e^{t\mathcal L}
        \right)(O)
    }_2
    \leq{}&
    C_k
    \frac{t^2}{r}
    \frac{\e^{(2k+w)\theta}}{\theta}
    \left(
        J^2d^2+J\bar\gamma d^2
    \right)
    \norm{O}_2.
\end{aligned}
    \label{eq:low-weight-observable-trotter-error}
\end{equation}

For a weight-$w$ observable, define
\begin{equation}
    \theta_w
    :=
    \min\left\{
        \theta_\star,
        1,
        \frac{1}{2k+w}
    \right\}.
    \label{eq:quantum-optimal-weight-parameter}
\end{equation}
Then
\begin{align}
    \e^{(2k+w)\theta_w}
    \leq
    \e,
\end{align}
and it is sufficient to take
\begin{equation}
    r
    \geq
    \max\left\{
        \frac{6Jt\e^{k\theta_w}}{\theta_w},
        \,
        C_k
        \frac{t^2}{\varepsilon\theta_w}
        \left(
            J^2d^2+J\bar\gamma d^2
        \right)
    \right\}
    \label{eq:quantum-required-number-steps}
\end{equation}
to guarantee
\begin{equation}
    \norm{
        \left(
            (\Phi_{t/r}^{\mathrm{Tr}})^r
            -
            \e^{t\mathcal L}
        \right)(O)
    }_2
    \leq
    \varepsilon\norm{O}_2.
    \label{eq:quantum-relative-error}
\end{equation}
The first term in \cref{eq:quantum-required-number-steps} enforces the small-step condition, while the second controls the accumulated product-formula error.

\subsection{Gate-count and depth bounds}
\label{supp:quantum-resources}

The step count in \cref{eq:quantum-required-number-steps} determines the gate count and circuit depth. Interactions within each color have disjoint supports and can run in parallel.

For every label $a$, the local split map is
\begin{equation}
    \mathcal T_{a,\delta}
    =
    \e^{\delta\mathcal D_a}
    \e^{\delta\mathcal H_a}.
    \label{eq:local-split-implementation}
\end{equation}
It consists of one Pauli rotation and one Pauli-diagonal dissipative channel, both supported on at most $k$ qubits.

One product-formula step contains one implementation of every local Lindbladian term $\mathcal L_a$. Its total primitive-operation count is therefore
\begin{equation}
    O_k(m).
    \label{eq:quantum-one-step-gate-count}
\end{equation}
Within a fixed color, all supports are disjoint, so the corresponding local maps can be executed in parallel. The sequential depth of one step is
\begin{equation}
    O_k(\chi)
    =
    O_k(d).
    \label{eq:quantum-one-step-depth}
\end{equation}
Over $r$ product-formula steps, the total primitive-operation count and depth are
\begin{equation}
    O_k(mr)
    \qquad\text{and}\qquad
    O_k(dr),
    \label{eq:quantum-total-resources}
\end{equation}
respectively.

In particular, for fixed $k$, $d$, $w$, $\gamma/J$, $t$, and target relative error $\varepsilon$, the required number of sequential product-formula steps is independent of both $m$ and $n$.

\subsection{Proof of the first-order quantum simulation theorem}
\label{supp:proof-quantum-simulation}

\begin{proof}[Proof of \cref{thm:supp-quantum-simulation}]
Equation \eqref{eq:quantum-simulation-weighted-error} follows by applying Corollary~\ref{cor:quantum-many-step-error} to $O$ and using $\delta=t/r$. If $\Pi_{\leq w}(O)=O$, then \cref{eq:low-weight-weighted-norm} gives \cref{eq:quantum-simulation-low-weight-error}.

The depth and primitive-operation counts follow from \cref{eq:local-split-implementation,eq:quantum-total-resources}.
\end{proof}

\suppsection{Quantum simulation via Richardson extrapolation}
\label{supp:quantum-richardson}

Richardson extrapolation improves an approximation by combining evaluations at different step sizes. It has been applied to estimates of time-evolved observables under Hamiltonian~\cite{low2019well, rendon2024improved, watson2024randomly, watson2025exponentially}, and Lindbladian simulation~\cite{wang2026lindbladian}. Here we combine this method with the weighted-norm bounds established above to approximate $\e^{t\mathcal L}(O)$ using circuits of maximum depth $O(\log^2(\e/\varepsilon))$ at fixed physical parameters. The following theorem states the resulting guarantee.

\begin{theorem}[Lindbladian evolution guarantees with Richardson extrapolation]
\label{thm:rich-lind-maximum-depth}
Assume the model of \cref{supp:model} with $J>0$ and the local implementation model of \cref{supp:quantum-resources}. Let $O=O^\dagger$ be a linear combination of Pauli strings of weight at most $w$, let $t>0$, and let $0<\varepsilon\leq1$. With $\theta_w$ from \cref{eq:quantum-optimal-weight-parameter} and the color count $\chi$ from \cref{eq:number-colors}, choose the number of circuits $q$ as
\begin{align}
    q:=\left\lceil\frac{\log(4096/\varepsilon)}{\log4}\right\rceil.
    \label{eq:rich-lind-extrapolation-choices}
\end{align}
Take $N_j$ and $\omega_j$ from \cref{eq:rich-lind-relative-step-counts,eq:rich-lind-extrapolation-counts-weights}. The extrapolated observable $O^{[q]}$ in \cref{eq:output}, formed from the $q$ circuits $\mathsf{S}_{t/N_j}^{N_j}$ with the symmetric step in \cref{eq:rich-lind-symmetric-step}, satisfies
\begin{align}
    \norm{O^{[q]}-\e^{t\mathcal L}(O)}_2
    \leq\varepsilon\norm{O}_2.
    \label{eq:rich-lind-extrapolated-observable-error}
\end{align}
The maximum depth $D_{\max}$ among these circuits obeys
\begin{align}
    D_{\max}
    \leq C_k\left[d+
      \frac{d^{5/2}(J+\bar\gamma)^{3/2}t^{3/2}}{\theta_w^{3/2}}\right]
      \log^2\!\left(\frac{\e}{\varepsilon}\right),
    \label{eq:rich-lind-maximum-depth-bound}
\end{align}
where $C_k$ depends only on $k$. The coefficients satisfy \cref{eq:rich-lind-extrapolation-coefficient-bound}, and the weighted combination is performed classically on the separate circuit outputs.
\end{theorem}

We first introduce Richardson extrapolation through an elementary example. We then construct the Lindbladian expansion in \cref{supp:richardson-lindbladian-construction}, bound its coefficients and remainder in \cref{supp:richardson-expansion-bounds,supp:richardson-remainder-bound}, and prove \cref{thm:rich-lind-maximum-depth} using the explicit extrapolation coefficients in \cref{supp:richardson-extrapolation-proof}.

\subsection{Introduction to Richardson extrapolation}
\label{supp:richardson-introduction}

Suppose that a quantity $F_0$ is approximated by a function $F(h)$, where $h>0$ is a tunable step size. Richardson extrapolation uses the dependence of the error on $h$ to improve this approximation. More precisely, suppose that, for a fixed integer $k\geq1$ and sufficiently small $h$,
\begin{align}
    F(h)=F_0+\sum_{\ell=1}^{k-1}a_\ell h^\ell+R_k(h),
    \qquad
    |R_k(h)|\leq C_k h^k.
    \label{eq:rich-intro-expansion}
\end{align}
Here the coefficients $a_\ell$ are independent of $h$, as is the remainder bound $C_k$. The integer $k$ specifies the order to which we will improve the approximation.

To see how the cancellation works, consider the leading error for $k\geq2$. At step size $h$ it is $a_1h$, while at $h/2$ it is $a_1h/2$. Thus, halving the step size halves the leading error, and
\begin{align}
    2F(h/2)-F(h)=F_0+O(h^2).
    \label{eq:rich-intro-first-cancellation}
\end{align}
The coefficients $2$ and $-1$ preserve $F_0$ and cancel the linear term without needing to know $a_1$. The error is therefore reduced from $O(h)$ to $O(h^2)$.

More generally, we combine the $k$ values $F(h),F(h/2),\ldots,F(h/k)$. Replacing $h$ by $h/j$ divides the coefficient of $h^\ell$ by $j^\ell$. Hence, to preserve $F_0$ and cancel the first $k-1$ error terms, the weights must satisfy
\begin{align}
    \sum_{j=1}^k\omega_j=1,
    \qquad
    \sum_{j=1}^k\frac{\omega_j}{j^\ell}=0
    \quad (\ell=1,\ldots,k-1).
    \label{eq:rich-intro-moments}
\end{align}
The Lagrange interpolation argument in the proof of \cref{lem:rich-lind-extrapolation-counts}, specifically \cref{eq:rich-lind-polynomial-extrapolation}, applies to the distinct points $1,1/2,\ldots,1/k$. Evaluating the interpolation polynomials at zero gives the explicit weights
\begin{align}
    \omega_j
    &=\prod_{\substack{1\leq m\leq k\\m\neq j}}
      \frac{-1/m}{1/j-1/m}
      =\prod_{\substack{1\leq m\leq k\\m\neq j}}\frac{j}{j-m}\\
    &=\frac{(-1)^{k-j}j^{k-1}}{(j-1)!(k-j)!},
    \qquad j=1,\ldots,k.
    \label{eq:rich-intro-integer-weights}
\end{align}
Interpolating the monomials $1,v,\ldots,v^{k-1}$ and evaluating at zero gives the moment conditions in \cref{eq:rich-intro-moments}. Substituting \cref{eq:rich-intro-expansion} into the weighted sum leaves only the remainders. We therefore obtain
\begin{align}
    \left|\sum_{j=1}^k\omega_j F(h/j)-F_0\right|
    \leq C_k h^k\sum_{j=1}^k\frac{|\omega_j|}{j^k}
    =O(h^k).
    \label{eq:rich-intro-remainder}
\end{align}
Thus $k$ evaluations give an order-$k$ approximation by cancelling $k-1$ error terms. This conclusion only needs the finite expansion and remainder bound in \cref{eq:rich-intro-expansion}. The error statement holds as $h\to0$ for fixed $k$, with a constant that can depend on $k$.

We illustrate the construction by accelerating the convergence of $(1+h)^{1/h}$ to $\e$.

\begin{example}[An elementary example: extrapolating $(1+h)^{1/h}$]
Consider
\begin{align}
    F(h):=(1+h)^{1/h}
    =\exp\!\left(\frac{\log(1+h)}{h}\right),
    \qquad F_0=\e.
    \label{eq:rich-intro-elementary-function}
\end{align}
For $|h|<1$, expanding the logarithm and then the exponential gives
\begin{align}
    \frac{\log(1+h)}{h}
    &=1-\frac h2+\frac{h^2}{3}-\frac{h^3}{4}+O(h^4),
    \label{eq:rich-intro-log-series}\\
    F(h)
    &=\e\left(1-\frac h2+\frac{11h^2}{24}
                    -\frac{7h^3}{16}+O(h^4)\right).
    \label{eq:rich-intro-elementary-series}
\end{align}
Thus $F(h)$ has error $O(h)$. At $h/2$, the leading error changes from $-\e h/2$ to $-\e h/4$. Combining the two evaluations as above removes this linear term, while three evaluations allow us to remove both the linear and quadratic terms:
\begin{align}
    2F(h/2)-F(h)
        &=\e-\frac{11\e}{48}h^2+O(h^3),\\
    \frac12F(h)-4F(h/2)+\frac92F(h/3)
        &=\e-\frac{7\e}{96}h^3+O(h^4).
    \label{eq:rich-intro-elementary-low-orders}
\end{align}
For example, the coefficients in the second line sum to one, while $\tfrac12-4/2+\tfrac92/3=0$ and $\tfrac12-4/2^2+\tfrac92/3^2=0$. These three identities preserve the constant term and cancel the coefficients of $h$ and $h^2$.

For arbitrary $k$, using the weights in \cref{eq:rich-intro-integer-weights} gives
\begin{align}
    \sum_{j=1}^k\omega_j F(h/j)
    =\sum_{j=1}^k\omega_j\left(1+\frac hj\right)^{j/h}
    =\e+O(h^k).
    \label{eq:rich-intro-elementary-order-k}
\end{align}
This constructs the order-$k$ approximation without needing to calculate the higher coefficients of the expansion of $F(h)$.
\end{example}

\subsection{A construction of the Richardson expansion for the Lindbladian evolution}
\label{supp:richardson-lindbladian-construction}

We now construct the expansion of the discrete Lindbladian evolution in inverse powers of the number of steps. We use the color classes $\mathcal C_1,\ldots,\mathcal C_\chi$ from \cref{lem:bounded-degree-coloring} and the corresponding generators in \cref{eq:color-open-system-generator}. Within each color we use the symmetric block $Q_{c,h}$, and we arrange these blocks symmetrically to obtain the full step $\mathsf{S}_h$:
\begin{align}
    Q_{c,h}
    &:=\e^{\frac h2\mathcal D^{(c)}}
       \e^{h\mathcal H^{(c)}}
       \e^{\frac h2\mathcal D^{(c)}},
    \label{eq:rich-lind-color-step}\\
    \mathsf{S}_h
    &:=Q_{1,h/2}\cdots Q_{\chi-1,h/2}\,
       Q_{\chi,h}\,
       Q_{\chi-1,h/2}\cdots Q_{1,h/2}.
    \label{eq:rich-lind-symmetric-step}
\end{align}
For $\chi=1$, this means $\mathsf{S}_h=Q_{1,h}$. Repeating $\mathsf{S}_{t/N}$ a total of $N$ times approximates $\e^{t\mathcal L}$. The following theorem identifies the corrections to this approximation and the remainder after truncating the expansion.

\begin{theorem}[Finite expansion of the symmetric Lindbladian step]
\label{thm:rich-lind-finite-expansion}
For a fixed finite system, evolution time $t>0$, and integer $q\geq1$, the repeated step has, for sufficiently large integers $N$, the expansion
\begin{align}
    \mathsf{S}_{t/N}^{N}
    =\e^{t\mathcal L}
      +\sum_{\ell=1}^{q-1}\frac{E_\ell(t)}{N^{2\ell}}
      +R_{q,N}(t).
    \label{eq:rich-lind-finite-expansion}
\end{align}
The correction maps $E_\ell(t)$ do not depend on $N$ or $q$. With the normalization of \cref{eq:methods-symmetric-expansion}, they are given explicitly by
\begin{align}
    E_\ell(t)
    =t^{2\ell}\sum_{r=1}^{\ell}
    \sum_{\substack{j_1+\cdots+j_r=\ell\\j_1,\ldots,j_r\geq1}}
    \int_{0\leq s_1\leq\cdots\leq s_r\leq t}
    &\e^{(t-s_r)\mathcal L}\mathcal K_{2j_r}
     \e^{(s_r-s_{r-1})\mathcal L}\cdots
     \mathcal K_{2j_1}\e^{s_1\mathcal L}
     \,\mathrm ds_1\cdots\mathrm ds_r,
    \label{eq:rich-lind-correction-maps}
\end{align}
where the maps $\mathcal K_{2j}$ are the odd Taylor coefficients of the logarithm of the step:
\begin{align}
    \mathcal K_{2j}
    :=\frac{1}{(2j+1)!}
      \left.\frac{\mathrm d^{2j+1}}{\mathrm dh^{2j+1}}
      \log \mathsf{S}_h\right|_{h=0},
    \qquad j\geq1.
    \label{eq:rich-lind-log-coefficients}
\end{align}
Here $R_{q,N}(t)$ is the remainder.
\end{theorem}
\begin{proof}
We first establish the expansion of one step and then collect the correction terms over the full evolution time. Each color has total coherent duration $h$ and total dissipative duration $h$ in $\mathsf{S}_h$. The first-order expansion used in \cref{eq:split-color-first-order-expansion}, together with \cref{eq:full-lindbladian-decomposition}, therefore gives
\begin{align}
    \mathsf{S}_h
    =I+h\sum_{c=1}^{\chi}
      \bigl(\mathcal H^{(c)}+\mathcal D^{(c)}\bigr)+O(h^2)
    =I+h\mathcal L+O(h^2).
    \label{eq:rich-lind-first-order-step}
\end{align}

The symmetric ordering gives
\begin{align}
    Q_{c,-h}=Q_{c,h}^{-1},
    \qquad \mathsf{S}_{-h}=\mathsf{S}_h^{-1},
    \qquad \log \mathsf{S}_{-h}=-\log \mathsf{S}_h.
    \label{eq:rich-lind-odd-logarithm}
\end{align}
In a fixed finite system, $\mathsf{S}_0=I$ ensures that $\log \mathsf{S}_h$ is analytic for sufficiently small $h$. Since $(-h)^{2j}=h^{2j}$, comparing coefficients in $\log \mathsf{S}_{-h}=-\log \mathsf{S}_h$ makes each even coefficient equal to its negative, and hence zero. Therefore,
\begin{align}
    \frac{1}{h}\log \mathsf{S}_h
    =\mathcal L+\sum_{j=1}^{q-1}h^{2j}\mathcal K_{2j}
      +O(h^{2q}).
    \label{eq:rich-lind-log-expansion}
\end{align}

To construct $\mathcal K_{2j}$ explicitly, define the ordered list $\mathcal A_1,\ldots,\mathcal A_M$ by following the block order in \cref{eq:rich-lind-symmetric-step} and replacing each $Q_{c,h/2}$ by $(\mathcal D^{(c)}/4,\mathcal H^{(c)}/2,\mathcal D^{(c)}/4)$, and the central block $Q_{\chi,h}$ by $(\mathcal D^{(\chi)}/2,\mathcal H^{(\chi)},\mathcal D^{(\chi)}/2)$. Thus $\mathcal A_\nu$ is the scaled generator in the $\nu$th exponential factor, and
\begin{align}
    \mathsf{S}_h=\e^{h\mathcal A_1}\cdots\e^{h\mathcal A_M},
    \qquad M=3(2\chi-1).
    \label{eq:rich-lind-stage-list}
\end{align}
Expanding these exponentials gives
\begin{align}
    \mathsf{S}_h=I+\sum_{m=1}^{\infty}h^m B_m,
    \qquad
    B_m=
    \sum_{\substack{m_1+\cdots+m_M=m\\m_1,\ldots,m_M\geq0}}
    \frac{\mathcal A_1^{m_1}\cdots\mathcal A_M^{m_M}}
         {m_1!\cdots m_M!}.
    \label{eq:rich-lind-step-coefficients}
\end{align}
For sufficiently small $h$, the logarithm is given by
\begin{align}
    \log \mathsf{S}_h
    =\sum_{r=1}^{\infty}\frac{(-1)^{r+1}}{r}(\mathsf{S}_h-I)^r
    =\sum_{r=1}^{\infty}\frac{(-1)^{r+1}}{r}
      \sum_{m_1,\ldots,m_r\geq1}
      h^{m_1+\cdots+m_r}B_{m_1}\cdots B_{m_r}.
    \label{eq:rich-lind-log-series}
\end{align}
Thus its coefficient of $h^{2j+1}$ is the finite sum
\begin{align}
    \mathcal K_{2j}
    =\sum_{r=1}^{2j+1}\frac{(-1)^{r+1}}{r}
      \sum_{\substack{m_1+\cdots+m_r=2j+1\\m_1,\ldots,m_r\geq1}}
      B_{m_1}\cdots B_{m_r}.
    \label{eq:rich-lind-explicit-log-coefficients}
\end{align}
For example, $B_1=\mathcal L$ and the vanishing quadratic coefficient of $\log \mathsf{S}_h$ gives $B_2=\mathcal L^2/2$. Hence
\begin{align}
    \mathcal K_2
    =B_3-\frac12(B_1B_2+B_2B_1)+\frac13B_1^3
    =B_3-\frac16\mathcal L^3.
    \label{eq:rich-lind-explicit-first-log-coefficient}
\end{align}

We have obtained expansions for $\mathsf{S}_h$ and $\log \mathsf{S}_h$. We now use them to find the correction maps $E_\ell(t)$ in the expansion of the repeated product $\mathsf{S}_{t/N}^N$. For $h=t/N$, we denote the correction to $\mathcal L$ in $h^{-1}\log \mathsf{S}_h$ by $\mathcal V_h$ and rewrite the product as a single exponential:
\begin{align}
    \mathcal V_h:=\frac{\log \mathsf{S}_h}{h}-\mathcal L,
    \qquad
    \mathsf{S}_h^N=\e^{N\log \mathsf{S}_h}=\e^{t(\mathcal L+\mathcal V_h)}.
    \label{eq:rich-lind-generator-correction}
\end{align}
The remaining task is to expand $\e^{t(\mathcal L+\mathcal V_h)}$ around $\e^{t\mathcal L}$. We derive an integral identity for their difference with one explicit factor of $\mathcal V_h$. Repeated substitution then produces terms with further factors of $\mathcal V_h$, whose powers of $h$ are known from \cref{eq:rich-lind-log-expansion}. Collecting these powers will give the maps $E_\ell(t)$. To obtain the integral identity, differentiate the product
\begin{align}
    \frac{\mathrm d}{\mathrm ds}
    \left(\e^{(t-s)\mathcal L}
          \e^{s(\mathcal L+\mathcal V_h)}\right)
    &=-\e^{(t-s)\mathcal L}\mathcal L
       \e^{s(\mathcal L+\mathcal V_h)}
      +\e^{(t-s)\mathcal L}(\mathcal L+\mathcal V_h)
       \e^{s(\mathcal L+\mathcal V_h)}\\
    &=\e^{(t-s)\mathcal L}\mathcal V_h
       \e^{s(\mathcal L+\mathcal V_h)}.
    \label{eq:rich-lind-product-derivative}
\end{align}
The endpoint values at $s=t$ and $s=0$ are $\e^{t(\mathcal L+\mathcal V_h)}$ and $\e^{t\mathcal L}$, respectively. Integrating therefore gives
\begin{align}
    \e^{t(\mathcal L+\mathcal V_h)}
    =\e^{t\mathcal L}
     +\int_0^t\e^{(t-s)\mathcal L}\mathcal V_h
       \e^{s(\mathcal L+\mathcal V_h)}\,\mathrm ds.
    \label{eq:rich-lind-variation-of-constants}
\end{align}
This is the variation-of-constants identity used here. Applying it to the last exponential inside the integral yields
\begin{align}
    \e^{t(\mathcal L+\mathcal V_h)}
    &=\e^{t\mathcal L}
      +\int_0^t\e^{(t-s_1)\mathcal L}\mathcal V_h
       \e^{s_1\mathcal L}\,\mathrm ds_1\\
    &\quad+\int_0^t\int_0^{s_2}
       \e^{(t-s_2)\mathcal L}\mathcal V_h
       \e^{(s_2-s_1)\mathcal L}\mathcal V_h
       \e^{s_1(\mathcal L+\mathcal V_h)}
       \,\mathrm ds_1\,\mathrm ds_2.
    \label{eq:rich-lind-second-substitution}
\end{align}
Repeating this substitution until the last integral contains $q$ correction factors gives the exact finite identity
\begin{align}
    \mathsf{S}_h^N
    &=\e^{t\mathcal L}
      +\sum_{r=1}^{q-1}
       \int_{0\leq s_1\leq\cdots\leq s_r\leq t}
       \e^{(t-s_r)\mathcal L}\mathcal V_h
       \e^{(s_r-s_{r-1})\mathcal L}\cdots
       \mathcal V_h\e^{s_1\mathcal L}
       \,\mathrm ds_1\cdots\mathrm ds_r\\
    &\quad+\int_{0\leq s_1\leq\cdots\leq s_q\leq t}
       \e^{(t-s_q)\mathcal L}\mathcal V_h
       \e^{(s_q-s_{q-1})\mathcal L}\cdots
       \mathcal V_h\e^{s_1(\mathcal L+\mathcal V_h)}
       \,\mathrm ds_1\cdots\mathrm ds_q.
    \label{eq:rich-lind-finite-iteration}
\end{align}
By \cref{eq:rich-lind-log-expansion}, each correction factor starts at order $h^2$. Selecting $h^{2j_i}\mathcal K_{2j_i}$ from the $i$th factor produces
\begin{align}
    h^{2j_r}\mathcal K_{2j_r}
    \e^{(s_r-s_{r-1})\mathcal L}\cdots
    h^{2j_1}\mathcal K_{2j_1}
    =h^{2(j_1+\cdots+j_r)}
     \mathcal K_{2j_r}\e^{(s_r-s_{r-1})\mathcal L}\cdots
     \mathcal K_{2j_1}.
    \label{eq:rich-lind-insertion-powers}
\end{align}
Consequently, the coefficient of $h^{2\ell}$ receives contributions only from $j_1+\cdots+j_r=\ell$, with $1\leq r\leq\ell$. Collecting these terms in \cref{eq:rich-lind-finite-iteration} gives
\begin{align}
    \mathsf{S}_h^N
    &=\e^{t\mathcal L}
      +\sum_{\ell=1}^{q-1}h^{2\ell}
       \sum_{r=1}^{\ell}
       \sum_{\substack{j_1+\cdots+j_r=\ell\\j_1,\ldots,j_r\geq1}}
       \int_{0\leq s_1\leq\cdots\leq s_r\leq t}
       \e^{(t-s_r)\mathcal L}\mathcal K_{2j_r}
       \e^{(s_r-s_{r-1})\mathcal L}\cdots
       \mathcal K_{2j_1}\e^{s_1\mathcal L}
       \,\mathrm ds_1\cdots\mathrm ds_r
       +R_{q,N}(t).
    \label{eq:rich-lind-collected-expansion}
\end{align}
Here $R_{q,N}(t)$ collects the last integral in \cref{eq:rich-lind-finite-iteration} and the terms of degree at least $2q$ omitted from the other integrals. Substituting $h^{2\ell}=t^{2\ell}/N^{2\ell}$ identifies the coefficients in \cref{eq:rich-lind-correction-maps} and proves \cref{eq:rich-lind-finite-expansion}. Each coefficient is determined by the fixed step, $\mathcal L$, and $t$, independently of $N$ and $q$.
\end{proof}

Each term in \cref{eq:rich-lind-correction-maps} describes exact Lindbladian evolution with correction maps inserted at ordered times. The earliest insertion acts on the right. For example, the first two coefficients are
\begin{align}
    E_1(t)
    &=t^2\int_0^t
      \e^{(t-s)\mathcal L}\mathcal K_2\e^{s\mathcal L}
      \,\mathrm ds,
    \label{eq:rich-lind-first-correction}\\
    E_2(t)
    &=t^4\int_0^t
      \e^{(t-s)\mathcal L}\mathcal K_4\e^{s\mathcal L}
      \,\mathrm ds\\
    &\quad+t^4\int_0^t\int_0^{s_2}
      \e^{(t-s_2)\mathcal L}\mathcal K_2
      \e^{(s_2-s_1)\mathcal L}\mathcal K_2
      \e^{s_1\mathcal L}
      \,\mathrm ds_1\,\mathrm ds_2.
    \label{eq:rich-lind-second-correction}
\end{align}
Thus the order-$N^{-4}$ correction contains either one insertion of $\mathcal K_4$ or two insertions of $\mathcal K_2$. The same rule determines every higher coefficient: the indices of the inserted maps must add to the required order.

\subsection{Bounds on the terms of the expansion}
\label{supp:richardson-expansion-bounds}

For an expansion $F=\sum_q\lambda_q F_q$ in ordered products $F_q$ of the color generators $\mathcal H^{(c)}$ and $\mathcal D^{(c)}$ from \cref{eq:color-open-system-generator}, collect identical products and define the weighted sum of absolute coefficients by
\begin{align}
    \mathrm c(F):=\sum_q|\lambda_q|
    \left(\frac{2J\e^{k\theta_w}}{\e}\right)^{n_H(q)}
    \left(\frac{\bar\gamma}{\e}\right)^{n_D(q)}.
    \label{eq:rich-lind-coefficient-weight-definition}
\end{align}
Here $n_H(q)$ and $n_D(q)$ count the Hamiltonian and dissipative factors in $F_q$, respectively. We use $\theta_w$ from \cref{eq:quantum-optimal-weight-parameter} and the model parameters of \cref{supp:model}. The empty product has weight one.

\begin{lemma}[Bounds on the step and logarithm coefficients]
\label{lem:rich-lind-coefficient-bounds}
Let $B_m$ and $\mathcal K_{2j}$ be the coefficients defined in \cref{eq:rich-lind-step-coefficients,eq:rich-lind-explicit-log-coefficients}, and let $\chi$ be the number of colors from \cref{lem:bounded-degree-coloring}. Then
\begin{align}
    \mathrm c(B_m)&\leq\frac{1}{m!}
      \left[\frac{\chi(2J\e^{k\theta_w}+\bar\gamma)}{\e}\right]^m,
    &&m\geq1,\qquad B_1=\mathcal L,
    \label{eq:rich-lind-step-coefficient-bound}\\
    \mathrm c(\mathcal K_{2j})&\leq
      \left[\frac{2\chi(2J\e^{k\theta_w}+\bar\gamma)}{\e}\right]^{2j+1},
    &&j\geq1.
    \label{eq:rich-lind-log-coefficient-bound}
\end{align}
Each term in $B_m$ contains $m$ color generators, and each term in $\mathcal K_{2j}$ contains $2j+1$ color generators.
\end{lemma}
\begin{proof}
For two such expansions $F$ and $G$, the definition in \cref{eq:rich-lind-coefficient-weight-definition} and the triangle inequality give
\begin{align}
    \mathrm c(F+G)\leq\mathrm c(F)+\mathrm c(G),
    \qquad
    \mathrm c(FG)\leq\mathrm c(F)\mathrm c(G).
    \label{eq:rich-lind-coefficient-size-rules}
\end{align}
For the product, each choice of one term from $F$ and one from $G$ contributes the product of their absolute scalar coefficients and weights. Summing over these choices gives $\mathrm c(F)\mathrm c(G)$, which bounds the coefficient sum after identical ordered products are collected.

Consider the scaled generators $\mathcal A_1,\ldots,\mathcal A_M$ in \cref{eq:rich-lind-stage-list}. As established at the start of the proof of \cref{thm:rich-lind-finite-expansion}, the fractions multiplying each $\mathcal H^{(c)}$ add to one, and the same holds for each $\mathcal D^{(c)}$. Applying \cref{eq:rich-lind-coefficient-weight-definition} therefore gives
\begin{align}
    \sum_{\nu=1}^{M}\mathrm c(\mathcal A_\nu)
    &=\sum_{c=1}^{\chi}
      \left[\mathrm c(\mathcal H^{(c)})+\mathrm c(\mathcal D^{(c)})\right]
    =\frac{\chi(2J\e^{k\theta_w}+\bar\gamma)}{\e}.
    \label{eq:rich-lind-total-coefficient-weight}
\end{align}
Applying \cref{eq:rich-lind-coefficient-size-rules} to the formula for $B_m$ in \cref{eq:rich-lind-step-coefficients} gives
\begin{align}
    \mathrm c(B_m)
    &\leq
    \sum_{\substack{m_1+\cdots+m_M=m\\m_1,\ldots,m_M\geq0}}
    \prod_{\nu=1}^{M}
    \frac{\mathrm c(\mathcal A_\nu)^{m_\nu}}{m_\nu!}
    \\
    &=\frac{1}{m!}
      \left(\sum_{\nu=1}^{M}\mathrm c(\mathcal A_\nu)\right)^m
      =\frac{1}{m!}
      \left[\frac{\chi(2J\e^{k\theta_w}+\bar\gamma)}{\e}\right]^m,
    \label{eq:rich-lind-multinomial-bound}
\end{align}
where the first equality is the multinomial identity and the second uses \cref{eq:rich-lind-total-coefficient-weight}. Each product in \cref{eq:rich-lind-step-coefficients} contains $m_1+\cdots+m_M=m$ color generators. The identity $B_1=\mathcal L$ was proved in \cref{eq:rich-lind-first-order-step}.

The formula for $\mathcal K_{2j}$ in \cref{eq:rich-lind-explicit-log-coefficients} contains products $B_{m_1}\cdots B_{m_r}$ with $m_1+\cdots+m_r=2j+1$. Each of these products therefore contains $2j+1$ color generators. Taking weighted absolute coefficients and using \cref{eq:rich-lind-coefficient-size-rules,eq:rich-lind-multinomial-bound} gives
\begin{align}
    \mathrm c(\mathcal K_{2j})
    &\leq\sum_{r=1}^{2j+1}\frac{1}{r}
      \sum_{\substack{m_1+\cdots+m_r=2j+1\\m_1,\ldots,m_r\geq1}}
      \prod_{i=1}^{r}\mathrm c(B_{m_i})
      \\
    &\leq\left[\frac{\chi(2J\e^{k\theta_w}+\bar\gamma)}{\e}\right]^{2j+1}
      \sum_{r=1}^{2j+1}\frac{1}{r}
      \sum_{\substack{m_1+\cdots+m_r=2j+1\\m_1,\ldots,m_r\geq1}}
      \frac{1}{m_1!\cdots m_r!}.
    \label{eq:rich-lind-log-coefficient-sum}
\end{align}
For fixed $r$, there are $\binom{2j}{r-1}$ ordered lists of positive integers summing to $2j+1$. They are obtained by choosing $r-1$ of the $2j$ gaps between $2j+1$ consecutive ones as separators. Using $1/r\leq1$ and $1/(m_1!\cdots m_r!)\leq1$ therefore yields
\begin{align}
    \mathrm c(\mathcal K_{2j})
    &\leq\left[\frac{\chi(2J\e^{k\theta_w}+\bar\gamma)}{\e}\right]^{2j+1}
      \sum_{r=1}^{2j+1}\binom{2j}{r-1}\\
    &=2^{2j}\left[\frac{\chi(2J\e^{k\theta_w}+\bar\gamma)}{\e}\right]^{2j+1}\\
    &\leq\left[\frac{2\chi(2J\e^{k\theta_w}+\bar\gamma)}{\e}\right]^{2j+1},
    \label{eq:rich-lind-log-coefficient-count}
\end{align}
where the equality follows from the binomial theorem.
\end{proof}

The correction maps $E_\ell(t)$ in \cref{eq:rich-lind-correction-maps} contain products of the logarithm coefficients separated by exact evolution. We combine the preceding coefficient bounds with weighted contractivity to estimate these maps in the induced norms of \cref{eq:weighted-two-induced-norm}.

\begin{lemma}[Bounds on the correction coefficients at real time]
\label{lem:rich-lind-correction-bounds}
For every integer $\ell\geq1$, $t\geq0$, and $0\leq\beta<\alpha\leq\theta_w$, the correction maps defined in \cref{eq:rich-lind-correction-maps} satisfy
\begin{align}
    \norm{E_\ell(t)}_{2,\alpha\to 2,\beta}
    \leq t^{2\ell}\sum_{r=1}^{\ell}
    \binom{\ell-1}{r-1}\frac{t^r}{r!}
    \left[\frac{2\chi(2J\e^{k\theta_w}+\bar\gamma)(2\ell+r)}{\e(\alpha-\beta)}\right]^{2\ell+r}.
    \label{eq:rich-lind-correction-bound}
\end{align}
\end{lemma}
\begin{proof}
The expansion in \cref{eq:rich-lind-correction-maps} is a sum of integrals of products. We first use the triangle inequality to bound each integrand separately:
\begin{align}
    \norm{E_\ell(t)}_{2,\alpha\to 2,\beta}
    &\leq t^{2\ell}\sum_{r=1}^{\ell}
    \sum_{\substack{j_1+\cdots+j_r=\ell\\j_1,\ldots,j_r\geq1}}
    \int_{0\leq s_1\leq\cdots\leq s_r\leq t}
    \norm{
        \e^{(t-s_r)\mathcal L}\mathcal K_{2j_r}
        \e^{(s_r-s_{r-1})\mathcal L}\cdots
        \mathcal K_{2j_1}\e^{s_1\mathcal L}
    }_{2,\alpha\to 2,\beta}
    \,\mathrm ds_1\cdots\mathrm ds_r.
    \label{eq:rich-lind-correction-triangle}
\end{align}
For the color generators, we use the terms linear in the evolution time in the column and row estimates \cref{eq:parallel-color-column-bound,eq:parallel-color-row-bound}. Applying Schur's inequality \cref{eq:quantum-schur-test} to these terms gives the Hamiltonian bound below. The dissipative bound uses the diagonal action and rate bounds \cref{eq:local-pauli-diagonal-noise,eq:local-noise-window}, with the same color overlap count used in \cref{eq:parallel-color-column-bound}. These estimates give
\begin{align}
    \norm{\mathcal H^{(c)}}_{2,\alpha\to 2,\beta}
    &\leq\frac{2J\e^{k\theta_w}}{\e(\alpha-\beta)},
    \label{eq:rich-lind-hamiltonian-weight-loss}\\
    \norm{\mathcal D^{(c)}}_{2,\alpha\to 2,\beta}
    &\leq\frac{\bar\gamma}{\e(\alpha-\beta)}.
    \label{eq:rich-lind-dissipative-weight-loss}
\end{align}
The factor $1/\e$ comes from $\sup_{x\geq0}x\e^{-(\alpha-\beta)x} =1/[\e(\alpha-\beta)]$.

We now apply these generator bounds to the integrand in \cref{eq:rich-lind-correction-triangle}. Fix a list $j_1+\cdots+j_r=\ell$. The coefficient formulas \cref{eq:rich-lind-step-coefficients,eq:rich-lind-explicit-log-coefficients} show that each term in $\mathcal K_{2j_i}$ contains $2j_i+1$ color generators, giving $\sum_{i=1}^r(2j_i+1)=2\ell+r$ generators in total. Assign a weight decrease to each of these generators by setting
\begin{align}
    \alpha_i:=\alpha-\frac{i(\alpha-\beta)}{2\ell+r},
    \qquad 0\leq i\leq2\ell+r.
    \label{eq:rich-lind-correction-intermediate-weights}
\end{align}
Thus $\alpha_0=\alpha$, $\alpha_{2\ell+r}=\beta$, and $\alpha_{i-1}-\alpha_i=(\alpha-\beta)/(2\ell+r)$. The first correction map contains $2j_1+1$ generators and satisfies
\begin{align}
    \norm{\mathcal K_{2j_1}}_{2,\alpha\to 2,\alpha_{2j_1+1}}
    &\leq\mathrm c(\mathcal K_{2j_1})
      \left(\frac{2\ell+r}{\alpha-\beta}\right)^{2j_1+1}.
    \label{eq:rich-lind-correction-map-weight-loss}
\end{align}
To prove this bound, recall the scaled generators $\mathcal A_\nu$ from \cref{eq:rich-lind-stage-list} and the expansion \cref{eq:rich-lind-step-coefficients}:
\begin{align}
    \mathsf{S}_h&=I+\sum_{m=1}^{\infty}h^mB_m,\qquad
    B_m=\sum_{\substack{m_1+\cdots+m_M=m\\m_1,\ldots,m_M\geq0}}
    \frac{\mathcal A_1^{m_1}\cdots\mathcal A_M^{m_M}}{m_1!\cdots m_M!}.
    \label{eq:rich-lind-recalled-step-coefficients}
\end{align}
Substituting these coefficients into \cref{eq:rich-lind-explicit-log-coefficients} and collecting identical ordered products gives
\begin{align}
    \mathcal K_{2j_1}
    &=\sum_{p=1}^{2j_1+1}\frac{(-1)^{p+1}}{p}
      \sum_{\substack{m_1+\cdots+m_p=2j_1+1\\m_1,\ldots,m_p\geq1}}
      B_{m_1}\cdots B_{m_p}\\
    &=\sum_q\lambda_q\,
      \mathcal A_{\nu_{q,2j_1+1}}\cdots\mathcal A_{\nu_{q,1}}.
    \label{eq:rich-lind-collected-stage-expansion}
\end{align}
Here $q$ indexes the distinct ordered color-generator products, $\nu_{q,m}\in\{1,\ldots,M\}$ selects the stage generator acting in position $m$, and $\lambda_q$ is the collected scalar coefficient multiplying the displayed product of scaled generators. Every product has $2j_1+1$ factors because $m_1+\cdots+m_p=2j_1+1$. The stage list \cref{eq:rich-lind-stage-list} assigns the same fraction to every occurrence of a given color generator. These fractions remain in $\mathcal A_\nu$. The definition in \cref{eq:rich-lind-coefficient-weight-definition} therefore gives
\begin{align}
    \mathrm c(\mathcal K_{2j_1})
    &=\sum_q|\lambda_q|
      \prod_{m=1}^{2j_1+1}\mathrm c(\mathcal A_{\nu_{q,m}}).
    \label{eq:rich-lind-collected-coefficient-weight}
\end{align}

For $1\leq m\leq2j_1+1$, the weights in \cref{eq:rich-lind-correction-intermediate-weights} satisfy $0\leq\alpha_{2j_1+1}\leq\alpha_m<\alpha_{m-1}\leq\alpha\leq\theta_w$. Applying \cref{eq:rich-lind-hamiltonian-weight-loss,eq:rich-lind-dissipative-weight-loss} to the scaled generators of \cref{eq:rich-lind-stage-list}, using norm homogeneity and the definition in \cref{eq:rich-lind-coefficient-weight-definition}, gives
\begin{align}
    \norm{\mathcal A_{\nu_{q,m}}}_{2,\alpha_{m-1}\to2,\alpha_m}
    &\leq\frac{\mathrm c(\mathcal A_{\nu_{q,m}})}{\alpha_{m-1}-\alpha_m}\\
    &=\frac{2\ell+r}{\alpha-\beta}\,
      \mathrm c(\mathcal A_{\nu_{q,m}})\\
    &=\frac{2j_1+1}{\alpha-\alpha_{2j_1+1}}\,
      \mathrm c(\mathcal A_{\nu_{q,m}}).
    \label{eq:rich-lind-individual-stage-weight-loss}
\end{align}
We now prove the product estimate by induction on the number $m$ of generator factors. For the base case $m=1$, \cref{eq:rich-lind-individual-stage-weight-loss} gives
\begin{align}
    \norm{\mathcal A_{\nu_{q,1}}}_{2,\alpha\to2,\alpha_1}
    &\leq\frac{2j_1+1}{\alpha-\alpha_{2j_1+1}}\,
      \mathrm c(\mathcal A_{\nu_{q,1}}).
\end{align}
For the induction step, let $2\leq m\leq2j_1+1$ and assume inductively that the estimate holds for products of $m-1$ generators:
\begin{align}
    \norm{\mathcal A_{\nu_{q,m-1}}\cdots\mathcal A_{\nu_{q,1}}}_{2,\alpha\to2,\alpha_{m-1}}
    &\leq\left(\frac{2j_1+1}{\alpha-\alpha_{2j_1+1}}\right)^{m-1}
      \prod_{s=1}^{m-1}\mathrm c(\mathcal A_{\nu_{q,s}}).
    \label{eq:rich-lind-stage-product-induction-hypothesis}
\end{align}
The composition bound \cref{eq:weighted-induced-submultiplicativity}, the single-generator estimate \cref{eq:rich-lind-individual-stage-weight-loss}, and the induction hypothesis \cref{eq:rich-lind-stage-product-induction-hypothesis} give
\begin{align}
    &\norm{\mathcal A_{\nu_{q,m}}\cdots\mathcal A_{\nu_{q,1}}}_{2,\alpha\to2,\alpha_m}\\
    &\quad\leq
      \norm{\mathcal A_{\nu_{q,m}}}_{2,\alpha_{m-1}\to2,\alpha_m}
      \norm{\mathcal A_{\nu_{q,m-1}}\cdots\mathcal A_{\nu_{q,1}}}_{2,\alpha\to2,\alpha_{m-1}}\\
    &\quad\leq
      \left[\frac{2j_1+1}{\alpha-\alpha_{2j_1+1}}\,
      \mathrm c(\mathcal A_{\nu_{q,m}})\right]
      \left[\left(\frac{2j_1+1}{\alpha-\alpha_{2j_1+1}}\right)^{m-1}
      \prod_{s=1}^{m-1}\mathrm c(\mathcal A_{\nu_{q,s}})\right]\\
    &\quad=\left(\frac{2j_1+1}{\alpha-\alpha_{2j_1+1}}\right)^m
      \prod_{s=1}^{m}\mathrm c(\mathcal A_{\nu_{q,s}}).
\end{align}
This completes the induction. Setting $m=2j_1+1$ gives
\begin{align}
    \norm{\mathcal A_{\nu_{q,2j_1+1}}\cdots\mathcal A_{\nu_{q,1}}}_{2,\alpha\to2,\alpha_{2j_1+1}}
    &\leq\left(\frac{2j_1+1}{\alpha-\alpha_{2j_1+1}}\right)^{2j_1+1}
      \prod_{m=1}^{2j_1+1}\mathrm c(\mathcal A_{\nu_{q,m}}).
    \label{eq:rich-lind-stage-product-weight-loss}
\end{align}

Substituting \cref{eq:rich-lind-collected-stage-expansion} into the induced-norm definition \cref{eq:weighted-two-induced-norm}, applying the triangle inequality, and then using \cref{eq:rich-lind-stage-product-weight-loss} gives
\begin{align}
    \norm{\mathcal K_{2j_1}}_{2,\alpha\to2,\alpha_{2j_1+1}}
    &=\sup_{O\ne0}
      \frac{\norm{\sum_q\lambda_q\mathcal A_{\nu_{q,2j_1+1}}\cdots\mathcal A_{\nu_{q,1}}(O)}_{2,\alpha_{2j_1+1}}}
      {\norm{O}_{2,\alpha}}\\
    &\leq\sup_{O\ne0}\sum_q|\lambda_q|
      \frac{\norm{\mathcal A_{\nu_{q,2j_1+1}}\cdots\mathcal A_{\nu_{q,1}}(O)}_{2,\alpha_{2j_1+1}}}
      {\norm{O}_{2,\alpha}}\\
    &\leq\sum_q|\lambda_q|\sup_{O\ne0}
      \frac{\norm{\mathcal A_{\nu_{q,2j_1+1}}\cdots\mathcal A_{\nu_{q,1}}(O)}_{2,\alpha_{2j_1+1}}}
      {\norm{O}_{2,\alpha}}\\
    &=\sum_q|\lambda_q|
      \norm{\mathcal A_{\nu_{q,2j_1+1}}\cdots\mathcal A_{\nu_{q,1}}}_{2,\alpha\to2,\alpha_{2j_1+1}}\\
    &\leq\sum_q|\lambda_q|
      \left(\frac{2j_1+1}{\alpha-\alpha_{2j_1+1}}\right)^{2j_1+1}
      \prod_{m=1}^{2j_1+1}\mathrm c(\mathcal A_{\nu_{q,m}})\\
    &=\left(\frac{2j_1+1}{\alpha-\alpha_{2j_1+1}}\right)^{2j_1+1}
      \sum_q|\lambda_q|\prod_{m=1}^{2j_1+1}\mathrm c(\mathcal A_{\nu_{q,m}})\\
    &=\left(\frac{2j_1+1}{\alpha-\alpha_{2j_1+1}}\right)^{2j_1+1}
      \mathrm c(\mathcal K_{2j_1})\\
    &=\left(\frac{2\ell+r}{\alpha-\beta}\right)^{2j_1+1}
      \mathrm c(\mathcal K_{2j_1}).
    \label{eq:rich-lind-correction-weight-factor}
\end{align}
The first inequality is the triangle inequality, and the second moves the supremum inside the finite sum as an upper bound. The last two equalities use \cref{eq:rich-lind-collected-coefficient-weight,eq:rich-lind-correction-intermediate-weights}, respectively. This proves \cref{eq:rich-lind-correction-map-weight-loss}. We can therefore remove the two rightmost factors of the integrand as follows:
\begin{align}
    &\norm{
        \e^{(t-s_r)\mathcal L}\mathcal K_{2j_r}
        \e^{(s_r-s_{r-1})\mathcal L}\cdots
        \mathcal K_{2j_1}\e^{s_1\mathcal L}
    }_{2,\alpha\to 2,\beta}\\
    &\quad\leq\norm{
        \e^{(t-s_r)\mathcal L}\mathcal K_{2j_r}
        \e^{(s_r-s_{r-1})\mathcal L}\cdots
        \mathcal K_{2j_1}
    }_{2,\alpha\to 2,\beta}\\
    &\quad\leq\norm{
        \e^{(t-s_r)\mathcal L}\mathcal K_{2j_r}\cdots
        \mathcal K_{2j_2}\e^{(s_2-s_1)\mathcal L}
    }_{2,\alpha_{2j_1+1}\to 2,\beta}
    \norm{\mathcal K_{2j_1}}_{2,\alpha\to 2,\alpha_{2j_1+1}}\\
    &\quad\leq\mathrm c(\mathcal K_{2j_1})
      \left(\frac{2\ell+r}{\alpha-\beta}\right)^{2j_1+1}
      \norm{
        \e^{(t-s_r)\mathcal L}\mathcal K_{2j_r}\cdots
        \mathcal K_{2j_2}\e^{(s_2-s_1)\mathcal L}
      }_{2,\alpha_{2j_1+1}\to 2,\beta}.
    \label{eq:rich-lind-first-correction-step}
\end{align}
The first inequality uses the contraction \cref{eq:supp-main-weight-regularity} at weight $\alpha$ and the composition bound \cref{eq:weighted-induced-submultiplicativity}. The second uses the same composition bound with intermediate weight $\alpha_{2j_1+1}$, and the third uses \cref{eq:rich-lind-correction-map-weight-loss}. For $r=1$, the remaining product is $\e^{(t-s_1)\mathcal L}$ and $\alpha_{2j_1+1}=\beta$.

Repeat this step for the remaining correction maps, assigning a weight decrease $(2j_i+1)(\alpha-\beta)/(2\ell+r)$ to $\mathcal K_{2j_i}$. Since $\sum_{i=1}^r(2j_i+1)=2\ell+r$, the final input weight is $\alpha_{2\ell+r}=\beta$. Each exact evolution contributes a norm at most one by \cref{eq:supp-main-weight-regularity}. The coefficient bound \cref{eq:rich-lind-log-coefficient-bound} gives
\begin{align}
    \prod_{i=1}^{r}\mathrm c(\mathcal K_{2j_i})
    \leq\prod_{i=1}^{r}
      \left[\frac{2\chi(2J\e^{k\theta_w}+\bar\gamma)}{\e}\right]^{2j_i+1}
    =\left[\frac{2\chi(2J\e^{k\theta_w}+\bar\gamma)}{\e}\right]^{2\ell+r}.
    \label{eq:rich-lind-correction-coefficient-product}
\end{align}
Combining these estimates yields
\begin{align}
    \norm{
        \e^{(t-s_r)\mathcal L}\mathcal K_{2j_r}
        \e^{(s_r-s_{r-1})\mathcal L}\cdots
        \mathcal K_{2j_1}\e^{s_1\mathcal L}
    }_{2,\alpha\to 2,\beta}
    &\leq\left(\frac{2\ell+r}{\alpha-\beta}\right)^{2\ell+r}
      \prod_{i=1}^{r}\mathrm c(\mathcal K_{2j_i})\\
    &\leq\left[\frac{2\chi(2J\e^{k\theta_w}+\bar\gamma)(2\ell+r)}{\e(\alpha-\beta)}\right]^{2\ell+r}.
    \label{eq:rich-lind-correction-integrand-bound}
\end{align}
This bound is independent of the integration times. The ordered integration region has volume
\begin{align}
    \int_{0\leq s_1\leq\cdots\leq s_r\leq t}
    \,\mathrm ds_1\cdots\mathrm ds_r
    =\int_0^t\frac{s_r^{r-1}}{(r-1)!}\,\mathrm ds_r
    =\frac{t^r}{r!}.
    \label{eq:rich-lind-ordered-time-volume}
\end{align}
For fixed $r$, choosing $r-1$ separators among $\ell-1$ gaps counts the positive lists $j_1+\cdots+j_r=\ell$. This is the counting step preceding \cref{eq:rich-lind-log-coefficient-count}, with total $\ell$ in place of $2j+1$, and it gives $\binom{\ell-1}{r-1}$ lists. Substituting the integrand bound and the integration volume into \cref{eq:rich-lind-correction-triangle} proves
\begin{align}
    \norm{E_\ell(t)}_{2,\alpha\to 2,\beta}
    &\leq t^{2\ell}\sum_{r=1}^{\ell}
    \sum_{\substack{j_1+\cdots+j_r=\ell\\j_1,\ldots,j_r\geq1}}
    \frac{t^r}{r!}
    \left[\frac{2\chi(2J\e^{k\theta_w}+\bar\gamma)(2\ell+r)}{\e(\alpha-\beta)}\right]^{2\ell+r}
    \\
    &=t^{2\ell}\sum_{r=1}^{\ell}
    \binom{\ell-1}{r-1}\frac{t^r}{r!}
    \left[\frac{2\chi(2J\e^{k\theta_w}+\bar\gamma)(2\ell+r)}{\e(\alpha-\beta)}\right]^{2\ell+r}.
    \label{eq:rich-lind-correction-bound-completed}
\end{align}
\end{proof}

The preceding lemma bounds the correction maps individually. We now bound the difference between one step $\mathsf{S}_h$ and the exact evolution with the first $p$ correction maps retained.

\begin{lemma}[One-step remainder]
\label{lem:rich-lind-one-step-remainder}
For $0\leq\beta<\alpha\leq\theta_w$ and an integer $p\geq0$, define
\begin{align}
    \tau_p:=\frac{\alpha-\beta}
    {1024\chi(2J\e^{k\theta_w}+\bar\gamma)(p+1)}.
    \label{eq:rich-lind-one-step-radius}
\end{align}
For $0\leq h\leq\tau_p$, the one-step remainder $R_{p+1,1}(h)$ in \cref{eq:rich-lind-finite-expansion} satisfies
\begin{align}
    \norm{R_{p+1,1}(h)}_{2,\alpha\to2,\beta}
    \leq6\left(\frac{h}{\tau_p}\right)^{2p+3}.
    \label{eq:rich-lind-one-step-remainder-bound}
\end{align}
The constants are independent of the number of qubits.
\end{lemma}
\begin{proof}
We first identify the Taylor coefficients that cancel in $R_{p+1,1}(h)$. Write $[h^m]F(h):=\frac{1}{m!}\left.\frac{\mathrm d^m}{\mathrm dh^m}F(h)\right|_{h=0}$ for the coefficient of $h^m$. Rescaling the integrals in \cref{eq:rich-lind-correction-maps} to the unit interval gives, for $h>0$,
\begin{align}
    E_a(h)
    &=\sum_{r=1}^{a}
      \sum_{\substack{j_1+\cdots+j_r=a\\j_1,\ldots,j_r\geq1}}
      h^{2a+r}\int_{0\leq s_1\leq\cdots\leq s_r\leq1}
      \e^{h(1-s_r)\mathcal L}\mathcal K_{2j_r}
      \e^{h(s_r-s_{r-1})\mathcal L}\cdots
      \mathcal K_{2j_1}\e^{hs_1\mathcal L}
      \,\mathrm ds_1\cdots\mathrm ds_r.
    \label{eq:rich-lind-rescaled-correction}
\end{align}
The $r$ integrations contribute $h^r$ to the prefactor $h^{2a+r}$. With $s_0=0$ and $s_{r+1}=1$, the exponential series for each interval is
\begin{align}
    \e^{h(s_{i+1}-s_i)\mathcal L}
    =\sum_{m_i=0}^{\infty}
      h^{m_i}\frac{(s_{i+1}-s_i)^{m_i}}{m_i!}\mathcal L^{m_i},
    \qquad 0\leq i\leq r.
    \label{eq:rich-lind-interval-exponential-series}
\end{align}
At a fixed finite system, $\sum_{i=0}^{r}(s_{i+1}-s_i)=1$ gives the following bound on the product of the absolute norm series:
\begin{align}
    \prod_{i=0}^{r}\sum_{m_i=0}^{\infty}
      \frac{|h|^{m_i}(s_{i+1}-s_i)^{m_i}}{m_i!}
      \norm{\mathcal L}_{2,0\to2,0}^{m_i}
    =\exp\!\left(|h|\norm{\mathcal L}_{2,0\to2,0}\right).
    \label{eq:rich-lind-correction-series-convergence}
\end{align}
The bound is uniform on the integration region. Multiplication by the finitely many fixed correction maps therefore permits termwise integration in \cref{eq:rich-lind-rescaled-correction}. In particular, $E_a(h)$ has a convergent Taylor series for every finite $h$, and the factor $h^{2a+r}$ with $r\geq1$ gives
\begin{align}
    [h^m]E_a(h)=0,\qquad 0\leq m\leq2a.
    \label{eq:rich-lind-correction-minimum-degree}
\end{align}
Collecting Taylor coefficients as in \cref{eq:rich-lind-insertion-powers,eq:rich-lind-collected-expansion} gives
\begin{align}
    [h^m]\mathsf{S}_h
    &=[h^m]\e^{h\mathcal L}
      +\sum_{a=1}^{\lfloor m/2\rfloor}[h^m]E_a(h),\\
    [h^m]R_{p+1,1}(h)
    &=[h^m]\mathsf{S}_h-[h^m]\e^{h\mathcal L}
      -\sum_{a=1}^{p}[h^m]E_a(h)
    \label{eq:rich-lind-one-step-remainder-coefficient}\\
    &=\sum_{a=p+1}^{\lfloor m/2\rfloor}[h^m]E_a(h).
    \label{eq:rich-lind-one-step-coefficient-cancellation}
\end{align}
For $m\leq2p+1$, the last sum is empty. For $m=2p+2$, its only term is $[h^{2p+2}]E_{p+1}(h)=0$ by \cref{eq:rich-lind-correction-minimum-degree}. Consequently,
\begin{align}
    [h^m]R_{p+1,1}(h)&=0,\qquad 0\leq m\leq2p+2,\\
    R_{p+1,1}(h)
    &=\sum_{m=2p+3}^{\infty}h^m[h^m]R_{p+1,1}(h).
    \label{eq:rich-lind-one-step-tail}
\end{align}
The second equality uses the convergent exponential series for the finite product $\mathsf{S}_h$ and the convergence established in \cref{eq:rich-lind-correction-series-convergence}.

For $0<h\leq\tau_p$ and $m\geq2p+3$, we can extract the first remaining power as
\begin{align}
    h^m
    &=\left(\frac{h}{\tau_p}\right)^{2p+3}
      \tau_p^m\left(\frac{h}{\tau_p}\right)^{m-2p-3}
    \leq\left(\frac{h}{\tau_p}\right)^{2p+3}\tau_p^m.
    \label{eq:rich-lind-one-step-extracted-power}
\end{align}
Applying the triangle inequality to \cref{eq:rich-lind-one-step-tail}, then using \cref{eq:rich-lind-one-step-extracted-power,eq:rich-lind-one-step-remainder-coefficient}, gives
\begin{align}
    \norm{R_{p+1,1}(h)}_{2,\alpha\to2,\beta}
    &\leq\sum_{m=2p+3}^{\infty}h^m
      \norm{[h^m]R_{p+1,1}(h)}_{2,\alpha\to2,\beta}\\
    &\leq\left(\frac{h}{\tau_p}\right)^{2p+3}
      \sum_{m=2p+3}^{\infty}\tau_p^m
      \norm{[h^m]R_{p+1,1}(h)}_{2,\alpha\to2,\beta}\\
    &=\left(\frac{h}{\tau_p}\right)^{2p+3}
      \sum_{m=2p+3}^{\infty}\tau_p^m
      \norm{[h^m]\mathsf{S}_h-[h^m]\e^{h\mathcal L}
      -\sum_{a=1}^{p}[h^m]E_a(h)}_{2,\alpha\to2,\beta}\\
    &\leq\left(\frac{h}{\tau_p}\right)^{2p+3}
      \Biggl[
      \sum_{m=0}^{\infty}\tau_p^m
      \norm{[h^m]\mathsf{S}_h}_{2,\alpha\to2,\beta}
      +\sum_{m=0}^{\infty}\tau_p^m
      \norm{[h^m]\e^{h\mathcal L}}_{2,\alpha\to2,\beta}\\
    &\hspace{33mm}
      +\sum_{a=1}^{p}\sum_{m=0}^{\infty}\tau_p^m
      \norm{[h^m]E_a(h)}_{2,\alpha\to2,\beta}
      \Biggr].
    \label{eq:rich-lind-one-step-tail-bound}
\end{align}
The last inequality applies the triangle inequality to each coefficient and adds the nonnegative terms of degrees below $2p+3$. We now bound the three coefficient sums in the bracket.

For an expansion $F$ whose products contain $m$ color generators, the induction proving \cref{eq:rich-lind-stage-product-weight-loss} and the triangle-inequality calculation in \cref{eq:rich-lind-correction-weight-factor}, applied to $m$ factors between weights $\alpha$ and $\beta$, give
\begin{align}
    \norm{F}_{2,\alpha\to2,\beta}
    \leq\left(\frac{m}{\alpha-\beta}\right)^m\mathrm c(F),
    \qquad m\geq1.
    \label{eq:rich-lind-fixed-degree-norm-bound}
\end{align}
Using $B_1=\mathcal L$ and the bounds \cref{eq:rich-lind-step-coefficient-bound,eq:rich-lind-coefficient-size-rules} in Lemma~\ref{lem:rich-lind-coefficient-bounds}, we obtain
\begin{align}
    \mathrm c(\mathcal L^m)
    &\leq\mathrm c(\mathcal L)^m
    \leq\left[\frac{\chi(2J\e^{k\theta_w}+\bar\gamma)}{\e}\right]^m,
    \label{eq:rich-lind-generator-power-weight}\\
    \norm{[h^m]\mathsf{S}_h}_{2,\alpha\to2,\beta}
    &=\norm{B_m}_{2,\alpha\to2,\beta}
    \leq\left(\frac{m}{\alpha-\beta}\right)^m\mathrm c(B_m)\\
    &\leq\frac{1}{m!}
      \left[\frac{m\chi(2J\e^{k\theta_w}+\bar\gamma)}{\e(\alpha-\beta)}\right]^m,
    \label{eq:rich-lind-step-taylor-coefficient-bound}\\
    \norm{[h^m]\e^{h\mathcal L}}_{2,\alpha\to2,\beta}
    &=\frac{\norm{\mathcal L^m}_{2,\alpha\to2,\beta}}{m!}
    \leq\frac{1}{m!}\left(\frac{m}{\alpha-\beta}\right)^m\mathrm c(\mathcal L^m)\\
    &\leq\frac{1}{m!}
      \left[\frac{m\chi(2J\e^{k\theta_w}+\bar\gamma)}{\e(\alpha-\beta)}\right]^m,
    \qquad m\geq1.
    \label{eq:rich-lind-exact-taylor-coefficient-bound}
\end{align}
For $m=0$, both coefficients equal $I$, and \cref{eq:theta-monotonicity,eq:weighted-two-induced-norm} give $\norm{I}_{2,\alpha\to2,\beta}\leq1$. The scalar estimates needed to sum the higher degrees are
\begin{align}
    \e^m=\sum_{j=0}^{\infty}\frac{m^j}{j!}
    &\geq\frac{m^m}{m!},
    \label{eq:rich-lind-factorial-power-bound}\\
    y:=\frac{\chi(2J\e^{k\theta_w}+\bar\gamma)\tau_p}{\alpha-\beta}
    &=\frac{1}{1024(p+1)}\leq\frac{1}{1024},
    \label{eq:rich-lind-one-step-series-parameter}
\end{align}
where the identity for $y$ uses \cref{eq:rich-lind-one-step-radius}. Substituting into \cref{eq:rich-lind-step-taylor-coefficient-bound,eq:rich-lind-exact-taylor-coefficient-bound} yields
\begin{align}
    \sum_{m=0}^{\infty}\tau_p^m
      \norm{[h^m]\mathsf{S}_h}_{2,\alpha\to2,\beta}
    &\leq1+\sum_{m=1}^{\infty}\frac{m^m}{m!}
      \left(\frac{y}{\e}\right)^m
    \leq\sum_{m=0}^{\infty}y^m=\frac{1}{1-y},
    \label{eq:rich-lind-step-taylor-coefficient-sum}\\
    \sum_{m=0}^{\infty}\tau_p^m
      \norm{[h^m]\e^{h\mathcal L}}_{2,\alpha\to2,\beta}
    &\leq1+\sum_{m=1}^{\infty}\frac{m^m}{m!}
      \left(\frac{y}{\e}\right)^m
    \leq\frac{1}{1-y}.
    \label{eq:rich-lind-exact-taylor-coefficient-sum}
\end{align}

It remains to bound the retained-correction sum in \cref{eq:rich-lind-one-step-tail-bound}. Fix $1\leq a\leq p$. Multiplying the series \cref{eq:rich-lind-interval-exponential-series} in their original order inside \cref{eq:rich-lind-rescaled-correction} gives
\begin{align}
    [h^m]E_a(h)
    &=\sum_{r=1}^{a}
      \sum_{\substack{j_1+\cdots+j_r=a\\j_1,\ldots,j_r\geq1}}
      \int_{0\leq s_1\leq\cdots\leq s_r\leq1}
      \sum_{\substack{m_0+\cdots+m_r=m-2a-r\\m_0,\ldots,m_r\geq0}}
      \left(\prod_{i=0}^{r}\frac{(s_{i+1}-s_i)^{m_i}}{m_i!}\right)\\
    &\qquad\qquad\cdot
      \mathcal L^{m_r}\mathcal K_{2j_r}\mathcal L^{m_{r-1}}\cdots
      \mathcal K_{2j_1}\mathcal L^{m_0}
      \,\mathrm ds_1\cdots\mathrm ds_r.
    \label{eq:rich-lind-correction-taylor-coefficient}
\end{align}
The inner sum is empty when $m<2a+r$. Otherwise, every product in this coefficient has total generator degree
\begin{align}
    \sum_{i=0}^{r}m_i+\sum_{i=1}^{r}(2j_i+1)
    =(m-2a-r)+(2a+r)=m.
    \label{eq:rich-lind-correction-taylor-degree}
\end{align}
The coefficient inequalities \cref{eq:rich-lind-coefficient-size-rules,eq:rich-lind-log-coefficient-bound}, together with \cref{eq:rich-lind-generator-power-weight} at $m=1$, give
\begin{align}
    &\mathrm c\!\left(
      \mathcal L^{m_r}\mathcal K_{2j_r}\mathcal L^{m_{r-1}}\cdots
      \mathcal K_{2j_1}\mathcal L^{m_0}\right)\\
    &\quad\leq\mathrm c(\mathcal L)^{m-2a-r}
      \prod_{i=1}^{r}\mathrm c(\mathcal K_{2j_i})\\
    &\quad\leq
      \left[\frac{\chi(2J\e^{k\theta_w}+\bar\gamma)}{\e}\right]^{m-2a-r}
      \prod_{i=1}^{r}
      \left[\frac{2\chi(2J\e^{k\theta_w}+\bar\gamma)}{\e}\right]^{2j_i+1}\\
    &\quad=2^{2a+r}
      \left[\frac{\chi(2J\e^{k\theta_w}+\bar\gamma)}{\e}\right]^m.
    \label{eq:rich-lind-correction-taylor-product-weight}
\end{align}
Applying \cref{eq:rich-lind-fixed-degree-norm-bound} to the entire product therefore yields
\begin{align}
    &\norm{
      \mathcal L^{m_r}\mathcal K_{2j_r}\mathcal L^{m_{r-1}}\cdots
      \mathcal K_{2j_1}\mathcal L^{m_0}}_{2,\alpha\to2,\beta}\\
    &\quad\leq\left(\frac{m}{\alpha-\beta}\right)^m
      \mathrm c\!\left(
      \mathcal L^{m_r}\mathcal K_{2j_r}\mathcal L^{m_{r-1}}\cdots
      \mathcal K_{2j_1}\mathcal L^{m_0}\right)\\
    &\quad\leq2^{2a+r}
      \left[\frac{m\chi(2J\e^{k\theta_w}+\bar\gamma)}{\e(\alpha-\beta)}\right]^m.
    \label{eq:rich-lind-correction-taylor-product-norm}
\end{align}
The scalar factors in \cref{eq:rich-lind-correction-taylor-coefficient} can be summed before integration. For $m\geq2a+r$, the multinomial identity and $\sum_{i=0}^{r}(s_{i+1}-s_i)=1$ give
\begin{align}
    \sum_{\substack{m_0+\cdots+m_r=m-2a-r\\m_0,\ldots,m_r\geq0}}
      \prod_{i=0}^{r}\frac{(s_{i+1}-s_i)^{m_i}}{m_i!}
    &=\frac{\left[\sum_{i=0}^{r}(s_{i+1}-s_i)\right]^{m-2a-r}}
      {(m-2a-r)!}
    =\frac{1}{(m-2a-r)!},\\
    \int_{0\leq s_1\leq\cdots\leq s_r\leq1}
      \frac{\mathrm ds_1\cdots\mathrm ds_r}{(m-2a-r)!}
    &=\frac{1}{r!(m-2a-r)!}.
    \label{eq:rich-lind-correction-taylor-scalar-integral}
\end{align}
The integral uses the ordered-region volume \cref{eq:rich-lind-ordered-time-volume} at $t=1$. Applying the triangle inequality to \cref{eq:rich-lind-correction-taylor-coefficient}, followed by \cref{eq:rich-lind-correction-taylor-product-norm,eq:rich-lind-correction-taylor-scalar-integral}, gives
\begin{align}
    \norm{[h^m]E_a(h)}_{2,\alpha\to2,\beta}
    &\leq\sum_{\substack{1\leq r\leq a\\2a+r\leq m}}
      \sum_{\substack{j_1+\cdots+j_r=a\\j_1,\ldots,j_r\geq1}}
      \int_{0\leq s_1\leq\cdots\leq s_r\leq1}
      \sum_{\substack{m_0+\cdots+m_r=m-2a-r\\m_0,\ldots,m_r\geq0}}
      \left(\prod_{i=0}^{r}\frac{(s_{i+1}-s_i)^{m_i}}{m_i!}\right)\\
    &\qquad\qquad\cdot
      \norm{\mathcal L^{m_r}\mathcal K_{2j_r}\mathcal L^{m_{r-1}}\cdots
      \mathcal K_{2j_1}\mathcal L^{m_0}}_{2,\alpha\to2,\beta}
      \,\mathrm ds_1\cdots\mathrm ds_r\\
    &\leq\sum_{\substack{1\leq r\leq a\\2a+r\leq m}}
      \sum_{\substack{j_1+\cdots+j_r=a\\j_1,\ldots,j_r\geq1}}
      \frac{2^{2a+r}}{r!(m-2a-r)!}
      \left[\frac{m\chi(2J\e^{k\theta_w}+\bar\gamma)}{\e(\alpha-\beta)}\right]^m\\
    &=\sum_{\substack{1\leq r\leq a\\2a+r\leq m}}
      \binom{a-1}{r-1}\frac{2^{2a+r}}{r!(m-2a-r)!}
      \left[\frac{m\chi(2J\e^{k\theta_w}+\bar\gamma)}{\e(\alpha-\beta)}\right]^m.
    \label{eq:rich-lind-correction-taylor-coefficient-bound}
\end{align}
The equality uses the same count of positive lists as the step preceding \cref{eq:rich-lind-correction-bound-completed}, with total $a$ in place of $\ell$.

Multiplying by $\tau_p^m$, summing over $m$, and using \cref{eq:rich-lind-one-step-series-parameter} now gives
\begin{align}
    \sum_{m=0}^{\infty}\tau_p^m
      \norm{[h^m]E_a(h)}_{2,\alpha\to2,\beta}
    &\leq\sum_{r=1}^{a}\binom{a-1}{r-1}\frac{2^{2a+r}}{r!}
      \sum_{m=2a+r}^{\infty}
      \frac{m^m}{(m-2a-r)!}\left(\frac{y}{\e}\right)^m\\
    &\leq\sum_{r=1}^{a}\binom{a-1}{r-1}\frac{2^{2a+r}}{r!}
      \sum_{m=2a+r}^{\infty}\frac{m!}{(m-2a-r)!}y^m\\
    &=\sum_{r=1}^{a}\binom{a-1}{r-1}\frac{(2y)^{2a+r}}{r!}
      \sum_{m=2a+r}^{\infty}\frac{m!}{(m-2a-r)!}y^{m-2a-r}.
    \label{eq:rich-lind-correction-taylor-series-bound}
\end{align}
The second inequality uses $m^m\leq\e^m m!$ from \cref{eq:rich-lind-factorial-power-bound}, and the equality extracts $y^{2a+r}$ from the inner sum. All terms are nonnegative. For $0\leq y<1$, differentiating the convergent geometric series gives
\begin{align}
    \sum_{m=2a+r}^{\infty}\frac{m!}{(m-2a-r)!}y^{m-2a-r}
    &=\frac{\mathrm d^{2a+r}}{\mathrm dy^{2a+r}}
      \sum_{m=0}^{\infty}y^m\\
    &=\frac{\mathrm d^{2a+r}}{\mathrm dy^{2a+r}}\frac{1}{1-y}
    =\frac{(2a+r)!}{(1-y)^{2a+r+1}}.
    \label{eq:rich-lind-factorial-geometric-series}
\end{align}
Substitution into \cref{eq:rich-lind-correction-taylor-series-bound} yields
\begin{align}
    \sum_{m=0}^{\infty}\tau_p^m
      \norm{[h^m]E_a(h)}_{2,\alpha\to2,\beta}
    &\leq\sum_{r=1}^{a}\binom{a-1}{r-1}
      \frac{(2y)^{2a+r}}{r!}
      \frac{(2a+r)!}{(1-y)^{2a+r+1}}\\
    &=\frac{1}{1-y}\sum_{r=1}^{a}\binom{a-1}{r-1}
      \frac{(2a+r)!}{r!}
      \left(\frac{2y}{1-y}\right)^{2a+r}.
    \label{eq:rich-lind-correction-taylor-sum}
\end{align}
Since $1\leq r\leq a\leq p$, the factorial and the remaining scalar ratio satisfy
\begin{align}
    (2a+r)!
    &=\prod_{s=1}^{2a+r}s
    \leq(2a+r)^{2a+r}
    \leq[3(p+1)]^{2a+r},\\
    \frac{6(p+1)y}{1-y}
    &=\frac{6}{1024(1-y)}
    \leq\frac{6}{1023}<\frac{1}{16}.
    \label{eq:rich-lind-retained-correction-smallness}
\end{align}
Using these inequalities in \cref{eq:rich-lind-correction-taylor-sum}, followed by $r\geq1$ and $1/r!\leq1$, gives
\begin{align}
    \sum_{m=0}^{\infty}\tau_p^m
      \norm{[h^m]E_a(h)}_{2,\alpha\to2,\beta}
    &\leq\frac{1}{1-y}\sum_{r=1}^{a}\binom{a-1}{r-1}
      \frac{1}{r!}
      \left(\frac{6(p+1)y}{1-y}\right)^{2a+r}\\
    &\leq\frac{1}{1-y}\sum_{r=1}^{a}\binom{a-1}{r-1}
      \frac{1}{r!16^{2a+r}}\\
    &\leq\frac{1}{(1-y)16^{2a+1}}
      \sum_{r=1}^{a}\binom{a-1}{r-1}
    =\frac{2^{a-1}}{(1-y)16^{2a+1}}.
    \label{eq:rich-lind-single-correction-taylor-sum}
\end{align}
Summing over the retained orders therefore gives
\begin{align}
    \sum_{a=1}^{p}\sum_{m=0}^{\infty}\tau_p^m
      \norm{[h^m]E_a(h)}_{2,\alpha\to2,\beta}
    &\leq\frac{1}{1-y}\sum_{a=1}^{p}
      \frac{2^{a-1}}{16^{2a+1}}\\
    &=\frac{1}{(1-y)16^3}\sum_{a=1}^{p}
      \left(\frac{1}{128}\right)^{a-1}\\
    &\leq\frac{1024}{1023}\frac{1}{16^3}
      \sum_{j=0}^{\infty}\left(\frac{1}{128}\right)^j
    =\frac{1024}{1023}\frac{1}{16^3}\frac{128}{127}<1.
    \label{eq:rich-lind-retained-correction-taylor-sum}
\end{align}
For $p=0$, the correction sum is empty and equals zero. Substituting \cref{eq:rich-lind-step-taylor-coefficient-sum,eq:rich-lind-exact-taylor-coefficient-sum,eq:rich-lind-retained-correction-taylor-sum} into \cref{eq:rich-lind-one-step-tail-bound} concludes the estimate for $h>0$:
\begin{align}
    \norm{R_{p+1,1}(h)}_{2,\alpha\to2,\beta}
    &\leq\left(\frac{h}{\tau_p}\right)^{2p+3}
      \left(\frac{2}{1-y}+1\right)\\
    &\leq\left(\frac{h}{\tau_p}\right)^{2p+3}
      \left(\frac{2048}{1023}+1\right)
    \leq6\left(\frac{h}{\tau_p}\right)^{2p+3}.
    \label{eq:rich-lind-one-step-remainder-conclusion}
\end{align}
At $h=0$, \cref{eq:rich-lind-correction-minimum-degree} gives
\begin{align}
    R_{p+1,1}(0)=\mathsf{S}_0-\e^{0\mathcal L}-\sum_{a=1}^{p}E_a(0)=I-I=0,
\end{align}
which proves the bound also at zero.
\end{proof}

\subsection{A bound on the remainder term}
\label{supp:richardson-remainder-bound}

We now combine the one-step remainder bound with the bounds on the correction maps to control the remainder after $N$ steps. A telescoping sum expresses this error in terms of one-step remainders acting on the retained correction maps.

\begin{theorem}[Finite circuit expansion and its error]
\label{thm:rich-lind-finite-circuit-error}
Let $t\geq0$, let $q\geq1$ be an integer, and use $\theta_w$ and $x$ from \cref{eq:quantum-optimal-weight-parameter,eq:rich-lind-full-circuit-scale}. For every integer $N\geq\max\{1,qx\}$, the remainder in \cref{eq:rich-lind-finite-expansion} satisfies
\begin{align}
    \norm{R_{q,N}(t)}_{2,\theta_w\to2,0}
    \leq6q\left(\frac{\sqrt{\e}\,qx\sqrt{1+x}}{N}\right)^{2q}.
    \label{eq:rich-lind-full-circuit-remainder-bound}
\end{align}
\end{theorem}
\begin{proof}
At $t=0$, we have $\mathsf{S}_0^N=I$ and $E_\ell(0)=0$ by \cref{eq:rich-lind-correction-minimum-degree}, so $R_{q,N}(0)=0$. Suppose $t>0$ and set $h=t/N$. To compare the circuit with the truncated expansion at intermediate time $s$, we keep the step size $h$ fixed as $s$ varies. Each $E_a(s)$ contains the factor $s^{2a}$ by \cref{eq:rich-lind-correction-maps}, so replacing this factor by $h^{2a}$ gives
\begin{align}
    F_q(s,h):=\e^{s\mathcal L}
      +\sum_{a=1}^{q-1}\left(\frac{h}{s}\right)^{2a}E_a(s),
    \qquad s>0.
    \label{eq:rich-lind-truncated-evolution}
\end{align}
The minimum-degree property \cref{eq:rich-lind-correction-minimum-degree} gives $\lim_{s\to0}E_a(s)/s^{2a}=0$ and hence $F_q(0,h)=I$. We use these continuous extensions at $s=0$ throughout the proof. At the final time, $h/t=1/N$, and hence
\begin{align}
    R_{q,N}(t)
    =\mathsf{S}_h^N-F_q(t,h).
    \label{eq:rich-lind-full-circuit-target}
\end{align}
The difference accumulated over the circuit is a sum of local differences. Indeed, shifting the index in the second sum gives
\begin{align}
    &\sum_{j=0}^{N-1}\mathsf{S}_h^{N-1-j}
      \bigl[\mathsf{S}_hF_q(jh,h)-F_q((j+1)h,h)\bigr]\\
    &\quad=\sum_{j=0}^{N-1}\mathsf{S}_h^{N-j}F_q(jh,h)
      -\sum_{j=1}^{N}\mathsf{S}_h^{N-j}F_q(jh,h)\\
    &\quad=\mathsf{S}_h^NF_q(0,h)-F_q(Nh,h)
      =\mathsf{S}_h^N-F_q(t,h).
    \label{eq:rich-lind-full-circuit-telescope}
\end{align}
The terms with $1\leq j\leq N-1$ cancel. To bound the remaining circuit factors, apply \cref{eq:weighted-local-H-growth} at weight zero and use the color factorization \cref{eq:exact-color-factorization}. For the dissipative factors, use the diagonal action \cref{eq:local-pauli-diagonal-noise} and the nonnegative rates in \cref{eq:local-noise-window}. For every $\delta\geq0$ these give
\begin{align}
    \norm{\e^{\delta\mathcal H^{(c)}}}_{2,0\to2,0}
    &\leq\prod_{a\in\mathcal C_c}
      \norm{\e^{\delta\mathcal H_a}}_{2,0\to2,0}\leq1,\\
    \norm{\e^{\delta\mathcal D^{(c)}}}_{2,0\to2,0}
    &=\max_{P\in\cP_n}
      \exp\!\left[-\delta\sum_{a\in\mathcal C_c}\lambda_a(P)\right]\leq1.
    \label{eq:rich-lind-color-unweighted-contraction}
\end{align}
The step in \cref{eq:rich-lind-color-step,eq:rich-lind-symmetric-step} is a composition of these contractions. Applying \cref{eq:weighted-induced-submultiplicativity} therefore yields
\begin{align}
    \norm{\mathsf{S}_h^{N-1-j}}_{2,0\to2,0}
    \leq\norm{\mathsf{S}_h}_{2,0\to2,0}^{N-1-j}
    \leq1,\qquad 0\leq j\leq N-1.
    \label{eq:rich-lind-step-unweighted-contraction}
\end{align}
Taking norms in \cref{eq:rich-lind-full-circuit-telescope}, using the triangle inequality and then \cref{eq:rich-lind-step-unweighted-contraction}, gives
\begin{align}
    \norm{R_{q,N}(t)}_{2,\theta_w\to2,0}
    &\leq\sum_{j=0}^{N-1}
      \norm{\mathsf{S}_h^{N-1-j}}_{2,0\to2,0}
      \norm{\mathsf{S}_hF_q(jh,h)-F_q((j+1)h,h)}_{2,\theta_w\to2,0}\\
    &\leq\sum_{j=0}^{N-1}
      \norm{\mathsf{S}_hF_q(jh,h)-F_q((j+1)h,h)}_{2,\theta_w\to2,0}.
    \label{eq:rich-lind-full-circuit-local-reduction}
\end{align}

We next relate the retained corrections at $s+h$ to those at $s$ and $h$, so that subtracting $F_q(s+h,h)$ from $\mathsf{S}_hF_q(s,h)$ leaves the one-step remainders. The coefficient calculation in \cref{eq:rich-lind-insertion-powers,eq:rich-lind-collected-expansion}, with the normalization in \cref{eq:rich-lind-correction-maps}, gives
\begin{align}
    [\xi^0]\e^{s(\log \mathsf{S}_\xi)/\xi}
    &=\e^{s\mathcal L},\\
    [\xi^{2\ell}]\e^{s(\log \mathsf{S}_\xi)/\xi}
    &=\frac{E_\ell(s)}{s^{2\ell}},\qquad \ell\geq1,\quad s>0.
    \label{eq:rich-lind-normalized-correction-coefficient}
\end{align}
Only even powers occur by \cref{eq:rich-lind-log-expansion}. At a fixed finite system and for $\xi$ near zero, the two exponentials below have the same generator, so
\begin{align}
    \e^{h(\log \mathsf{S}_\xi)/\xi}\e^{s(\log \mathsf{S}_\xi)/\xi}
    =\e^{(s+h)(\log \mathsf{S}_\xi)/\xi}.
    \label{eq:rich-lind-coefficient-semigroup}
\end{align}
Taking the coefficient of $\xi^{2\ell}$ and separating the terms $a=0$ and $a=\ell$ yields, for $s>0$,
\begin{align}
    \frac{E_\ell(s+h)}{(s+h)^{2\ell}}
    &=\sum_{a=0}^{\ell}
      \left([\xi^{2a}]\e^{h(\log \mathsf{S}_\xi)/\xi}\right)
      \left([\xi^{2(\ell-a)}]\e^{s(\log \mathsf{S}_\xi)/\xi}\right)\\
    &=\e^{h\mathcal L}\frac{E_\ell(s)}{s^{2\ell}}
      +E_\ell(h)\frac{\e^{s\mathcal L}}{h^{2\ell}}
      +\sum_{a=1}^{\ell-1}
      \frac{E_a(h)E_{\ell-a}(s)}{h^{2a}s^{2(\ell-a)}}.
    \label{eq:rich-lind-correction-convolution}
\end{align}
Substitute this identity into \cref{eq:rich-lind-truncated-evolution}. Setting $b=\ell-a$ in the double sum and grouping terms with the same $E_b(s)$ gives
\begin{align}
    F_q(s+h,h)
    &=\e^{h\mathcal L}\e^{s\mathcal L}
      +\sum_{\ell=1}^{q-1}
      \left[\left(\frac hs\right)^{2\ell}\e^{h\mathcal L}E_\ell(s)
      +E_\ell(h)\e^{s\mathcal L}\right]\\
    &\quad+\sum_{\substack{a,b\geq1\\a+b<q}}
      \left(\frac hs\right)^{2b}E_a(h)E_b(s)\\
    &=\left[\e^{h\mathcal L}+\sum_{a=1}^{q-1}E_a(h)\right]\e^{s\mathcal L}
      +\sum_{b=1}^{q-1}\left(\frac hs\right)^{2b}
      \left[\e^{h\mathcal L}+\sum_{a=1}^{q-1-b}E_a(h)\right]E_b(s).
    \label{eq:rich-lind-truncated-convolution}
\end{align}
The coefficient identity \cref{eq:rich-lind-one-step-remainder-coefficient} and the convergent series \cref{eq:rich-lind-one-step-tail} in the proof of \cref{lem:rich-lind-one-step-remainder} identify the last bracket as $\mathsf{S}_h-R_{q-b,1}(h)$ and the first bracket as $\mathsf{S}_h-R_{q,1}(h)$. Subtracting from $\mathsf{S}_hF_q(s,h)$ therefore gives
\begin{align}
    \mathsf{S}_hF_q(s,h)-F_q(s+h,h)
    =R_{q,1}(h)\e^{s\mathcal L}
      +\sum_{b=1}^{q-1}\left(\frac hs\right)^{2b}
      R_{q-b,1}(h)E_b(s).
    \label{eq:rich-lind-local-remainder-identity}
\end{align}

In each product, the correction map acts first. We bound it from weight $\theta_w$ to $\theta_w/2$, then bound the one-step remainder from $\theta_w/2$ to zero. The triangle and composition inequalities give
\begin{align}
    &\norm{\mathsf{S}_hF_q(s,h)-F_q(s+h,h)}_{2,\theta_w\to2,0}\\
    &\quad\leq\norm{R_{q,1}(h)}_{2,\theta_w/2\to2,0}
      \norm{\e^{s\mathcal L}}_{2,\theta_w\to2,\theta_w/2}\\
    &\qquad+\sum_{b=1}^{q-1}\left(\frac hs\right)^{2b}
      \norm{R_{q-b,1}(h)}_{2,\theta_w/2\to2,0}
      \norm{E_b(s)}_{2,\theta_w\to2,\theta_w/2}.
    \label{eq:rich-lind-local-remainder-norm-reduction}
\end{align}
For $0\leq b\leq q-1$, the radius in \cref{eq:rich-lind-one-step-radius}, with $p=q-1-b$, $\alpha=\theta_w/2$, and $\beta=0$, satisfies
\begin{align}
    \tau_{q-1-b}^{-1}
    &=\frac{2048\chi(2J\e^{k\theta_w}+\bar\gamma)(q-b)}{\theta_w}
      =\frac{(q-b)x}{t},\\
    \frac{h}{\tau_{q-1-b}}
    &=\frac{(q-b)x}{N}\leq\frac{qx}{N}\leq1.
    \label{eq:rich-lind-local-radius-condition}
\end{align}
Thus the one-step bound \cref{eq:rich-lind-one-step-remainder-bound} gives
\begin{align}
    \norm{R_{q-b,1}(h)}_{2,\theta_w/2\to2,0}
    \leq6\left(\frac{h}{\tau_{q-1-b}}\right)^{2q+1-2b}
    \leq6\left(\frac{qx}{N}\right)^{2q+1-2b}.
    \label{eq:rich-lind-local-remainder-used}
\end{align}
For $1\leq r\leq b\leq q-1$, we have $2b+r\leq3q$ and
\begin{align}
    \frac{4\chi(2J\e^{k\theta_w}+\bar\gamma)(2b+r)}{\e\theta_w}
    \leq\frac{12\chi(2J\e^{k\theta_w}+\bar\gamma)q}{\e\theta_w}
    \leq\frac{qx}{t}.
    \label{eq:rich-lind-correction-scale-comparison}
\end{align}
Applying \cref{eq:rich-lind-correction-bound} with $\alpha=\theta_w$, $\beta=\theta_w/2$, and time $s$ therefore yields
\begin{align}
    \norm{E_b(s)}_{2,\theta_w\to2,\theta_w/2}
    &\leq s^{2b}\sum_{r=1}^{b}\binom{b-1}{r-1}\frac{s^r}{r!}
      \left[\frac{4\chi(2J\e^{k\theta_w}+\bar\gamma)(2b+r)}{\e\theta_w}\right]^{2b+r}\\
    &\leq s^{2b}\sum_{r=1}^{b}\binom{b-1}{r-1}\frac{s^r}{r!}
      \left(\frac{qx}{t}\right)^{2b+r}.
    \label{eq:rich-lind-full-circuit-correction-used}
\end{align}
The prefactor $(h/s)^{2b}s^{2b}=h^{2b}$ combines with the power $h^{2q+1-2b}$ from the one-step remainder to give $h^{2b}h^{2q+1-2b}=h^{2q+1}$ for every $b$. Substituting \cref{eq:rich-lind-local-remainder-used,eq:rich-lind-full-circuit-correction-used} into each $b\geq1$ term of \cref{eq:rich-lind-local-remainder-norm-reduction}, and using $h=t/N$, gives
\begin{align}
    &\left(\frac hs\right)^{2b}
      \norm{R_{q-b,1}(h)}_{2,\theta_w/2\to2,0}
      \norm{E_b(s)}_{2,\theta_w\to2,\theta_w/2}\\
    &\quad\leq6h^{2b}\left(\frac{qx}{N}\right)^{2q+1-2b}
      \sum_{r=1}^{b}\binom{b-1}{r-1}\frac{s^r}{r!}
      \left(\frac{qx}{t}\right)^{2b+r}\\
    &\quad=6\left(\frac{qx}{N}\right)^{2q+1}
      \sum_{r=1}^{b}\binom{b-1}{r-1}\frac{(qx s/t)^r}{r!}.
    \label{eq:rich-lind-local-correction-product-bound}
\end{align}
For the term without a correction map, \cref{eq:supp-main-weight-regularity,eq:theta-monotonicity} imply
\begin{align}
    \norm{\e^{s\mathcal L}}_{2,\theta_w\to2,\theta_w/2}
    &\leq\norm{\e^{s\mathcal L}}_{2,\theta_w\to2,\theta_w}\leq1,\\
    \norm{R_{q,1}(h)}_{2,\theta_w/2\to2,0}
      \norm{\e^{s\mathcal L}}_{2,\theta_w\to2,\theta_w/2}
    &\leq6\left(\frac{qx}{N}\right)^{2q+1}.
    \label{eq:rich-lind-local-exact-product-bound}
\end{align}
Combining \cref{eq:rich-lind-local-correction-product-bound,eq:rich-lind-local-exact-product-bound} in \cref{eq:rich-lind-local-remainder-norm-reduction} and using $0\leq s\leq t$ now gives
\begin{align}
    &\norm{\mathsf{S}_hF_q(s,h)-F_q(s+h,h)}_{2,\theta_w\to2,0}\\
    &\quad\leq6\left(\frac{qx}{N}\right)^{2q+1}
      \left[1+\sum_{b=1}^{q-1}\sum_{r=1}^{b}
      \binom{b-1}{r-1}\frac{(qx)^r}{r!}\right].
    \label{eq:rich-lind-uniform-local-remainder-bound}
\end{align}

It remains to sum the retained orders. Pascal's identity gives, for $1\leq r\leq q-1$,
\begin{align}
    \sum_{b=r}^{q-1}\binom{b-1}{r-1}
    =\sum_{b=r}^{q-1}\left[\binom br-\binom{b-1}r\right]
    =\binom{q-1}{r}.
    \label{eq:rich-lind-order-sum-identity}
\end{align}
Interchanging the finite sums and applying this identity yields
\begin{align}
    1+\sum_{b=1}^{q-1}\sum_{r=1}^{b}\binom{b-1}{r-1}\frac{(qx)^r}{r!}
    &=1+\sum_{r=1}^{q-1}\frac{(qx)^r}{r!}
      \sum_{b=r}^{q-1}\binom{b-1}{r-1}\\
    &=\sum_{r=0}^{q-1}\binom{q-1}{r}\frac{(qx)^r}{r!}.
    \label{eq:rich-lind-regrouped-order-sum}
\end{align}
Each term of the nonnegative exponential series obeys $q^r/r!\leq\sum_{m=0}^{\infty}q^m/m!=\e^q$. The binomial theorem then gives
\begin{align}
    \sum_{r=0}^{q-1}\binom{q-1}{r}\frac{(qx)^r}{r!}
    \leq\e^q\sum_{r=0}^{q-1}\binom{q-1}{r}x^r
    =\e^q(1+x)^{q-1}.
    \label{eq:rich-lind-full-circuit-order-sum-bound}
\end{align}
Finally, summing $N$ local bounds with the common power $h^{2q+1}$ gives $Nh^{2q+1}=t^{2q+1}/N^{2q}$. Substituting \cref{eq:rich-lind-uniform-local-remainder-bound,eq:rich-lind-regrouped-order-sum,eq:rich-lind-full-circuit-order-sum-bound} into \cref{eq:rich-lind-full-circuit-local-reduction} at $s=jh\leq t$ makes this dependence explicit:
\begin{align}
    \norm{R_{q,N}(t)}_{2,\theta_w\to2,0}
    &\leq6N\left(\frac{qx}{N}\right)^{2q+1}\e^q(1+x)^{q-1}\\
    &=\frac{6(qx)^{2q+1}}{N^{2q}}\e^q(1+x)^{q-1}\\
    &=6q\left(\frac{\sqrt{\e}\,qx\sqrt{1+x}}{N}\right)^{2q}
      \frac{x}{1+x}\\
    &\leq6q\left(\frac{\sqrt{\e}\,qx\sqrt{1+x}}{N}\right)^{2q},
\end{align}
which proves the claimed bound.
\end{proof}

\subsection{Richardson extrapolation proof}
\label{supp:richardson-extrapolation-proof}

We now choose integer step counts and coefficients that cancel the correction terms in \cref{eq:rich-lind-finite-expansion}. Their absolute sum controls the amplification of the remainders bounded in \cref{eq:rich-lind-full-circuit-remainder-bound}. The construction follows the interpolation points and rounding prescription of Low, Kliuchnikov, and Wiebe~\cite[Eqs.~(5), (8), and (10)]{low2019well}.

\begin{lemma}[Integer step counts and extrapolation coefficients]
\label{lem:rich-lind-extrapolation-counts}
Let $q\geq2$ be an integer and define
\begin{align}
    \varphi_j:=\frac{\pi(2j-1)}{8q},\qquad
    \widehat\kappa_j:=\frac{\sqrt8\,q}{\pi\sin\varphi_j},\qquad
    \kappa_j:=\left\lceil\widehat\kappa_j\right\rceil,
    \qquad 1\leq j\leq q.
    \label{eq:rich-lind-relative-step-counts}
\end{align}
For a positive integer $b$, choose the step counts and coefficients
\begin{align}
    N_j:=b\kappa_j,\qquad
    \omega_j:=\prod_{\substack{1\leq\ell\leq q\\\ell\neq j}}
      \frac{N_j^2}{N_j^2-N_\ell^2}.
    \label{eq:rich-lind-extrapolation-counts-weights}
\end{align}
The counts $N_j$ are distinct positive integers, the coefficients $\omega_j$ are real and independent of $b$, and
\begin{align}
    q&\leq\kappa_j\leq3q^2,\qquad 1\leq j\leq q,
    \label{eq:rich-lind-relative-step-range}\\
    \sum_{j=1}^{q}\omega_j&=1,\qquad
    \sum_{j=1}^{q}\omega_jN_j^{-2a}=0,\qquad 1\leq a<q,
    \label{eq:rich-lind-extrapolation-moments}\\
    \sum_{j=1}^{q}|\omega_j|&\leq48q^3.
    \label{eq:rich-lind-extrapolation-coefficient-bound}
\end{align}
\end{lemma}
\begin{proof}
For $0<u<\pi/4$, concavity of sine and its chord between $0$ and $\pi/4$ give
\begin{align}
    \frac{2\sqrt2}{\pi}u\leq\sin u\leq\frac1{\sqrt2}.
\end{align}
Since $0<\varphi_j<\pi/4$, substitution into \cref{eq:rich-lind-relative-step-counts} yields
\begin{align}
    \widehat\kappa_j
    &\geq\frac{4q}{\pi}>q,
    \label{eq:rich-lind-unrounded-step-lower-bound}\\
    \widehat\kappa_j
    &\leq\frac{\sqrt8\,q}{2\sqrt2\,\varphi_j}
      =\frac{8q^2}{\pi(2j-1)},\\
    q<\kappa_j
    &\leq\widehat\kappa_j+1
      \leq\frac8\pi q^2+1\leq3q^2.
    \label{eq:rich-lind-rounded-step-range}
\end{align}
The last inequality uses $q\geq2$. To check that rounding preserves distinctness, note that $\csc u\cot u=\cos u/\sin^2u>\sqrt2$ for $0<u<\pi/4$. Consecutive unrounded counts therefore satisfy
\begin{align}
    \widehat\kappa_j-\widehat\kappa_{j+1}
    &=\frac{\sqrt8\,q}{\pi}
      \int_{\varphi_j}^{\varphi_{j+1}}\csc u\cot u\,\mathrm du\\
    &>\frac{\sqrt8\,q}{\pi}\sqrt2\,
      \frac{\pi}{4q}=1,\qquad 1\leq j<q,
    \label{eq:rich-lind-unrounded-step-separation}\\
    \kappa_j-\kappa_{j+1}
    &=\left\lceil\widehat\kappa_j\right\rceil
      -\left\lceil\widehat\kappa_{j+1}\right\rceil\geq1.
    \label{eq:rich-lind-step-count-separation}
\end{align}
Multiplication by the positive integer $b$ preserves these strict inequalities.

The coefficients are the values at zero of the Lagrange polynomials for the distinct points $N_j^{-2}$:
\begin{align}
    L_j(v):=\prod_{\substack{1\leq\ell\leq q\\\ell\neq j}}
      \frac{v-N_\ell^{-2}}{N_j^{-2}-N_\ell^{-2}},\qquad
    L_j(N_j^{-2})=1,\qquad L_j(N_i^{-2})=0\quad(i\neq j),
    \label{eq:rich-lind-lagrange-polynomials}\\
    L_j(0)
    =\prod_{\ell\neq j}\frac{-N_\ell^{-2}}{N_j^{-2}-N_\ell^{-2}}
    =\prod_{\ell\neq j}\frac{N_j^2}{N_j^2-N_\ell^2}
    =\omega_j.
    \label{eq:rich-lind-lagrange-weights}
\end{align}
For every polynomial $F$ of degree below $q$, the difference $F(v)-\sum_{j=1}^{q}F(N_j^{-2})L_j(v)$ also has degree below $q$ and vanishes at all $q$ interpolation points. It is therefore zero, and evaluation at $v=0$ gives
\begin{align}
    F(v)&=\sum_{j=1}^{q}F(N_j^{-2})L_j(v),\\
    F(0)&=\sum_{j=1}^{q}\omega_jF(N_j^{-2}).
    \label{eq:rich-lind-polynomial-extrapolation}
\end{align}
Taking $F(v)=1$ and then $F(v)=v^a$ gives the cancellation conditions:
\begin{align}
    1=\sum_{j=1}^{q}\omega_j,\qquad
    0=\sum_{j=1}^{q}\omega_j(N_j^{-2})^a
      =\sum_{j=1}^{q}\omega_jN_j^{-2a},\qquad 1\leq a<q.
\end{align}
The common multiplier cancels in each factor of the coefficient formula:
\begin{align}
    \omega_j
    =\prod_{\ell\neq j}\frac{b^2\kappa_j^2}{b^2\kappa_j^2-b^2\kappa_\ell^2}
    =\prod_{\ell\neq j}\frac{\kappa_j^2}{\kappa_j^2-\kappa_\ell^2}.
    \label{eq:rich-lind-multiplier-cancellation}
\end{align}

To bound the absolute coefficients, we first use the unrounded counts and then control the change caused by rounding. Extend $\varphi_j=\pi(2j-1)/(8q)$ to $1\leq j\leq2q$ and set $v_j:=\sin^2\varphi_j$. For $1\leq j\leq q$, the coefficients before rounding are
\begin{align}
    \widehat\omega_j
    :=\prod_{\substack{1\leq\ell\leq q\\\ell\neq j}}
      \frac{\widehat\kappa_j^2}{\widehat\kappa_j^2-\widehat\kappa_\ell^2}
    =\prod_{\substack{1\leq\ell\leq q\\\ell\neq j}}
      \frac{v_\ell}{v_\ell-v_j}.
    \label{eq:rich-lind-unrounded-weights}
\end{align}
The second equality follows by substituting $\widehat\kappa_j^2=8q^2/(\pi^2v_j)$ and cancelling the common factor. The extended points are the roots of the degree-$2q$ polynomial $F(v):=T_{2q}(1-2v)$, where $T_m(\cos u)=\cos(mu)$ defines the Chebyshev polynomial. Indeed,
\begin{align}
    F(v_j)
    &=T_{2q}(\cos(2\varphi_j))
      =\cos(4q\varphi_j)
      =\cos\!\left(\frac{\pi(2j-1)}2\right)=0,\\
    F(0)&=T_{2q}(1)=1.
\end{align}
Differentiating the defining identity for $T_{2q}$ gives
\begin{align}
    -\sin u\,T_{2q}'(\cos u)&=-2q\sin(2qu),\\
    |F'(v_j)|
    &=2|T_{2q}'(\cos(2\varphi_j))|
      =\frac{4q|\sin(4q\varphi_j)|}{\sin(2\varphi_j)}
      =\frac{4q}{\sin(2\varphi_j)}.
    \label{eq:rich-lind-chebyshev-root-derivative}
\end{align}
For all $2q$ roots, the $j$th Lagrange polynomial is $F(v)/[(v-v_j)F'(v_j)]$. Its value at zero therefore satisfies
\begin{align}
    \prod_{\substack{1\leq\ell\leq2q\\\ell\neq j}}
      \frac{v_\ell}{|v_\ell-v_j|}
    =\frac{|F(0)|}{v_j|F'(v_j)|}
    =\frac{\sin(2\varphi_j)}{4q\sin^2\varphi_j}
    =\frac{\cot\varphi_j}{2q}.
    \label{eq:rich-lind-full-chebyshev-weight}
\end{align}
For $j\leq q<\ell$, we have $v_\ell>v_j$, so each additional factor $v_\ell/(v_\ell-v_j)$ is greater than one. Combining \cref{eq:rich-lind-unrounded-weights,eq:rich-lind-full-chebyshev-weight} gives
\begin{align}
    |\widehat\omega_j|
    \leq|\widehat\omega_j|\prod_{\ell=q+1}^{2q}\frac{v_\ell}{v_\ell-v_j}
    =\frac{\cot\varphi_j}{2q}.
\end{align}
Since $\tan u=\int_0^u\sec^2v\,\mathrm dv\geq u$ for $0<u<\pi/2$, we have $\cot u\leq1/u$. Summing the preceding bound yields
\begin{align}
    \sum_{j=1}^{q}|\widehat\omega_j|
    \leq\frac1{2q}\sum_{j=1}^{q}\cot\varphi_j
    \leq\frac4\pi\sum_{j=1}^{q}\frac1{2j-1}
    \leq\frac4\pi q.
    \label{eq:rich-lind-unrounded-coefficient-bound}
\end{align}

Write $\delta_j:=\kappa_j-\widehat\kappa_j$, so $0\leq\delta_j<1$. Factoring the squared differences in \cref{eq:rich-lind-multiplier-cancellation,eq:rich-lind-unrounded-weights} gives
\begin{align}
    \log\frac{|\omega_j|}{|\widehat\omega_j|}
    &=2(q-1)\log\frac{\kappa_j}{\widehat\kappa_j}
      +\sum_{\ell\neq j}\log
      \frac{|\widehat\kappa_j-\widehat\kappa_\ell|}{|\kappa_j-\kappa_\ell|}\\
    &\quad+\sum_{\ell\neq j}\log
      \frac{\widehat\kappa_j+\widehat\kappa_\ell}{\kappa_j+\kappa_\ell}.
    \label{eq:rich-lind-rounding-log-ratio}
\end{align}
The numerator contribution is bounded using $\log(1+u)\leq u$ and \cref{eq:rich-lind-unrounded-step-lower-bound}:
\begin{align}
    2(q-1)\log\frac{\kappa_j}{\widehat\kappa_j}
    &=2(q-1)\log\!\left(1+\frac{\delta_j}{\widehat\kappa_j}\right)
      \leq\frac{2(q-1)\delta_j}{\widehat\kappa_j}
      \leq\frac{2q}{\widehat\kappa_j}\leq\frac\pi2,
    \label{eq:rich-lind-rounding-numerator-bound}\\
    \sum_{\ell\neq j}\log
      \frac{\widehat\kappa_j+\widehat\kappa_\ell}{\kappa_j+\kappa_\ell}
    &\leq0.
    \label{eq:rich-lind-rounding-simple-terms}
\end{align}
To control the differences in the remaining sum, interpolate between the unrounded and rounded counts:
\begin{align}
    \kappa_j(u):=\widehat\kappa_j+u\delta_j,\qquad 0\leq u\leq1.
\end{align}
Both endpoint sequences have consecutive gaps at least one by \cref{eq:rich-lind-unrounded-step-separation,eq:rich-lind-step-count-separation}. Thus
\begin{align}
    \kappa_j(u)-\kappa_{j+1}(u)
    &=(1-u)(\widehat\kappa_j-\widehat\kappa_{j+1})
      +u(\kappa_j-\kappa_{j+1})\geq1,\\
    |\kappa_j(u)-\kappa_\ell(u)|
    &=\sum_{r=\min\{j,\ell\}}^{\max\{j,\ell\}-1}
      [\kappa_r(u)-\kappa_{r+1}(u)]\geq|j-\ell|.
    \label{eq:rich-lind-interpolated-count-separation}
\end{align}
For $j\neq\ell$, differentiation and integration then give
\begin{align}
    \left|\frac{\mathrm d}{\mathrm du}
      \log|\kappa_j(u)-\kappa_\ell(u)|\right|
    &=\frac{|\delta_j-\delta_\ell|}{|\kappa_j(u)-\kappa_\ell(u)|}
      \leq\frac1{|j-\ell|},\\
    \log\frac{|\widehat\kappa_j-\widehat\kappa_\ell|}{|\kappa_j-\kappa_\ell|}
    &=-\int_0^1\frac{\mathrm d}{\mathrm du}
      \log|\kappa_j(u)-\kappa_\ell(u)|\,\mathrm du\\
    &\leq\int_0^1\frac{\mathrm du}{|j-\ell|}
      =\frac1{|j-\ell|}.
    \label{eq:rich-lind-rounding-difference-bound}
\end{align}
For $m\geq1$, the harmonic sum satisfies
\begin{align}
    \sum_{r=1}^{m}\frac1r
    =1+\sum_{r=2}^{m}\frac1r
    \leq1+\sum_{r=2}^{m}\int_{r-1}^{r}\frac{\mathrm du}{u}
    =1+\log m.
\end{align}
Empty sums are zero, so summing \cref{eq:rich-lind-rounding-difference-bound} over $\ell\neq j$ gives
\begin{align}
    \sum_{\ell\neq j}\log
      \frac{|\widehat\kappa_j-\widehat\kappa_\ell|}{|\kappa_j-\kappa_\ell|}
    &\leq\sum_{\ell\neq j}\frac1{|j-\ell|}\\
    &=\sum_{r=1}^{j-1}\frac1r+\sum_{r=1}^{q-j}\frac1r
      \leq2+2\log q.
    \label{eq:rich-lind-rounding-harmonic-bound}
\end{align}
Substituting \cref{eq:rich-lind-rounding-numerator-bound,eq:rich-lind-rounding-simple-terms,eq:rich-lind-rounding-harmonic-bound} into \cref{eq:rich-lind-rounding-log-ratio} and exponentiating yields
\begin{align}
    \log\frac{|\omega_j|}{|\widehat\omega_j|}
    &\leq\frac\pi2+2+2\log q,\\
    |\omega_j|&\leq\e^{\pi/2+2}q^2|\widehat\omega_j|,\\
    \sum_{j=1}^{q}|\omega_j|
    &\leq\e^{\pi/2+2}q^2\sum_{j=1}^{q}|\widehat\omega_j|
      \leq\frac4\pi\e^{\pi/2+2}q^3<48q^3.
\end{align}
The last line uses \cref{eq:rich-lind-unrounded-coefficient-bound} and proves the coefficient bound.
\end{proof}

The common multiplier $b$ increases all step counts while leaving the extrapolation coefficients unchanged. 

\subsubsection{Proof of \cref{thm:rich-lind-maximum-depth}}
With the choices in \cref{eq:rich-lind-extrapolation-choices}, we now combine the cancellation identities and remainder bound to prove \cref{thm:rich-lind-maximum-depth}. Define
\begin{align}
    x:=\frac{2048\chi(2J\e^{k\theta_w}+\bar\gamma)t}{\theta_w}, \quad 
    b:=\left\lceil\max\{1,4\sqrt{\e}\,x\sqrt{1+x}\}\right\rceil\,,
    \label{eq:rich-lind-full-circuit-scale}
\end{align}
and denote
\begin{align}
     O^{[q]}&=\sum_{j=1}^q\omega_j \mathsf{S}_{t/N_j}^{N_j}(O)\,.
 \label{eq:output}
\end{align}
Then

\begin{proof}
The choice of $q$ gives $q\geq\log4096/\log4=6$, so \cref{lem:rich-lind-extrapolation-counts} applies. The lower bound on $\kappa_j$ in \cref{eq:rich-lind-relative-step-range} and the choice of $b$ give
\begin{align}
    N_j=b\kappa_j\geq bq
    \geq4\sqrt{\e}\,qx\sqrt{1+x}\geq qx,\qquad N_j\geq1.
    \label{eq:rich-lind-extrapolation-remainder-condition}
\end{align}
Thus the finite expansion \cref{eq:rich-lind-finite-expansion} has the remainder bound \cref{eq:rich-lind-full-circuit-remainder-bound} for every circuit. Substitution into the definition \cref{eq:output}, followed by the moment identities \cref{eq:rich-lind-extrapolation-moments}, gives
\begin{align}
    O^{[q]}-\e^{t\mathcal L}(O)
    &=\left(\sum_{j=1}^{q}\omega_j-1\right)\e^{t\mathcal L}(O)
      +\sum_{\ell=1}^{q-1}
      \left(\sum_{j=1}^{q}\omega_jN_j^{-2\ell}\right)E_\ell(t)(O)\\
    &\quad+\sum_{j=1}^{q}\omega_jR_{q,N_j}(t)(O)\\
    &=\sum_{j=1}^{q}\omega_jR_{q,N_j}(t)(O).
    \label{eq:rich-lind-extrapolation-cancellation}
\end{align}
Both scalar coefficients in parentheses vanish. To bound the remaining terms, first substitute \cref{eq:rich-lind-extrapolation-choices} and $N_j\geq bq$ into the ratio in the remainder estimate:
\begin{align}
    \frac{\sqrt{\e}\,qx\sqrt{1+x}}{N_j}
    &\leq\frac{\sqrt{\e}\,x\sqrt{1+x}}{b}\leq\frac14,\\
    \norm{R_{q,N_j}(t)}_{2,\theta_w\to2,0}
    &\leq6q\left(\frac14\right)^{2q}=6q\,16^{-q}.
    \label{eq:rich-lind-extrapolated-circuit-remainder}
\end{align}
For the input observable, \cref{eq:low-weight-weighted-norm,eq:quantum-optimal-weight-parameter} imply
\begin{align}
    w\theta_w&\leq\frac{w}{2k+w}\leq1,\\
    \norm{O}_{2,\theta_w}
    &\leq\e^{w\theta_w}\norm{O}_2\leq\e\norm{O}_2.
    \label{eq:rich-lind-extrapolation-input-weight}
\end{align}
Taking norms in \cref{eq:rich-lind-extrapolation-cancellation}, applying the triangle inequality and the induced-norm definition \cref{eq:weighted-two-induced-norm}, and then using \cref{eq:rich-lind-extrapolated-circuit-remainder,eq:rich-lind-extrapolation-input-weight,eq:rich-lind-extrapolation-coefficient-bound}, gives
\begin{align}
    \norm{O^{[q]}-\e^{t\mathcal L}(O)}_2
    &\leq\sum_{j=1}^{q}|\omega_j|\,
      \norm{R_{q,N_j}(t)(O)}_2\\
    &\leq\sum_{j=1}^{q}|\omega_j|\,
      \norm{R_{q,N_j}(t)}_{2,\theta_w\to2,0}\norm{O}_{2,\theta_w}\\
    &\leq6\e q\,16^{-q}\left(\sum_{j=1}^{q}|\omega_j|\right)\norm{O}_2\\
    &\leq288\e q^4\,16^{-q}\norm{O}_2.
    \label{eq:rich-lind-extrapolation-bias}
\end{align}
The exponential decay absorbs $q^4$. To see this explicitly, the first two values and the ratio of consecutive values satisfy
\begin{align}
    \frac{2^4}{4^2}=1,\qquad \frac{3^4}{4^3}=\frac{81}{64}<2,\\
    \frac{(q+1)^4/4^{q+1}}{q^4/4^q}
    =\frac{(1+1/q)^4}{4}
    \leq\frac{(4/3)^4}{4}=\frac{64}{81}<1,
    \qquad q\geq3.
\end{align}
Hence $q^4\leq2\cdot4^q$ for every integer $q\geq2$. Substitution into \cref{eq:rich-lind-extrapolation-bias}, followed by the choice of $q$, gives
\begin{align}
    \norm{O^{[q]}-\e^{t\mathcal L}(O)}_2
    &\leq576\e\,4^{-q}\norm{O}_2\\
    &\leq4096\,4^{-q}\norm{O}_2
      \leq4096\,4^{-\log(4096/\varepsilon)/\log4}\norm{O}_2
      =\varepsilon\norm{O}_2.
\end{align}

It remains to bound the depth of the longest circuit. The upper bound on $\kappa_j$ in \cref{eq:rich-lind-relative-step-range} gives
\begin{align}
    \max_{1\leq j\leq q}N_j\leq3bq^2.
    \label{eq:rich-lind-maximum-step-count}
\end{align}
For $x\geq0$, the inequalities $\sqrt{1+x}\leq1+\sqrt x$ and $x\leq1+x^{3/2}$ bound the multiplier as
\begin{align}
    b
    &\leq2+4\sqrt{\e}\,x\sqrt{1+x}\\
    &\leq2+4\sqrt{\e}\,(x+x^{3/2})
      \leq(2+8\sqrt{\e})(1+x^{3/2}).
    \label{eq:rich-lind-extrapolation-multiplier-bound}
\end{align}
The number of circuits obeys
\begin{align}
    q
    &\leq1+\frac{\log4096+\log(1/\varepsilon)}{\log4}
      =7+\frac{\log(1/\varepsilon)}{\log4}
      \leq7\log\!\left(\frac{\e}{\varepsilon}\right).
    \label{eq:rich-lind-extrapolation-order-bound}
\end{align}
Within one color, the local maps act on disjoint supports and can be performed in parallel, as in \cref{supp:quantum-resources}. Each $Q_{c,h}$ in \cref{eq:rich-lind-color-step} consists of three such layers, and $\mathsf{S}_h$ in \cref{eq:rich-lind-symmetric-step} contains $2\chi-1$ color blocks. Using \cref{eq:number-colors} and \cref{eq:rich-lind-maximum-step-count,eq:rich-lind-extrapolation-multiplier-bound,eq:rich-lind-extrapolation-order-bound}, we obtain
\begin{align}
    D_{\max}
    &\leq3C_k(2\chi-1)\max_jN_j
      \leq18C_k k d\,bq^2\\
    &\leq18C_k k d\,(2+8\sqrt{\e})(1+x^{3/2})
      \cdot49\log^2\!\left(\frac{\e}{\varepsilon}\right)\\
    &\leq C_k d(1+x^{3/2})\log^2\!\left(\frac{\e}{\varepsilon}\right).
    \label{eq:rich-lind-scaled-maximum-depth}
\end{align}
The last line absorbs factors depending only on $k$ into $C_k$. To express $x$ in the physical parameters, use \cref{eq:quantum-optimal-weight-parameter} and \cref{eq:number-colors}:
\begin{align}
    k\theta_w\leq\frac{k}{2k+w}\leq\frac12,
    \qquad \e^{k\theta_w}\leq\sqrt{\e},
    \qquad \chi\leq k(d-1)+1\leq kd,\\
    x
    =\frac{2048\chi(2J\e^{k\theta_w}+\bar\gamma)t}{\theta_w}
    \leq\frac{4096\sqrt{\e}\,kd(J+\bar\gamma)t}{\theta_w}.
    \label{eq:rich-lind-physical-scale-bound}
\end{align}
Substitution into \cref{eq:rich-lind-scaled-maximum-depth} yields
\begin{align}
    D_{\max}
    &\leq C_k d\left[1+
      \left(\frac{d(J+\bar\gamma)t}{\theta_w}\right)^{3/2}\right]
      \log^2\!\left(\frac{\e}{\varepsilon}\right)\\
    &=C_k\left[d+
      \frac{d^{5/2}(J+\bar\gamma)^{3/2}t^{3/2}}{\theta_w^{3/2}}\right]
      \log^2\!\left(\frac{\e}{\varepsilon}\right),
\end{align}
where factors depending only on $k$ have again been absorbed into $C_k$.
\end{proof}

\suppsection{Classical simulation using Pauli propagation}
\label{supp:classical}

We now describe a classical algorithm that propagates a sparse Pauli representation through the split color product formula and applies coefficient and weight truncation after every color layer. Coefficient truncation controls the number of stored Pauli coefficients, while weight truncation controls the cost of updating each retained Pauli operator. Recall that
\begin{align}
    \mathcal T_{\delta}^{(c)}
    =
    \e^{\delta\mathcal D^{(c)}}\e^{\delta\mathcal H^{(c)}},
    \qquad
    \Phi_\delta^{\mathrm{Tr}}
    =
    \mathcal T_{\delta}^{(\chi)}\cdots\mathcal T_{\delta}^{(1)}.
\end{align}

Fix
\begin{align}
    p_\star\leq p<2,
    \qquad
    0<\theta\leq\theta_\star,
\end{align}
and set
\begin{equation}
    B_p:=\norm{O}_p,
    \qquad
    B_\theta:=\norm{O}_{2,\theta}.
    \label{eq:classical-initial-norms}
\end{equation}
Let $r\in\mathbb N$, $\delta:=t/r$, and
\begin{equation}
    L:=r\chi
    \label{eq:number-classical-layers}
\end{equation}
be the total number of color layers. For $\ell\in[L]$, define
\begin{equation}
    c_\ell
    :=
    1+((\ell-1)\bmod\chi),
    \qquad
    \Gamma_\ell
    :=
    \mathcal T_{\delta}^{(c_\ell)}.
    \label{eq:ordered-classical-color-layers}
\end{equation}
Thus
\begin{equation}
    \Gamma_L\cdots\Gamma_1
    =
    (\Phi_\delta^{\mathrm{Tr}})^r.
    \label{eq:classical-color-product}
\end{equation}

\subsection{Coefficient-truncation algorithm}
\label{supp:classical-algorithm}

Let $w\geq0$ and $\tau>0$. Starting from
\begin{align}
    \widehat O_0:=O,
\end{align}
the algorithm recursively computes
\begin{equation}
    X_\ell
    :=
    \Gamma_\ell(\widehat O_{\ell-1}),
    \qquad
    \widehat O_\ell
    :=
    \mathcal C_{w,\tau}(X_\ell),
    \qquad
    \ell=1,\ldots,L,
    \label{eq:classical-truncated-recursion}
\end{equation}
where $\mathcal C_{w,\tau}$ is the joint weight-and-coefficient truncation map from \cref{eq:joint-truncation-map}. Explicitly,
\begin{equation}
    \widehat O_\ell
    =
    \sum_{\substack{P\in\cP_n\\
                    |P|\leq w,\;
                    |(X_\ell)_P|\geq\tau}}
    (X_\ell)_P P.
    \label{eq:classical-truncated-update}
\end{equation}
All contributions to the same output Pauli operator are first added together, and coefficient truncation is applied only after this merging.

Since $\Gamma_\ell$ is contractive in both relevant norms and $\mathcal C_{w,\tau}$ only removes Pauli coefficients,
\begin{equation}
    \norm{\widehat O_\ell}_p
    \leq
    B_p,
    \qquad
    \norm{\widehat O_\ell}_{2,\theta}
    \leq
    B_\theta
    \label{eq:classical-iterate-regularity}
\end{equation}
for every $\ell\in\{0,\ldots,L\}$. The same bounds hold for the untruncated intermediate operator $X_\ell$.

\subsection{Local color updates}
\label{supp:classical-sparse-updates}

We store the retained nonzero coefficients in a dictionary
\begin{equation}
    P\longmapsto c_P.
    \label{eq:classical-sparse-dictionary}
\end{equation}
Each Pauli string is stored as a sorted list of its nonidentity factors and their qubit indices, as in \texttt{OpenFermion}~\cite{mcclean2017openfermion}. Related implementations for Pauli and Majorana propagation are provided in Algorithmiq's \texttt{monoprop} package~\cite{monoprop,miller2025simulation}.

For a label $a$, recall that
\begin{align}
    \mathcal T_{a,\delta}
    =
    \e^{\delta\mathcal D_a}
    \e^{\delta\mathcal H_a}.
\end{align}
If $P_a$ commutes with $P$, then
\begin{equation}
    \mathcal T_{a,\delta}(P)
    =
    \exp\!\left[-\delta\lambda_a(P)\right]P.
    \label{eq:classical-commuting-local-update}
\end{equation}
If $P_a$ anticommutes with $P$, let $Q$ be the phase-free Pauli operator proportional to $P_aP$, and let $\xi_{a,P}\in\{-1,+1\}$ be as in \cref{eq:local-pauli-pair-generator}. Then
\begin{align}
    \mathcal T_{a,\delta}(P)
    ={}&
    \cos(2J_a\delta)
    \exp\!\left[-\delta\lambda_a(P)\right]P
    \nonumber\\
    &+
    \xi_{a,P}\sin(2J_a\delta)
    \exp\!\left[-\delta\lambda_a(Q)\right]Q.
    \label{eq:classical-anticommuting-local-update}
\end{align}
Because $\mathcal D_a$ is supported on at most $k$ qubits, the rates $\lambda_a(P)$ and $\lambda_a(Q)$ can be evaluated in $O_k(1)$ time from the local description of $\mathcal D_a$. Thus one local map sends a Pauli operator to at most two Pauli operators.

For a fixed color $c$, define
\begin{equation}
    s_c(P)
    :=
    \left|
        \left\{
            a\in\mathcal C_c:
            \supp(P_a)\cap\supp(P)\neq\varnothing
        \right\}
    \right|.
    \label{eq:number-active-color-terms}
\end{equation}
Since the supports within one color are pairwise disjoint,
\begin{equation}
    s_c(P)
    \leq
    |P|.
    \label{eq:active-color-terms-weight-bound}
\end{equation}
Consequently, if $|P|\leq w$, then
\begin{equation}
    \mathcal T_{\delta}^{(c)}(P)
\end{equation}
contains at most
\begin{equation}
    2^{s_c(P)}
    \leq
    2^w
    \label{eq:number-color-branches}
\end{equation}
Pauli operators.

The active labels can be found without scanning the whole color class. After preprocessing the incidence lists
\begin{align}
    x\longmapsto
    \left\{
        a:x\in \supp(P_a)
    \right\},
\end{align}
one identifies the labels of color $c$ intersecting $P$ by scanning $\supp(P)$. For fixed locality, this requires $O_k(d|P|)$ elementary operations.

\subsection{One-layer coefficient-truncation error}
\label{supp:classical-one-layer-error}

\begin{lemma}[Error of one truncated color update]
\label{lem:classical-one-layer-error}
For every $\ell\in[L]$,
\begin{equation}
    \norm{
        X_\ell-\widehat O_\ell
    }_2
    \leq
    \e^{-\theta(w+1)}B_\theta
    +
    \tau^{1-p/2}B_p^{p/2}.
    \label{eq:classical-one-layer-error}
\end{equation}
\end{lemma}

\begin{proof}
By \cref{eq:classical-iterate-regularity},
\begin{align}
    \norm{X_\ell}_p
    \leq
    B_p,
    \qquad
    \norm{X_\ell}_{2,\theta}
    \leq
    B_\theta.
\end{align}
Applying the joint truncation estimate \cref{cor:joint-truncation} to $X_\ell$ gives the claim.
\end{proof}

\subsection{Accumulated truncation error}
\label{supp:classical-accumulated-error}

Let
\begin{equation}
    O^{\mathrm{PF}}_0:=O,
    \qquad
    O^{\mathrm{PF}}_\ell
    :=
    \Gamma_\ell(O^{\mathrm{PF}}_{\ell-1})
    \label{eq:untruncated-product-formula-iterates}
\end{equation}
denote the untruncated product-formula iterates. Then
\begin{align}
    O^{\mathrm{PF}}_L
    =
    (\Phi_\delta^{\mathrm{Tr}})^r(O).
\end{align}

\begin{proposition}[Accumulation of truncation errors]
\label{prop:classical-accumulated-error}
The output of \cref{eq:classical-truncated-recursion} satisfies
\begin{equation}
\begin{aligned}
    \norm{
        (\Phi_\delta^{\mathrm{Tr}})^r(O)
        -
        \widehat O_L
    }_2
    \leq
    L\biggl[
        &
        \e^{-\theta(w+1)}B_\theta
        \\
        &+
        \tau^{1-p/2}B_p^{p/2}
    \biggr].
\end{aligned}
    \label{eq:classical-accumulated-error}
\end{equation}
\end{proposition}

\begin{proof}
Using \cref{eq:classical-truncated-recursion}, we have
\begin{align}
    O^{\mathrm{PF}}_\ell-\widehat O_\ell
    ={}&
    \Gamma_\ell
    \left(
        O^{\mathrm{PF}}_{\ell-1}
        -
        \widehat O_{\ell-1}
    \right)
    \nonumber\\
    &+
    X_\ell-\mathcal C_{w,\tau}(X_\ell).
    \label{eq:classical-error-recursion}
\end{align}
The color map $\Gamma_\ell$ is contractive in the normalized Hilbert--Schmidt norm. Therefore
\begin{align}
    \norm{
        O^{\mathrm{PF}}_\ell-\widehat O_\ell
    }_2
    \leq{}&
    \norm{
        O^{\mathrm{PF}}_{\ell-1}
        -
        \widehat O_{\ell-1}
    }_2
    \nonumber\\
    &+
    \norm{
        X_\ell-\mathcal C_{w,\tau}(X_\ell)
    }_2.
    \label{eq:classical-error-recursion-bound}
\end{align}
Iterating this inequality from $\ell=1$ to $L$ and applying \cref{lem:classical-one-layer-error} proves the result.
\end{proof}

For a truncation-error budget $\varepsilon_{\mathrm{tr}}>0$, choose
\begin{equation}
    w
    :=
    \max\left\{
        0,\,
        \left\lceil
            \frac{1}{\theta}
            \log\left(
                \frac{2LB_\theta}{\varepsilon_{\mathrm{tr}}}
            \right)
        \right\rceil
    \right\}
    \label{eq:classical-weight-cutoff}
\end{equation}
and
\begin{equation}
    \tau
    :=
    \left(
        \frac{
            \varepsilon_{\mathrm{tr}}
        }{
            2L B_p^{p/2}
        }
    \right)^{\frac{2}{2-p}}.
    \label{eq:classical-coefficient-threshold}
\end{equation}
Then
\begin{equation}
    \norm{
        (\Phi_\delta^{\mathrm{Tr}})^r(O)
        -
        \widehat O_L
    }_2
    \leq
    \varepsilon_{\mathrm{tr}}.
    \label{eq:classical-truncation-budget}
\end{equation}

\subsection{Runtime and memory complexity}
\label{supp:classical-complexity}

After every truncation, the number of retained coefficients is bounded by
\begin{equation}
    N
    :=
    \tau^{-p}B_p^p.
    \label{eq:classical-number-retained}
\end{equation}
With the choice in \cref{eq:classical-coefficient-threshold},
\begin{equation}
    N
    \leq
    \left(
        \frac{2LB_p}{\varepsilon_{\mathrm{tr}}}
    \right)^{\frac{2p}{2-p}}.
    \label{eq:classical-number-retained-explicit}
\end{equation}

A retained Pauli operator has weight at most $w$ and produces at most $2^w$ branches under one color update. Hence one color layer can be implemented using
\begin{equation}
    \widetilde{\mathcal O}_k
    \left(
        Nw2^w
    \right)
    \label{eq:classical-one-layer-runtime}
\end{equation}
arithmetic and dictionary operations, where the tilde suppresses logarithmic factors associated with storing and merging Pauli keys. The temporary memory during one color update is
\begin{equation}
    \widetilde{\mathcal O}_k
    \left(
        N2^w
    \right),
    \label{eq:classical-temporary-memory}
\end{equation}
while after truncation only $N$ coefficients remain.

The choice in \cref{eq:classical-weight-cutoff} implies
\begin{equation}
    2^w
    \leq
    2
    \left(
        \frac{2LB_\theta}{\varepsilon_{\mathrm{tr}}}
    \right)^{\frac{\log 2}{\theta}}.
    \label{eq:classical-branching-explicit}
\end{equation}
Including the $O_k(m)$ preprocessing cost, the total runtime is therefore
\begin{equation}
\begin{aligned}
    \widetilde{\mathcal O}_k\biggl(
        m
        +
        Lw
        &
        \left(
            \frac{2LB_p}{\varepsilon_{\mathrm{tr}}}
        \right)^{\frac{2p}{2-p}}
        \\
        &\times
        \left(
            \frac{2LB_\theta}{\varepsilon_{\mathrm{tr}}}
        \right)^{\frac{\log 2}{\theta}}
    \biggr).
\end{aligned}
    \label{eq:classical-total-runtime}
\end{equation}
The corresponding peak memory is
\begin{equation}
    \widetilde{\mathcal O}_k\left(
        \left(
            \frac{2LB_p}{\varepsilon_{\mathrm{tr}}}
        \right)^{\frac{2p}{2-p}}
        \left(
            \frac{2LB_\theta}{\varepsilon_{\mathrm{tr}}}
        \right)^{\frac{\log 2}{\theta}}
    \right).
    \label{eq:classical-total-memory}
\end{equation}

The two regularity estimates play distinct roles in these bounds. Pauli-$p$ regularity controls the number $N$ of retained coefficients. Weighted Pauli-$2$ regularity ensures that the cutoff $w$ grows only logarithmically with the target precision. This converts the branching factor $2^w$ of a color update into a polynomial dependence on $\varepsilon_{\mathrm{tr}}^{-1}$. Coefficient sparsity alone would not control this branching factor, because a sparse Pauli expansion may still contain Pauli operators of extensive weight.

\subsection{First-order propagation guarantees}
\label{supp:proof-classical-simulation}

\begin{theorem}[Classical simulation by Pauli propagation]
\label{thm:supp-classical-simulation}
Let
\begin{align}
    p_\star\leq p<2,
    \qquad
    0<\theta\leq\theta_\star,
\end{align}
and let $\widehat O_L$ be the output of \cref{eq:classical-truncated-recursion}, with the parameters chosen in \cref{eq:classical-weight-cutoff,eq:classical-coefficient-threshold}. Then
\begin{equation}
\begin{aligned}
    \norm{
        \e^{t\mathcal L}(O)-\widehat O_L
    }_2
    \leq{}&
    \norm{
        \left(
            \e^{t\mathcal L}
            -
            (\Phi_{t/r}^{\mathrm{Tr}})^{\,r}
        \right)(O)
    }_2
    +
    \varepsilon_{\mathrm{tr}}.
\end{aligned}
    \label{eq:classical-total-error}
\end{equation}
Under the assumptions of \cref{cor:quantum-many-step-error}, this becomes
\begin{equation}
\begin{aligned}
    \norm{
        \e^{t\mathcal L}(O)-\widehat O_L
    }_2
    \leq{}&
    C_k
    \frac{t^2}{r}
    \frac{\e^{2k\theta}}{\theta}
    \left(
        J^2d^2+J\bar\gamma d^2
    \right)
    B_\theta
    \\
    &+
    \varepsilon_{\mathrm{tr}}.
\end{aligned}
    \label{eq:classical-total-error-explicit}
\end{equation}
The runtime and memory are bounded by \cref{eq:classical-total-runtime,eq:classical-total-memory}.
\end{theorem}

\begin{proof}
Insert the split product-formula evolution:
\begin{align}
    \norm{
        \e^{t\mathcal L}(O)-\widehat O_L
    }_2
    \leq{}&
    \norm{
        \left(
            \e^{t\mathcal L}
            -
            (\Phi_{t/r}^{\mathrm{Tr}})^{\,r}
        \right)(O)
    }_2
    \nonumber\\
    &+
    \norm{
        (\Phi_{t/r}^{\mathrm{Tr}})^{\,r}(O)
        -
        \widehat O_L
    }_2.
\end{align}
The second term is at most $\varepsilon_{\mathrm{tr}}$ by \cref{eq:classical-truncation-budget}. The explicit product-formula estimate follows from \cref{cor:quantum-many-step-error}. The complexity bounds were established in \cref{eq:classical-total-runtime,eq:classical-total-memory}.
\end{proof}

Choosing
\begin{align}
    \varepsilon_{\mathrm{tr}}
    =
    \frac{\varepsilon}{2}
\end{align}
and taking $r$ large enough that the product-formula contribution in \cref{eq:classical-total-error-explicit} is at most $\varepsilon/2$ gives a normalized Hilbert--Schmidt approximation error at most $\varepsilon$.

For fixed $k$, $d$, $\gamma/J$, and initial observable weight, the parameters $p<2$ and $\theta>0$ are constants. Hence the runtime and memory are polynomial in the evolution time and inverse target precision, provided that the required product-formula step count $r$ is polynomial in these quantities.

\subsection{Propagation guarantees with Richardson extrapolation}
\label{supp:classical-richardson}

We now extend the first-order analysis above to the symmetric circuits used in the main text. We give the extension explicitly, including the initial observable and the signed extrapolation error. As in the main theorem, $O$ is supported on at most $w_0$ qubits and $\norm{O}_2=1$. Its Pauli expansion has at most $M_0\leq4^{w_0}$ nonzero coefficients and can be formed from its local matrix representation at a cost depending only on $w_0$.

\begin{lemma}[Contractivity and sparsity of a symmetric colour update]
\label{lem:classical-symmetric-colour}
For $s\geq0$, define
\begin{equation}
 Q_{a,s}:=\e^{s\mathcal D_a/2}\e^{s\mathcal H_a}
                \e^{s\mathcal D_a/2},
 \qquad Q_{c,s}=\prod_{a\in\mathcal C_c}Q_{a,s}.
\end{equation}
The map $Q_{c,s}$ is contractive in Pauli-$p$ for $p_\star\leq p\leq2$, and in weighted Pauli-$2$ for $0\leq\theta\leq\theta_\star$. A weight-$w$ Pauli string has at most $2^w$ output branches, each of weight at most $kw$.
\end{lemma}
\begin{proof}
Contractivity follows from \cref{cor:symmetric-colour-regularity}. The commuting local update is $Q_{a,s}(P)=\e^{-s\lambda_a(P)}P$. For an anticommuting partner $Q$ and the sign $\xi_{a,P}$ from \cref{eq:local-pauli-pair-generator},
\begin{align}
 Q_{a,s}(P)={}&\cos(2J_as)\e^{-s\lambda_a(P)}P\nonumber\\
 &+\xi_{a,P}\sin(2J_as)
       \e^{-s[\lambda_a(P)+\lambda_a(Q)]/2}Q.
 \label{eq:classical-symmetric-pauli-update}
\end{align}
Each active block therefore produces at most two branches. There are at most $|P|$ active blocks because the blocks of one colour are disjoint. A block initially disjoint from $P$ stays disjoint from all branches. Their supports are contained in $\supp(P)$ together with its active blocks, a set of size at most $k|P|$.
\end{proof}

Consider any circuit consisting of $L$ symmetric colour maps with nonnegative durations. Use these maps as $\Gamma_\ell$ in \cref{eq:classical-truncated-recursion}, and put
\begin{equation}
 W:=\max\left\{w_0,
       \left\lceil\theta^{-1}\log(2LB_\theta/\eta)\right\rceil
       \right\},\qquad
 \tau:=\left(\frac{\eta}{2LB_p^{p/2}}\right)^{2/(2-p)}.
 \label{eq:classical-symmetric-cutoffs}
\end{equation}
Use $w=W$ and the displayed $\tau$ in \cref{eq:classical-truncated-recursion}. Here $p_\star\leq p<2$, $0<\theta\leq\theta_\star$, and $0<\eta\leq1$. By \cref{lem:classical-symmetric-colour}, the norm bounds \cref{eq:classical-iterate-regularity} hold for these iterates. The proof of \cref{prop:classical-accumulated-error} therefore applies without change: it uses only those bounds and Pauli-$2$ contractivity of each complete colour map. It gives
\begin{equation}
 \norm{\Gamma_L\cdots\Gamma_1(O)-\widehat O_L}_2\leq\eta.
 \label{eq:classical-symmetric-truncation-error}
\end{equation}
Truncation is performed after a complete symmetric colour map, rather than between its coherent and dissipative factors. The bound is uniform in the durations of the colour maps. Its dependence on the number $L$ of truncation points is included in the cutoffs.

After truncation, at most $M=(2LB_p/\eta)^{2p/(2-p)}$ coefficients remain. To account for the initial update, use $M_{\mathrm{eff}}:=\max\{4^{w_0},M\}$; every input string then has weight at most $W$. Construct once a sparse table $(x,c)\mapsto a$ whenever the unique block of colour $c$ containing qubit $x$ has label $a$. It identifies active blocks by scanning the string support and avoids a factor $d$ from repeated scans of all incident labels. A greedy colouring and this table cost $\widetilde{\mathcal O}_k(n+md)$ and are reused for all circuits. After this preprocessing, the propagation and final expectation evaluation require
\begin{equation}
 \widetilde{\mathcal O}_k\!\left(
      L M_{\mathrm{eff}}(W+1)2^W\right)
 \label{eq:classical-symmetric-propagation-cost}
\end{equation}
operations. This includes storing the at-most-$kW$ output supports and merging identical strings. The expectation-evaluation cost $O(M_{\mathrm{eff}}(W+1))$ follows from the assumed access to $\Tr(\rho P)$.

Since $B_p\leq2^{w_0(2/p-1)}$ and $B_\theta\leq\e^{\theta w_0}$,
\begin{equation}
 2^W\leq2^{w_0}+2(2LB_\theta/\eta)^{\log2/\theta}.
\end{equation}
For fixed physical parameters and $w_0$, the propagation cost is thus
\begin{equation}
 \widetilde{\mathcal O}\!\left(L^{1+\beta}\eta^{-\beta}\right),
 \qquad \beta:=\frac{2p}{2-p}+\frac{\log2}{\theta}.
 \label{eq:classical-symmetric-precision-cost}
\end{equation}
The tilde here also suppresses the logarithmic factor $W+1$ and the cost of storing qubit indices.

\subsubsection{Proof of the main classical theorem}
\label{supp:proof-main-classical}

\begin{proof}[Proof of \cref{thm:main-classical-simulation}]
Apply \cref{thm:rich-lind-maximum-depth} with target observable error $\epsilon/2$. Let $q,N_j,\omega_j$ be its choices, and define
\begin{equation}
 A_q:=\sum_{j=1}^q|\omega_j|,\qquad
 \eta:=\frac{\epsilon}{2A_q},\qquad
 L_j:=(2\chi-1)N_j.
\end{equation}
The coefficient identities imply $1\leq A_q\leq48q^3$. For each circuit, take $p=p_\star$, $\theta=\theta_\star$, and apply the symmetric propagation just described with $L=L_j$. Its output $\widehat O^{(j)}$ satisfies $\|\mathsf{S}_{t/N_j}^{N_j}(O)-\widehat O^{(j)}\|_2\leq\eta$. Output
\begin{equation}
 \widehat f:=\sum_{j=1}^q\omega_j\Tr(\rho\widehat O^{(j)}).
\end{equation}
The triangle inequality and the Richardson guarantee give
\begin{align}
 \left\|\e^{t\mathcal L}(O)
          -\sum_{j=1}^q\omega_j\widehat O^{(j)}\right\|_2
 &\leq\left\|\e^{t\mathcal L}(O)
          -\sum_{j=1}^q\omega_j \mathsf{S}_{t/N_j}^{N_j}(O)\right\|_2
       +\sum_{j=1}^q|\omega_j|
          \|\mathsf{S}_{t/N_j}^{N_j}(O)-\widehat O^{(j)}\|_2\\
 &\leq\epsilon/2+A_q\eta=\epsilon.
 \label{eq:classical-richardson-total-error}
\end{align}
For Hermitian $X$, positivity of $\rho$ and Cauchy--Schwarz imply $|\Tr(\rho X)|^2\leq\Tr(\rho X^2)$. Averaging and using $\mathbb E\rho=I/2^n$ yields $\mathbb E|\Tr(\rho X)|^2\leq\|X\|_2^2$. Applying this inequality to the error observable proves the stated RMS estimate.

Fix $k,d,J,\gamma,\bar\gamma,t,w_0$. The Richardson construction has $q=O(\log(\e/\epsilon))$ and $L_j\leq Cq^2$, where $C$ depends only on these parameters. At the chosen $p$ and $\theta$,
\begin{equation}
 \beta=\frac{2p_\star}{2-p_\star}
             +\frac{\log2}{\theta_\star}
 =\max\!\left\{2,\frac{4J}{\gamma}\right\}
        +\frac{k\log2}{\operatorname{arsinh}(\gamma/(2J))}
 =\beta_\star.
\end{equation}
After empty-support blocks are removed, $m\leq nd$, so the shared preprocessing is $\widetilde{\mathcal O}(n)$ for fixed $d$. Summing \cref{eq:classical-symmetric-precision-cost} over all $q$ circuits gives
\begin{align}
 T_{\mathrm C}
 &\in\widetilde{\mathcal O}\!\left(
       n+\sum_{j=1}^q L_j^{1+\beta_\star}
                         (2A_q/\epsilon)^{\beta_\star}\right)\\
 &\subseteq\widetilde{\mathcal O}\!\left(
       n+\epsilon^{-\beta_\star}q^{3+5\beta_\star}\right).
\end{align}
This proves \cref{eq:main-classical-runtime}. For $t=0$, directly evaluating the input observable suffices.
\end{proof}

\suppsection{Improved classical simulation of noisy circuits}
\label{supp:gate-based-noise}

We consider an $n$-qubit circuit $\mathcal C$ containing $\operatorname{poly}(n)$ gates
\[
 U_j=e^{-i\delta H_j},
 \qquad H_j=H_j^\dagger,
 \qquad \|H_j\|_{\mathrm{op}}\leq1,
 \qquad \delta\geq0,
\]
where we denoted the operator norm as $\|\cdot\|_{\mathrm{op}}$. Each gate acts on at most $k=O(1)$ qubits and is followed by independent depolarizing noise on its support, with
\[
 \mathcal N_\gamma(I)=I,
 \qquad
 \mathcal N_\gamma(P)=e^{-\gamma}P
 \quad(P\in\{X,Y,Z\}),
 \qquad \gamma>0.
\]

\begin{lemma}[Nonincrease of Pauli coefficient norms]
\label{lem:gate-pauli-p-contractivity}
Define
\[
 a:=\max\left\{2,\frac{2^{k+2}\delta}{\gamma}\right\},
 \qquad
 p_\star:=\frac{2a}{a+2}.
\]
Every noisy gate $\Gamma_j \coloneqq e^{-iH_j\delta} \mathcal{N}_\gamma^{\otimes k}(\cdot) e^{iH_j\delta} $ satisfies
\[
 \|\Gamma_j(X)\|_p\leq\|X\|_p,
 \qquad p_\star\leq p\leq2.
\]
Consequently, the same bound holds for any composition of these noisy gates.
\end{lemma}

\begin{proof}
Expand $H_j=\sum_P h_{j,P}P$ on its gate support. Pauli orthogonality and Cauchy--Schwarz give
\[
 \sum_{P\neq I}|h_{j,P}|
 \leq\sqrt{4^k-1}\,\|H_j\|_2
 \leq2^k\|H_j\|_{\mathrm{op}}
 \leq2^k.
\]
Since a commutator of two Pauli operators is either zero or a Pauli operator with coefficient of magnitude two,
\[
 \|i[H_j,\cdot]\|_{1\to1}
 \leq2\sum_{P\neq I}|h_{j,P}|
 \leq2^{k+1}.
\]
Thus, for $\mathcal V_j(X):=U_j^\dagger XU_j$,
\[
 \|\mathcal V_j\|_{1\to1}\leq e^{2^{k+1}\delta}.
\]
Unitary conjugation is a Pauli-$2$ isometry. Riesz--Thorin interpolation therefore yields
\[
 \|\mathcal V_j\|_{p\to p}
 \leq
 \exp\left[
 2^{k+1}\delta\left(\frac{2}{p}-1\right)
 \right],
 \qquad 1\leq p\leq2.
\]

Let $\mathcal N_j$ denote the product depolarizing channel on the gate support. Since noise follows the physical gate,
\[
 \Gamma_j(X)=U_j^\dagger\mathcal N_j(X)U_j.
\]
For each fixed Pauli label outside the gate support, the local identity and traceless sectors are invariant. The identity sector is fixed, whereas $\mathcal N_j$ contracts the Pauli-$p$ norm on the traceless sector by at least $e^{-\gamma}$. Hence, for $p\geq p_\star$,
\[
 \|\Gamma_j|_{\mathrm{traceless}}\|_{p\to p}
 \leq
 \exp\left[
 -\gamma+2^{k+1}\delta\left(\frac{2}{p}-1\right)
 \right]
 \leq
 \exp\left[-\gamma+\frac{2^{k+2}\delta}{a}\right]
 \leq1.
\]
\end{proof}

\begin{theorem}[Hilbert--Schmidt simulation guarantee]
\label{thm:appendix-gate-noise}
Let $O=O^\dagger$, $\|O\|_2\leq1$, be supplied as an explicit expansion of $\operatorname{poly}(n)$ Pauli operators. For every $0<\varepsilon\leq1$, a deterministic classical algorithm constructs a sparse Pauli expansion $\widetilde O$ satisfying
\[
 \|\mathcal C(O)-\widetilde O\|_2\leq\varepsilon
\]
in time
\[
 \left(\frac{n}{\varepsilon}\right)^{
 \mathcal O(1+\delta/\gamma)}.
\]
The implicit constants depend on the fixed gate size and the fixed polynomial bounds on the input size.
\end{theorem}

\begin{proof}
Take $p=p_\star$ and $a$ from Lemma~\ref{lem:gate-pauli-p-contractivity}, so that
\[
 \frac1p-\frac12=\frac1a,
 \qquad
 \frac{2p}{2-p}=a.
\]
Let $q=q(n)\geq2$ be a polynomial upper bound on both the number of gates and the number of input Pauli terms. Then
\[
 \|O\|_p
 \leq q^{1/p-1/2}\|O\|_2
 \leq q^{1/a}.
\]

Enumerate the complete noisy gates in Heisenberg propagation order. Starting from $\widehat O_0=O$, apply each update, merge equal Pauli strings, and discard coefficients below
\[
 \tau:=q^{-1/2}
       \left(\frac{\varepsilon}{q}\right)^{(a+2)/2}.
\]
Explicitly,
\[
 X_j:=\Gamma_j(\widehat O_{j-1}),
 \qquad
 \widehat O_j:=
 \sum_{P:|(X_j)_P|\geq\tau}(X_j)_P P.
\]
No Pauli-weight cutoff is imposed.

By Lemma~\ref{lem:gate-pauli-p-contractivity}, and because coefficient deletion cannot increase a Pauli-$p$ norm,
\[
 \|X_j\|_p,\ \|\widehat O_j\|_p\leq q^{1/a}.
\]
The squared error of one truncation consequently satisfies
\[
 \|X_j-\widehat O_j\|_2^2
 \leq\tau^{2-p}\|X_j\|_p^p
 \leq\tau^{2-p}q^{p/a}
 =\left(\frac{\varepsilon}{q}\right)^2.
\]

Let $O_0=O$ and $O_j=\Gamma_j(O_{j-1})$ denote the untruncated iterates. Pauli-$2$ contractivity gives
\[
 \|O_j-\widehat O_j\|_2
 \leq
 \|O_{j-1}-\widehat O_{j-1}\|_2
 +\|X_j-\widehat O_j\|_2
 \leq
 \|O_{j-1}-\widehat O_{j-1}\|_2+\frac{\varepsilon}{q}.
\]
There are at most $q$ updates, so the final truncated observable satisfies
\[
 \|\mathcal C(O)-\widetilde O\|_2\leq\varepsilon.
\]

It remains to bound the cost. Every retained coefficient has magnitude at least $\tau$, so each truncated iterate contains at most
\[
 \tau^{-p}q^{p/a}
 =q\left(\frac{q}{\varepsilon}\right)^a
\]
Pauli terms. This bound also covers the initial expansion. A gate on at most $k$ qubits creates at most $4^k$ candidates per incoming term, and diagonal noise introduces no additional branching.

Using length-$n$ Pauli keys, at most $q$ updates therefore require
\[
 \widetilde{\mathcal O}_k\!\left[
 n q^2\left(\frac{q}{\varepsilon}\right)^a
 \right]
\]
arithmetic operations, including initialization, construction of the local transfer matrices, coefficient merging, and truncation. The tilde suppresses dictionary logarithms. Since $q=\operatorname{poly}(n)$ and $a=\mathcal O_k(1+\delta/\gamma)$, this is bounded by
\[
 \left(\frac{n}{\varepsilon}\right)^{
 \mathcal O(1+\delta/\gamma)},
\]
as claimed.
\end{proof}

While we stated our result for gate-based noise, we remark that uniform local depolarizing noise after every circuit layer is included as a special case, as it can be recovered by inserting noisy identity gates on idle qubits.

\suppsection{Operational guarantees for normalized Hilbert--Schmidt error}
\label{supp:operational}

For an observable $O$ and an approximation $\widetilde O$, write
\begin{align}
    \Delta:=O-\widetilde O \coloneqq \sum_{P\in \cP_n} \Delta_P P.
\end{align}
Recall that
\begin{align}
    \norm{\Delta}_2^2
    =
    \frac{1}{2^n}\Tr(\Delta^2)
    =
    \sum_{P\in\cP_n}\Delta_P^2,
\end{align}
where we assumed $\Delta$ to be Hermitian.

A normalized Hilbert--Schmidt bound does not imply a uniform bound on $\Tr(\rho\Delta)$ over all states $\rho$. It does, however, directly control infinite-temperature correlation functions, randomized input states, and expectation values averaged over Pauli-invariant random-Hamiltonian ensembles. We will detail these guarantees in the remainder of this section.

\subsection{Correlation functions, out-of-time-order correlators and the DQC1 model of computation}
\label{supp:operational-correlations}

We first derive deterministic error bounds for correlation functions, out-of-time-order correlators, and DQC1 output expectations. These bounds require no averaging over input states or dynamics.

\begin{lemma}[Normalized correlation functions]
\label{lem:operational-correlation}
For every pair of operators $A$ and $\Delta$,
\begin{equation}
    \left|
        \frac{1}{2^n}\Tr(A\Delta)
    \right|
    \leq
    \norm{A}_2\norm{\Delta}_2.
    \label{eq:operational-correlation}
\end{equation}
Consequently, if
\begin{align}
    \norm{O-\widetilde O}_2\leq\varepsilon,
\end{align}
then
\begin{equation}
    \left|
        \frac{1}{2^n}\Tr(A O)
        -
        \frac{1}{2^n}\Tr(A\widetilde O)
    \right|
    \leq
    \varepsilon\norm{A}_2.
    \label{eq:operational-correlation-approximation}
\end{equation}
\end{lemma}

\begin{proof}
This is the Cauchy--Schwarz inequality for the normalized Hilbert--Schmidt inner product.
\end{proof}

In particular, if $O=O(t)$ is a Heisenberg-evolved observable, \cref{eq:operational-correlation-approximation} controls the infinite-temperature dynamical correlation function
\begin{align}
    \frac{1}{2^n}\Tr\!\left(A^\dagger O(t)\right).
\end{align}

For a unitary operator $V$, define
\begin{equation}
    F_V(X)
    :=
    \frac{1}{2^n}
    \Tr\!\left(
        X^\dagger V^\dagger X V
    \right).
    \label{eq:operational-OTOC-definition}
\end{equation}
When $X$ and $V$ are Hermitian unitaries, this is the standard infinite-temperature out-of-time-order correlator.

\begin{lemma}[Out-of-time-order correlators, Eq.\ 63 in Ref.\ \cite{dowling2026classical}]
\label{lem:operational-OTOC}
Let $V$ be unitary. If
\begin{align}
    \norm{O-\widetilde O}_2\leq\varepsilon,
\end{align}
then
\begin{equation}
    \left|
        F_V(O)-F_V(\widetilde O)
    \right|
    \leq
    \varepsilon
    \left(
        \norm{O}_2+\norm{\widetilde O}_2
    \right).
    \label{eq:operational-OTOC-error}
\end{equation}
In particular,
\begin{equation}
    \left|
        F_V(O)-F_V(\widetilde O)
    \right|
    \leq
    2\varepsilon\norm{O}_2+\varepsilon^2.
    \label{eq:operational-OTOC-error-simple}
\end{equation}
\end{lemma}

\begin{proof}
Writing $\Delta=O-\widetilde O$, we have
\begin{align}
    F_V(O)-F_V(\widetilde O)
    ={}&
    \frac{1}{2^n}
    \Tr\!\left(
        \Delta^\dagger V^\dagger O V
    \right)
    \nonumber\\
    &+
    \frac{1}{2^n}
    \Tr\!\left(
        \widetilde O^\dagger V^\dagger\Delta V
    \right).
\end{align}
Cauchy--Schwarz and unitary invariance of the normalized Hilbert--Schmidt norm give
\begin{align}
    \left|
        F_V(O)-F_V(\widetilde O)
    \right|
    \leq
    \norm{\Delta}_2
    \left(
        \norm{O}_2+\norm{\widetilde O}_2
    \right).
\end{align}
The second claim follows from
\begin{align}
    \norm{\widetilde O}_2
    \leq
    \norm{O}_2+\varepsilon.
\end{align}
\end{proof}

The same correlation bound controls output expectations in the DQC1 (deterministic quantum computation with one clean qubit) model~\cite{knill1998power}. On $n$ qubits, the input consists of one pure qubit and $n-1$ maximally mixed qubits,
\begin{equation}
    \rho_{\mathrm{in}}
    = |0\rangle\langle0|\otimes\frac{I}{2^{n-1}}
    = \frac{I+Z_1}{2^n},
\end{equation}
where $Z_1=Z\otimes I^{\otimes(n-1)}$ is a full-system Pauli operator. For a circuit unitary $W$, measuring $Z_1$ at the output gives the expectation
\begin{equation}
    m=\Tr(\rho_{\mathrm{in}}O)
    =\frac{1}{2^n}\Tr(Z_1O),
    \qquad O=W^\dagger Z_1W,
    \label{eq:operational-DQC1-output}
\end{equation}
where the second equality uses $\Tr(O)=0$. Given $\norm{O-\widetilde O}_2\leq\varepsilon$, define $\widetilde m:=2^{-n}\Tr(Z_1\widetilde O)$. Since $\norm{Z_1}_2=1$, \cref{eq:operational-correlation-approximation} gives
\begin{equation}
    |m-\widetilde m|\leq\varepsilon.
    \label{eq:operational-DQC1-error}
\end{equation}

\subsection{Random input states}
\label{supp:operational-random-inputs}

We follow the low-average-ensemble framework introduced in \cite{schuster2024polynomial} and subsequently used in \cite{angrisani2025simulating}. This framework requires only that the average input state be sufficiently close to the maximally mixed state in operator norm; no second-moment or state-design assumption is needed.

\begin{definition}[Low-average ensemble]
\label{def:low-average-ensemble}
Let $\mathcal E$ be an ensemble of $n$-qubit density operators, and define its average state by
\begin{equation}
    \overline\rho_{\mathcal E}
    :=
    \mathbb E_{\rho\sim\mathcal E}[\rho].
    \label{eq:average-input-state}
\end{equation}
We say that $\mathcal E$ is a low-average ensemble with parameter $c\geq1$ if
\begin{equation}
    \norm{\overline\rho_{\mathcal E}}_\infty
    \leq
    \frac{c}{2^n},
    \qquad
    D=2^n.
    \label{eq:low-average-condition}
\end{equation}
\end{definition}

The case $c=1$ is equivalent to
\begin{align}
    \overline\rho_{\mathcal E}
    =
    \frac{I}{2^n}.
\end{align}
It includes, for example, a uniformly sampled state from any orthonormal basis and, more generally, any state $1$-design.

\begin{lemma}[Expectation values on low-average input ensembles, Lemma B1 in Ref.\ \cite{schuster2024polynomial}]
\label{lem:low-average-expectation-error}
Let $\mathcal E$ be a low-average ensemble with parameter $c$, and let $O$ and $\widetilde O$ be Hermitian operators. Then
\begin{equation}
    \mathbb E_{\rho\sim\mathcal E}
    \left[
        \left|
            \Tr\!\left[
                \rho\left(O-\widetilde O\right)
            \right]
        \right|^2
    \right]
    \leq
    c\,\norm{O-\widetilde O}_2^2.
    \label{eq:low-average-expectation-error}
\end{equation}
Equivalently,
\begin{equation}
    \left(
        \mathbb E_{\rho\sim\mathcal E}
        \left[
            \left|
                \Tr\!\left[
                    \rho\left(O-\widetilde O\right)
                \right]
            \right|^2
        \right]
    \right)^{1/2}
    \leq
    \sqrt{c}\,
    \norm{O-\widetilde O}_2.
    \label{eq:low-average-rms-error}
\end{equation}
\end{lemma}

Consequently, any normalized Hilbert--Schmidt approximation
\begin{equation}
    \norm{O-\widetilde O}_2
    \leq
    \varepsilon
    \label{eq:random-input-HS-assumption}
\end{equation}
implies the root-mean-square operational guarantee
\begin{equation}
    \left(
        \mathbb E_{\rho\sim\mathcal E}
        \left[
            \left|
                \Tr(\rho O)
                -
                \Tr(\rho\widetilde O)
            \right|^2
        \right]
    \right)^{1/2}
    \leq
    \sqrt{c}\,\varepsilon.
    \label{eq:random-input-operational-guarantee}
\end{equation}
In particular, for an ensemble whose average state is maximally mixed, $c=1$, and the root-mean-square expectation-value error is at most $\varepsilon$.

By Markov's inequality, for every $0<\delta<1$,
\begin{equation}
    \Pr_{\rho\sim\mathcal E}
    \left[
        \left|
            \Tr(\rho O)
            -
            \Tr(\rho\widetilde O)
        \right|
        \geq
        \sqrt{\frac{c}{\delta}}\,\varepsilon
    \right]
    \leq
    \delta.
    \label{eq:random-input-high-probability-bound}
\end{equation}
Thus, with probability at least $1-\delta$ over the input state, the expectation-value error is at most $\sqrt{c/\delta}\,\varepsilon$.

\subsection{Pauli-invariant random-Hamiltonian ensembles}
\label{supp:pauli-invariant-Hamiltonians}

We now explain how a normalized Hilbert--Schmidt simulation guarantee can be converted into an expectation-value guarantee for an arbitrary \emph{fixed} input state when the Hamiltonian is random. We will use the following two ingredients:

\begin{enumerate}
    \item The exact evolution and the complete approximation procedure are    covariant under simultaneous conjugation by Pauli unitaries. This is a    deterministic statement, valid for every Hamiltonian realization    separately.

    \item The joint probability law of the local Hamiltonian terms is    invariant under simultaneous conjugation by Pauli unitaries.
\end{enumerate}

It is useful to keep these two statements separate. Covariance alone relates the error for one Hamiltonian realization to the error for a conjugated realization. Randomness is used only later, when the two realizations are identified in distribution.

For $R,P\in\cP_n$, define the Pauli conjugation character by
\begin{equation}
    RPR=\chi_R(P)P,
    \qquad
    \chi_R(P)\in\{-1,+1\}.
    \label{eq:Pauli-conjugation-character}
\end{equation}
The characters separate distinct Pauli strings: if $P\neq Q$, then there is a Pauli $R$ such that
\begin{equation}
    \chi_R(P)\chi_R(Q)=-1.
    \label{eq:Pauli-characters-separate}
\end{equation}
Indeed, choose a site on which $P$ and $Q$ differ. A single-site Pauli can be chosen to commute with one of the two local factors and anticommute with the other.

We will also need the following definitions.

\begin{definition}[Equality in distribution]
Let $A\sim\mathcal D_A$ and $B\sim\mathcal D_B$ be random objects taking values in the same measurable space. We write $A\stackrel{\mathrm d}{=}B$ if
\begin{equation}
    \mathbb E_{A\sim\mathcal D_A}\!\left[F(A)\right]
    =
    \mathbb E_{B\sim\mathcal D_B}\!\left[F(B)\right]
\end{equation}
for every bounded measurable function $F$. Equivalently, $\mathcal D_A=\mathcal D_B$.
\end{definition}

\begin{definition}[Pauli-invariant Hamiltonian ensemble]
\label{def:Pauli-invariant-Hamiltonian}
Let
\begin{equation}
    H=\sum_{a=1}^mH_a
    \label{eq:sampled-Hamiltonian-data}
\end{equation}
be a random Hamiltonian with a specified decomposition into local terms. The ensemble $(H_1, H_2,\dots, H_m)$ is Pauli invariant if for all $R\in \cP_n$ 
\begin{equation}
    (H_1, H_2,\dots, H_m)\stackrel{\mathrm d}{=} (RH_1R, RH_2R,\dots, RH_mR).
    \label{eq:Pauli-invariant-Hamiltonian}
\end{equation}
\end{definition}

A common sufficient condition is easy to check. Suppose
\begin{equation}
    H=\sum_{a=1}^m J_aP_a,
    \qquad
    P_a\in\cP_n.
    \label{eq:random-Pauli-Hamiltonian}
\end{equation}
Conjugation by $R$ sends the coupling vector to
\begin{equation}
    (J_1,\ldots,J_m)
    \longmapsto
    \bigl(\chi_R(P_1)J_1,\ldots,\chi_R(P_m)J_m\bigr).
    \label{eq:random-coupling-sign-action}
\end{equation}
Thus if the coupling $J_a$ are sampled independently from a sign-symmetric distribution, such as the uniform distribution over $[-J,J]$, then the ensemble is Pauli invariant.

\subsection{Deterministic covariance of the exact dynamics}
\label{supp:operational-exact-covariance}

Fix one realization $H=\sum_{a=1}^mH_a$ and retain the same decomposition into local terms as in \cref{eq:full-lindbladian}. Define
\begin{equation}
    \mathcal H_{a,H}(A):=i[H_a,A],
    \qquad
    \mathcal L_{a,H}:=\mathcal H_{a,H}+\mathcal D_a,
    \qquad
    \mathcal L_H:=\sum_{a=1}^m\mathcal L_{a,H},
    \label{eq:random-Hamiltonian-block-generator}
\end{equation}
and let
\begin{equation}
    \mathcal A_{H,s}:=\e^{s\mathcal L_H}.
    \label{eq:random-Hamiltonian-exact-map}
\end{equation}
We assume that each $\mathcal D_a$ is Pauli covariant:
\begin{equation}
    \mathcal D_a(RAR)=R\mathcal D_a(A)R,
    \qquad
    a\in [m],
    \label{eq:Pauli-covariant-block-dissipator}
\end{equation}
for every $R\in\cP_n$. Each Pauli-diagonal dissipator $\mathcal D_a$ used throughout this work has this property. 

\begin{lemma}[Covariance of the exact dynamics]
\label{lem:exact-Pauli-covariance}
For every realization $H$, Pauli $R$, time $s$, and operator $A$,
\begin{equation}
    \mathcal A_{RHR,s}(RAR)=R\mathcal A_{H,s}(A)R.
    \label{eq:exact-Pauli-covariance}
\end{equation}
\end{lemma}

\begin{proof}
The local generators satisfy the pointwise intertwining relation
\begin{align}
    \mathcal L_{RHR}(RAR)
    &=
    \sum_{a=1}^m
    \left(i[RH_aR,RAR]+\mathcal D_a(RAR)\right)
    \nonumber\\
    &=
    R\left[
        \sum_{a\in\mathcal A}
        \left(i[H_a,A]+\mathcal D_a(A)\right)
    \right]R
    \nonumber\\
    &=R\mathcal L_H(A)R.
    \label{eq:exact-generator-intertwining}
\end{align}
Applying this identity repeatedly gives
\begin{equation}
    \mathcal L_{RHR}^{k}(RAR)=R\mathcal L_H^{k}(A)R
\end{equation}
for every $k\geq0$. Substituting into the exponential series proves \cref{eq:exact-Pauli-covariance}.
\end{proof}

\subsection{Deterministic covariance of the Trotterized and truncated dynamics}
\label{supp:operational-approximate-covariance}

We next verify the same covariance for the complete approximation used by the classical simulator. The elementary objects are the split local maps associated with the local terms $\mathcal L_a$, and truncation is applied immediately after every such map.

Fix the deterministic ordered list
\begin{equation}
    \alpha_1,\ldots,\alpha_M
    \label{eq:random-H-ordered-local-terms}
\end{equation}
of local labels encountered during the entire Trotter circuit. The list includes all product-formula steps, so the same label may occur many times. The index $j$ counts local operations in the circuit, not physical time intervals. Let $\delta_j$ be the step parameter used at operation $j$, and define the dissipative map associated with that label by
\begin{equation}
    \mathcal N_j:=\e^{\delta_j\mathcal D_{\alpha_j}}.
    \label{eq:random-H-local-dissipative-map}
\end{equation}
Define
\begin{align}
    \mathcal U_{j,H}(A)
    &:=
    \e^{i\delta_j H_{\alpha_j}}
    A
    \e^{-i\delta_j H_{\alpha_j}},
    \label{eq:random-H-local-unitary-map}
    \\
    \mathcal G_{j,H}
    &:=
    \mathcal N_j\circ\mathcal U_{j,H}.
    \label{eq:random-H-noisy-local-rotation}
\end{align}
Thus $\mathcal G_{j,H}$ is precisely the split map of the local Lindbladian term $\mathcal L_{\alpha_j}$: it first applies the coherent rotation and then the associated dissipative evolution.

Both parts of this elementary map are covariant under simultaneous Pauli conjugation. For every $R\in\cP_n$,
\begin{equation}
    \mathcal U_{j,RHR}(RAR)=R\mathcal U_{j,H}(A)R,
    \label{eq:random-H-local-unitary-covariance}
\end{equation}
because $\e^{i\delta_j RH_{\alpha_j}R} =R\e^{i\delta_jH_{\alpha_j}}R$. The Pauli covariance of the dissipator in \cref{eq:Pauli-covariant-block-dissipator} implies
\begin{equation}
    \mathcal N_j(RAR)=R\mathcal N_j(A)R.
    \label{eq:random-H-local-noise-covariance}
\end{equation}
Consequently,
\begin{equation}
    \mathcal G_{j,RHR}(RAR)=R\mathcal G_{j,H}(A)R.
    \label{eq:random-H-noisy-local-covariance}
\end{equation}
This identity is deterministic and holds for each local operation separately. It does not require the Hamiltonians used at different operations to be independent.

The truncation maps obey the same covariance. If
\begin{equation}
    A=\sum_{P\in\cP_n}a_PP,
\end{equation}
then
\begin{equation}
    RAR=\sum_{P\in\cP_n}\chi_R(P)a_PP.
    \label{eq:Pauli-conjugated-expansion}
\end{equation}
Pauli conjugation preserves both the weight of each Pauli string and the magnitude of its coefficient. Therefore
\begin{equation}
    \Pi_{\leq w}(RAR)=R\Pi_{\leq w}(A)R,
    \label{eq:weight-truncation-equivariance}
\end{equation}
\begin{equation}
    \mathcal T_\tau(RAR)=R\mathcal T_\tau(A)R,
    \label{eq:coefficient-truncation-equivariance}
\end{equation}
and
\begin{equation}
    \mathcal C_{w,\tau}(RAR)=R\mathcal C_{w,\tau}(A)R.
    \label{eq:joint-truncation-equivariance}
\end{equation}
Coefficient truncation is nonlinear, but it is odd:
\begin{equation}
    \mathcal T_\tau(-A)=-\mathcal T_\tau(A),
    \qquad
    \mathcal C_{w,\tau}(-A)=-\mathcal C_{w,\tau}(A).
    \label{eq:truncation-sign-equivariance}
\end{equation}
Indeed, changing $A$ to $-A$ changes every retained coefficient by a minus sign without changing which coefficients pass the threshold.

Let $\Pi_j$ denote the weight-and-coefficient truncation used after the $j$th noisy local rotation; the parameters may be fixed throughout, in which case $\Pi_j=\mathcal C_{w,\tau}$. For an arbitrary input $A$, define the untruncated Trotter iterates by
\begin{equation}
    O^{\mathrm{Tr}}_{0,H}(A):=A,
    \qquad
    O^{\mathrm{Tr}}_{j,H}(A)
    :=
    \mathcal G_{j,H}
    \bigl(O^{\mathrm{Tr}}_{j-1,H}(A)\bigr),
    \label{eq:random-H-untruncated-local-recursion}
\end{equation}
and the truncated iterates by
\begin{align}
    \widetilde O_{0,H}(A)
    &:=
    A,
    \nonumber\\
    X_{j,H}(A)
    &:=
    \mathcal G_{j,H}
    \bigl(\widetilde O_{j-1,H}(A)\bigr),
    \nonumber\\
    \widetilde O_{j,H}(A)
    &:=
    \Pi_j\bigl(X_{j,H}(A)\bigr),
    \qquad
    j=1,\ldots,M.
    \label{eq:random-H-truncated-local-recursion}
\end{align}
Thus the output of one noisy local rotation is first formed as $X_{j,H}$, then truncated, and only then passed to the next local rotation. Equivalently, the complete untruncated and truncated circuit maps are
\begin{align}
    \mathcal V_{H,t}
    &:=
    \mathcal G_{M,H}\circ\cdots\circ\mathcal G_{1,H},
    \label{eq:random-H-untruncated-Trotter-map}
    \\
    \widetilde{\mathcal V}_{H,t}
    &:=
    (\Pi_M\circ\mathcal G_{M,H})\circ\cdots\circ
    (\Pi_1\circ\mathcal G_{1,H}).
    \label{eq:random-H-truncated-Trotter-map}
\end{align}
In particular,
\begin{equation}
    O^{\mathrm{Tr}}_{M,H}(A)=\mathcal V_{H,t}(A),
    \qquad
    \widetilde O_{M,H}(A)=\widetilde{\mathcal V}_{H,t}(A).
    \label{eq:random-H-final-circuit-iterates}
\end{equation}

\begin{lemma}[Covariance of the complete approximation]
\label{lem:approximate-Pauli-covariance}
For every realization $H$, Pauli $R$, and operator $A$,
\begin{align}
    \mathcal V_{RHR,t}(RAR)
    &=
    R\mathcal V_{H,t}(A)R,
    \label{eq:untruncated-Trotter-Pauli-covariance}
    \\
    \widetilde{\mathcal V}_{RHR,t}(RAR)
    &=
    R\widetilde{\mathcal V}_{H,t}(A)R.
    \label{eq:approximate-Pauli-covariance}
\end{align}
Moreover, both circuit maps are odd in their input:
\begin{align}
    \mathcal V_{H,t}(-A)
    &=
    -\mathcal V_{H,t}(A),
    \nonumber\\
    \widetilde{\mathcal V}_{H,t}(-A)
    &=
    -\widetilde{\mathcal V}_{H,t}(A).
    \label{eq:approximate-sign-equivariance}
\end{align}
\end{lemma}

\begin{proof}
Equation~\eqref{eq:random-H-noisy-local-covariance} proves covariance of each noisy local rotation, while \cref{eq:joint-truncation-equivariance} proves covariance of every truncation. Induction on $j$ in \cref{eq:random-H-untruncated-local-recursion,eq:random-H-truncated-local-recursion} therefore gives covariance of every intermediate iterate and, at $j=M$, \cref{eq:untruncated-Trotter-Pauli-covariance,eq:approximate-Pauli-covariance}. Each $\mathcal G_{j,H}$ is linear and every $\Pi_j$ is odd by \cref{eq:truncation-sign-equivariance}; the same induction proves oddness.
\end{proof}

All statements above hold for a fixed realization of the Hamiltonian and its local terms. In particular, they remain true when the same sampled local terms are reused throughout the Trotter circuit.

\subsection{The simulation error as a hybrid exact--approximate evolution}
\label{supp:operational-hybrid-error}

We now define the error operator while keeping separate the two approximations made by the simulator. Fix a Pauli observable $P_0$ and set
\begin{align}
    O^{\mathrm{ex}}_H(t)
    &:=
    \mathcal A_{H,t}(P_0),
    \nonumber\\
    O^{\mathrm{Tr}}_H(t)
    &:=
    \mathcal V_{H,t}(P_0),
    \nonumber\\
    \widetilde O_H(t)
    &:=
    \widetilde{\mathcal V}_{H,t}(P_0).
    \label{eq:exact-approximate-random-H-observables}
\end{align}
Here $O^{\mathrm{ex}}_H(t)$ is the exact continuous-time observable, $O^{\mathrm{Tr}}_H(t)$ is the output of the same local Trotter circuit without any truncation, and $\widetilde O_H(t)$ is the actual truncated output.

Accordingly, define
\begin{align}
    \Delta^{\mathrm{PF}}_H(t)
    &:=
    O^{\mathrm{ex}}_H(t)-O^{\mathrm{Tr}}_H(t),
    \label{eq:random-H-product-formula-error}
    \\
    \Delta^{\mathrm{cut}}_H(t)
    &:=
    O^{\mathrm{Tr}}_H(t)-\widetilde O_H(t),
    \label{eq:random-H-truncation-error}
    \\
    \Delta_H(t)
    &:=
    O^{\mathrm{ex}}_H(t)-\widetilde O_H(t)
    \nonumber\\
    &=
    \Delta^{\mathrm{PF}}_H(t)+\Delta^{\mathrm{cut}}_H(t).
    \label{eq:random-Hamiltonian-operator-error}
\end{align}
The first term is the product-formula error. The second is the cumulative error produced by the truncations. We now resolve the latter at exactly the granularity at which truncation is performed.

For $j=0,1,\ldots,M$, define the hybrid observable
\begin{equation}
    Z_H^{(j)}(t)
    :=
    \mathcal G_{M,H}\circ\cdots\circ\mathcal G_{j+1,H}
    \bigl(\widetilde O_{j,H}(P_0)\bigr),
    \label{eq:random-H-hybrid-observables}
\end{equation}
where the suffix is the identity when $j=M$. In $Z_H^{(j)}(t)$, the first $j$ noisy local rotations are followed by their prescribed truncations, while the remaining $M-j$ rotations are applied without truncation. Thus
\begin{equation}
    Z_H^{(0)}(t)=O^{\mathrm{Tr}}_H(t),
    \qquad
    Z_H^{(M)}(t)=\widetilde O_H(t).
    \label{eq:random-H-hybrid-endpoints}
\end{equation}
The sequence therefore interpolates one local truncation at a time between the untruncated Trotter circuit and the actual algorithm.

The operator removed immediately after the $j$th noisy local rotation is
\begin{equation}
    D_{j,H}(t)
    :=
    X_{j,H}(P_0)-\widetilde O_{j,H}(P_0)
    =
    (\mathrm{Id}-\Pi_j)\bigl(X_{j,H}(P_0)\bigr).
    \label{eq:random-H-discarded-local-operator}
\end{equation}
Its contribution to the final output is obtained by propagating it through all subsequent untruncated noisy local rotations:
\begin{align}
    E_{j,H}(t)
    &:=
    Z_H^{(j-1)}(t)-Z_H^{(j)}(t)
    \nonumber\\
    &=
    \mathcal G_{M,H}\circ\cdots\circ\mathcal G_{j+1,H}
    \bigl(D_{j,H}(t)\bigr).
    \label{eq:random-H-propagated-local-error}
\end{align}
Telescoping the hybrid sequence gives
\begin{equation}
    \Delta^{\mathrm{cut}}_H(t)=\sum_{j=1}^{M}E_{j,H}(t),
    \label{eq:random-H-truncation-error-telescope}
\end{equation}
and hence the total simulation error has the explicit decomposition
\begin{equation}
    \Delta_H(t)
    =
    \Delta^{\mathrm{PF}}_H(t)+\sum_{j=1}^{M}E_{j,H}(t).
    \label{eq:random-H-total-error-decomposition}
\end{equation}
This identity distinguishes the product-formula error from the individual truncation errors and also records where every discarded operator enters the final observable.

The entire construction is covariant before any randomness is invoked. For an arbitrary input $A$, let the iterates, hybrids, and error terms above be defined with $P_0$ replaced by $A$. Exact covariance, \cref{eq:untruncated-Trotter-Pauli-covariance,eq:approximate-Pauli-covariance}, and induction through the hybrid sequence give
\begin{equation}
    Z_{RHR}^{(j)}(t;RAR)=RZ_H^{(j)}(t;A)R.
    \label{eq:hybrid-Pauli-covariance}
\end{equation}
All these maps are odd in $A$: the exact and untruncated evolutions are linear, and the truncated evolution is odd by \cref{eq:approximate-sign-equivariance}. Set
\begin{equation}
    \sigma_R:=\chi_R(P_0),
    \qquad
    RP_0R=\sigma_RP_0.
\end{equation}
Using $P_0=R(\sigma_RP_0)R$ then yields
\begin{align}
    \Delta^{\mathrm{PF}}_{RHR}(t)
    &=
    \sigma_R R\Delta^{\mathrm{PF}}_H(t)R,
    \label{eq:twisted-product-formula-error-covariance}
    \\
    E_{j,RHR}(t)
    &=
    \sigma_R RE_{j,H}(t)R
    \qquad
    (j=1,\ldots,M).
    \label{eq:twisted-local-error-covariance}
\end{align}
Summing these identities in \cref{eq:random-H-total-error-decomposition} proves the covariance of the total error.

\begin{lemma}[Twisted covariance of the simulation error]
\label{lem:twisted-error-covariance}
For every realization $H$ and every Pauli $R$,
\begin{equation}
    \Delta_{RHR}(t)=\chi_R(P_0)R\Delta_H(t)R.
    \label{eq:twisted-error-covariance}
\end{equation}
\end{lemma}

Up to this point, no average over Hamiltonians has been taken. The product-formula error, every propagated truncation error, and their sum satisfy the same deterministic transformation rule. They may be strongly correlated, because they are all functions of the same sampled Hamiltonian and its local terms.

\subsection{What Pauli-invariant randomness does}
\label{supp:operational-error-orthogonality}

Expand the total error in the Pauli basis:
\begin{equation}
    \Delta_H(t)=\sum_{P\in\cP_n}\delta_P(H,t)P,
    \qquad
    \delta_P(H,t)
    :=
    \frac{1}{2^n}\Tr\!\left(P\Delta_H(t)\right).
    \label{eq:random-error-Pauli-expansion}
\end{equation}
The normalized Hilbert--Schmidt norm is the total squared Pauli-coefficient weight:
\begin{equation}
    \norm{\Delta_H(t)}_2^2
    =
    \sum_{P\in\cP_n}|\delta_P(H,t)|^2.
    \label{eq:random-error-HS-coefficients}
\end{equation}
Taking the coefficient of $P$ in \cref{eq:twisted-error-covariance} gives the deterministic transformation law
\begin{equation}
    \delta_P(RHR,t)
    =
    \chi_R(P_0)\chi_R(P)\delta_P(H,t).
    \label{eq:random-error-coefficient-transformation}
\end{equation}
For a fixed realization, this equation only relates two different Hamiltonians. It does not by itself imply any cancellation. We now use Pauli invariance of the ensemble to say that those two Hamiltonians occur with the same probability law.

For Pauli strings $P,Q$, define the second-moment matrix
\begin{equation}
    M_{P,Q}(t)
    :=
    \mathbb E_H
    \left[
        \overline{\delta_P(H,t)}\delta_Q(H,t)
    \right].
    \label{eq:random-error-second-moment-matrix}
\end{equation}
By joint Pauli invariance of the local terms, for every fixed $R\in\cP_n$,
\begin{align}
    M_{P,Q}(t)
    &=
    \mathbb E_H
    \left[
        \overline{\delta_P(RHR,t)}\delta_Q(RHR,t)
    \right]
    \nonumber\\
    &=
    \chi_R(P_0)^2\chi_R(P)\chi_R(Q)M_{P,Q}(t)
    \nonumber\\
    &=
    \chi_R(P)\chi_R(Q)M_{P,Q}(t).
    \label{eq:random-error-covariance-character}
\end{align}
The sign associated with the fixed initial observable appears twice and cancels. If $P\neq Q$, choose $R$ as in \cref{eq:Pauli-characters-separate}. Then $\chi_R(P)\chi_R(Q)=-1$, so \cref{eq:random-error-covariance-character} gives $M_{P,Q}(t)=-M_{P,Q}(t)$.

\begin{lemma}[Orthogonality of distinct Pauli coefficients]
\label{lem:random-error-Pauli-orthogonality}
If the Hamiltonian ensemble is Pauli invariant, then, for every $P\neq Q$,
\begin{equation}
    \mathbb E_H
    \left[
        \overline{\delta_P(H,t)}\delta_Q(H,t)
    \right]
    =0.
    \label{eq:random-error-Pauli-orthogonality}
\end{equation}
\end{lemma}

This is the precise consequence of randomness. It does not assert that the Pauli coefficients are independent, nor that the error is small for every sample. It says that distinct Pauli coefficients are orthogonal as random variables in $L_2$ of the Hamiltonian ensemble. All correlations generated by reusing the same Hamiltonian throughout the simulation are allowed; the only off-diagonal correlations ruled out are those between different Pauli labels.

\subsection{Expectation values on a fixed input state}
\label{supp:operational-random-H-guarantee}

The operational consequence can now be seen directly. Let $\rho$ be a fixed density operator, chosen independently of the Hamiltonian realization, and write
\begin{equation}
    \rho=\frac{1}{2^n}\sum_{P\in\cP_n}r_PP,
    \qquad
    r_P:=\Tr(\rho P).
    \label{eq:state-Pauli-expansion}
\end{equation}
Since every Pauli is unitary,
\begin{equation}
    |r_P|\leq1.
    \label{eq:state-Pauli-coefficient-bound}
\end{equation}
Using \cref{eq:random-error-Pauli-expansion}, the expectation-value error is
\begin{equation}
    \Tr\!\left(\rho\Delta_H(t)\right)
    =
    \sum_{P\in\cP_n}r_P\delta_P(H,t).
    \label{eq:expectation-error-Pauli-expansion}
\end{equation}
Without averaging, the square of this sum contains interference terms between different Pauli coefficients. Those terms are the reason that a normalized Hilbert--Schmidt bound does not, in general, control a fixed-state expectation value. After averaging over a Pauli-invariant Hamiltonian ensemble, however,
\begin{align}
    &\mathbb E_H
    \left|
        \Tr\!\left(\rho\Delta_H(t)\right)
    \right|^2
    \nonumber\\
    &\quad=
    \sum_{P,Q\in\cP_n}\overline{r_P}r_Q
    \mathbb E_H
    \left[
        \overline{\delta_P(H,t)}\delta_Q(H,t)
    \right]
    \nonumber\\
    &\quad=
    \sum_{P\in\cP_n}|r_P|^2
    \mathbb E_H|\delta_P(H,t)|^2
    \nonumber\\
    &\quad\leq
    \sum_{P\in\cP_n}
    \mathbb E_H|\delta_P(H,t)|^2
    \nonumber\\
    &\quad=
    \mathbb E_H\norm{\Delta_H(t)}_2^2.
    \label{eq:random-H-fixed-state-master-bound}
\end{align}
The second equality is exactly where Pauli-invariant randomness is used: it removes every $P\neq Q$ term. The remaining diagonal sum is the normalized Hilbert--Schmidt error.

\begin{theorem}[Fixed-state guarantee for Pauli-invariant dynamics]
\label{thm:random-H}
\label{thm:supp-operational}
Let $H=\sum_{a\in\mathcal A}H_a$ be drawn from a Pauli-invariant ensemble, let $P_0$ be a fixed Pauli observable, and let $O^{\mathrm{ex}}_H(t)$ and $\widetilde O_H(t)$ be the exact and Trotterized-and-truncated observables in \cref{eq:exact-approximate-random-H-observables}. If
\begin{equation}
    \mathbb E_H\norm{\Delta_H(t)}_2^2\leq\varepsilon^2,
    \label{eq:random-H-mean-square-HS-error}
\end{equation}
then, for every density operator $\rho$ fixed independently of $H$,
\begin{equation}
    \mathbb E_H
    \left|
        \Tr\!\left(\rho O^{\mathrm{ex}}_H(t)\right)
        -
        \Tr\!\left(\rho\widetilde O_H(t)\right)
    \right|^2
    \leq\varepsilon^2.
    \label{eq:random-H-expectation-value-bound}
\end{equation}
\end{theorem}

\begin{proof}
Equation~\eqref{eq:random-H-fixed-state-master-bound} bounds the left-hand side by $\mathbb E_H\norm{\Delta_H(t)}_2^2$, and the latter is at most $\varepsilon^2$ by assumption.
\end{proof}

A deterministic guarantee
\begin{equation}
    \norm{\Delta_H(t)}_2\leq\varepsilon
    \qquad
    \text{for every realization $H$}
    \label{eq:random-H-deterministic-HS-error}
\end{equation}
immediately implies \cref{eq:random-H-mean-square-HS-error}. Thus the deterministic simulation bounds proved earlier in the manuscript can be inserted directly into \cref{thm:random-H}.

By Cauchy--Schwarz, \cref{eq:random-H-expectation-value-bound} also gives
\begin{equation}
    \mathbb E_H
    \left|
        \Tr\!\left(\rho\Delta_H(t)\right)
    \right|
    \leq\varepsilon.
    \label{eq:random-H-mean-absolute-error}
\end{equation}
For every $0<\eta<1$, Markov's inequality applied to the squared error yields
\begin{equation}
    \Pr_H\!\left[
        \left|
            \Tr\!\left(\rho\Delta_H(t)\right)
        \right|
        \geq
        \frac{\varepsilon}{\sqrt{\eta}}
    \right]
    \leq\eta.
    \label{eq:random-H-high-probability-error}
\end{equation}
Equivalently, with probability at least $1-\eta$ over the Hamiltonian sample,
\begin{equation}
    \left|
        \Tr\!\left(\rho O^{\mathrm{ex}}_H(t)\right)
        -
        \Tr\!\left(\rho\widetilde O_H(t)\right)
    \right|
    \leq
    \frac{\varepsilon}{\sqrt{\eta}}.
    \label{eq:random-H-high-probability-operational}
\end{equation}

The quantifiers are important. The state $\rho$ must be fixed independently of the Hamiltonian sample. The result does not cover a state chosen adversarially after the realization of $H$ is known.

\medskip

We conclude this section by extending our guarantees also to symmetric product formulas and Richardson extrapolation.

\begin{corollary}[Symmetric formulas and Richardson extrapolation]
\label{cor:random-H-richardson}
The conclusion of \cref{thm:random-H} also holds for a symmetric product formula and for its coefficient-truncated version. It also holds for an extrapolated approximation
\[
    \widetilde{\mathcal V}^{[q]}_{H,t}(A)
    :=\sum_{j=1}^{q}\omega_j
      \widetilde{\mathcal V}_{j,H,t}(A),
    \qquad \sum_{j=1}^{q}\omega_j=1,
\]
provided the circuit choices, cutoffs, and real coefficients $\omega_j$ are unchanged under simultaneous Pauli conjugation of the local Hamiltonian terms. Here each $\widetilde{\mathcal V}_{j,H,t}$ is the full approximation for the $j$th circuit, with truncation omitted for the quantum construction.
\end{corollary}

\begin{proof}
Each complete symmetric update is covariant, as shown above, and every coefficient truncation is covariant and odd. Hence each circuit approximation has these properties. A fixed real linear combination preserves both properties. For the input Pauli $P_0$, the total extrapolation error therefore obeys
\[
    \Delta^{[q]}_{RHR}(t)
    =\chi_R(P_0)R\Delta^{[q]}_H(t)R.
\]
The coefficient-orthogonality argument and \cref{eq:random-H-fixed-state-master-bound} apply without change.
\end{proof}

For a general fixed Hermitian observable $O=\sum_P o_PP$, a fixed-state guarantee follows by simulating each Pauli term separately. Set $B_1:=\sum_P|o_P|$ and suppose $B_1>0$. If the approximation to each $P$ has mean-square Pauli-$2$ error at most $(\epsilon/B_1)^2$, the theorem and the triangle inequality in $L_2$ of the Hamiltonian ensemble give
\begin{align*}
    &\left(\mathbb E_H\left|
      \Tr\!\left[\rho\left(\e^{t\mathcal L_H}(O)
          -\sum_Po_P\widetilde O_{P,H}(t)\right)\right]
       \right|^2\right)^{1/2}\\
    &\quad\leq\sum_P|o_P|
       \left(\mathbb E_H\left|
          \Tr\!\left[\rho\left(\e^{t\mathcal L_H}(P)
             -\widetilde O_{P,H}(t)\right)\right]
       \right|^2\right)^{1/2}
    \leq\epsilon.
\end{align*}
The case $B_1=0$ is trivial. If $O$ is supported on at most $w_0$ qubits, there are at most $4^{w_0}$ terms and $B_1\leq2^{w_0}\norm{O}_2$. Thus the overhead is independent of system size at fixed $w_0$. This construction uses separate Pauli simulations; it does not assert the same fixed-state bound for a single nonlinear truncated simulation of an arbitrary $O$.

\suppsection{Numerical details}
\label{supp:numerics}

We use the magnetization and Hamiltonian defined in the main text, with $J=1$, $J_z=-1.8$, and distance exponent $\alpha=3$. The size scan covers $n=3,\ldots,12$ at $T=0.3$ and $\gamma\in\{0,0.05,0.2,0.5,1,1.5,2,3,4\}$. The qubits occupy the first $n$ sites, filled row by row, of rectangles $2\times2$ for $n=3,4$, $3\times2$ for $n=5,6$, $3\times3$ for $n=7,8,9$, and $3\times4$ for $n=10,11,12$. All patches have open boundaries and unit lattice spacing. Couplings use Euclidean distances between every pair, without a size-dependent rescaling. The separate $7\times7$ dynamics runs use step size $\delta=3/320$ and are shown up to $t=0.9$.

For the size scan, we integrate $\dot M_x=i[H,M_x]+\cD_\gamma(M_x)$ with \texttt{QuantumToolbox.jl}. Local depolarization uses jump operators $L_{i,\mu}=\sqrt{\gamma/4}\,\sigma_i^\mu$, giving $\cD_\gamma(P)=-\gamma\abs{P}P$. Relative and absolute solver tolerances are $10^{-10}$ and $10^{-12}$. We evaluate the total error $\varepsilon_{\mathrm{tot}}(N,\tau)$ defined in \cref{eq:numerical-total-error} from the Pauli coefficients, using this continuous solution as the reference. For $\tau=0$, the same measure $\varepsilon_{\mathrm{tot}}(r,0)$ is the untruncated Trotter error. Classical runs are selected using $\varepsilon_{\mathrm{tot}}(N,\tau)\leq0.01$ directly.

Group the Hamiltonian into its $XX$, $YY$, and $ZZ$ blocks, $H=H_x+H_y+H_z$, with $\mathcal H_\mu(A)=i[H_\mu,A]$. For $\delta=T/N$, one first-order Trotter step is
\begin{equation}
    \Phi_\delta=
        \e^{\delta\cD_\gamma}
        \e^{\delta\mathcal H_z}
        \e^{\delta\mathcal H_y}
        \e^{\delta\mathcal H_x},
    \qquad
    \widetilde M_{x;N,0}(T)=\Phi_\delta^N(M_x(0)).
    \label{eq:numerical-trotter-step}
\end{equation}
The noise layer multiplies a Pauli of weight $w$ by $\e^{-\gamma w\delta}$, so refining the step size preserves the physical noise rate. PP discards coefficients with magnitude below $\tau$ after individual gates, keeping a fixed gate order. No additional Pauli-weight cutoff or Richardson extrapolation is used.

A Trotter step contains several hardware layers. Within each block, all gates commute, but simultaneous gates act on disjoint qubits. Assuming all-to-all connectivity and native two-qubit Pauli rotations, the pairs can be scheduled in $n-1$ layers for even $n$ and $n$ layers for odd $n$. For $\gamma>0$, we also count one parallel layer of single-qubit depolarizing channels per Trotter step. The total depth is therefore
\begin{equation}
    \begin{aligned}
        D(N,n,\gamma)&=N\bigl[3\ell(n)+\mathbf{1}_{\{\gamma>0\}}\bigr],\\
        \ell(n)&=\begin{cases}
            n-1,&n\text{ even},\\
            n,&n\text{ odd}.
        \end{cases}
    \end{aligned}
    \label{eq:numerical-circuit-depth}
\end{equation}
For simplicity, we exclude routing, channel synthesis, state preparation, and measurement. For the main figure, we interpolate the two untruncated-error measurements bracketing $0.01$ on logarithmic axes, round the inferred step count to the nearest integer, and apply this conversion.

We fit $\log D=a_D+b_D\log n$ and $\log t_{\mathrm{PP}}=a_{\mathrm{PP}}+b_{\mathrm{PP}}\log n$ by ordinary least squares over all ten sizes $n=3,\ldots,12$, separately for each noise strength displayed in \cref{fig:relative_resource}. Depth includes the noise layers in Eq.~\eqref{eq:numerical-circuit-depth}. The fitted exponents are
\begin{equation}
\begin{aligned}
\gamma&=(0,\,0.05,\,0.2,\,0.5,\,1,\,2,\,3,\,4),\\
b_D&=(1.73,\,1.64,\,1.63,\,1.56,\,1.47,\,1.31,\,1.20,\,1.15),\\
b_{\mathrm{PP}}&=(9.70,\,9.65,\,9.55,\,9.53,\,9.25,\,8.52,\,8.08,\,7.31).
\end{aligned}
\label{eq:resource-fit-exponents}
\end{equation}

We define $|\psi\rangle=\bigotimes_{i=1}^n|\psi_i\rangle$ by independent
draws $z_i\sim\mathrm{Unif}[-1,1]$ and
$\phi_i\sim\mathrm{Unif}[0,2\pi)$, with
\[
|\psi_i\rangle=\sqrt{\frac{1+z_i}{2}}|0\rangle
+\e^{i\phi_i}\sqrt{\frac{1-z_i}{2}}|1\rangle.
\]
Each qubit is uniform on the Bloch sphere, so
$\mathbb E_\psi[|\psi\rangle\langle\psi|]=I/2^n$ and the ensemble
is a state $1$-design.
For the difference from the numerical continuous-time reference,
$\Delta O=\widetilde O(T)-O(T)=\sum_P\Delta c_P P$, the RMS error is
\[
\varepsilon_{\mathrm{ens}}
=\left(\mathbb E_\psi|\langle\psi|\Delta O|\psi\rangle|^2\right)^{1/2}
=\left(\sum_P3^{-|P|}|\Delta c_P|^2\right)^{1/2},
\]
where $|P|$ is the Pauli weight.

\begin{figure}[h!]
    \centering
    \includegraphics[width=\linewidth]{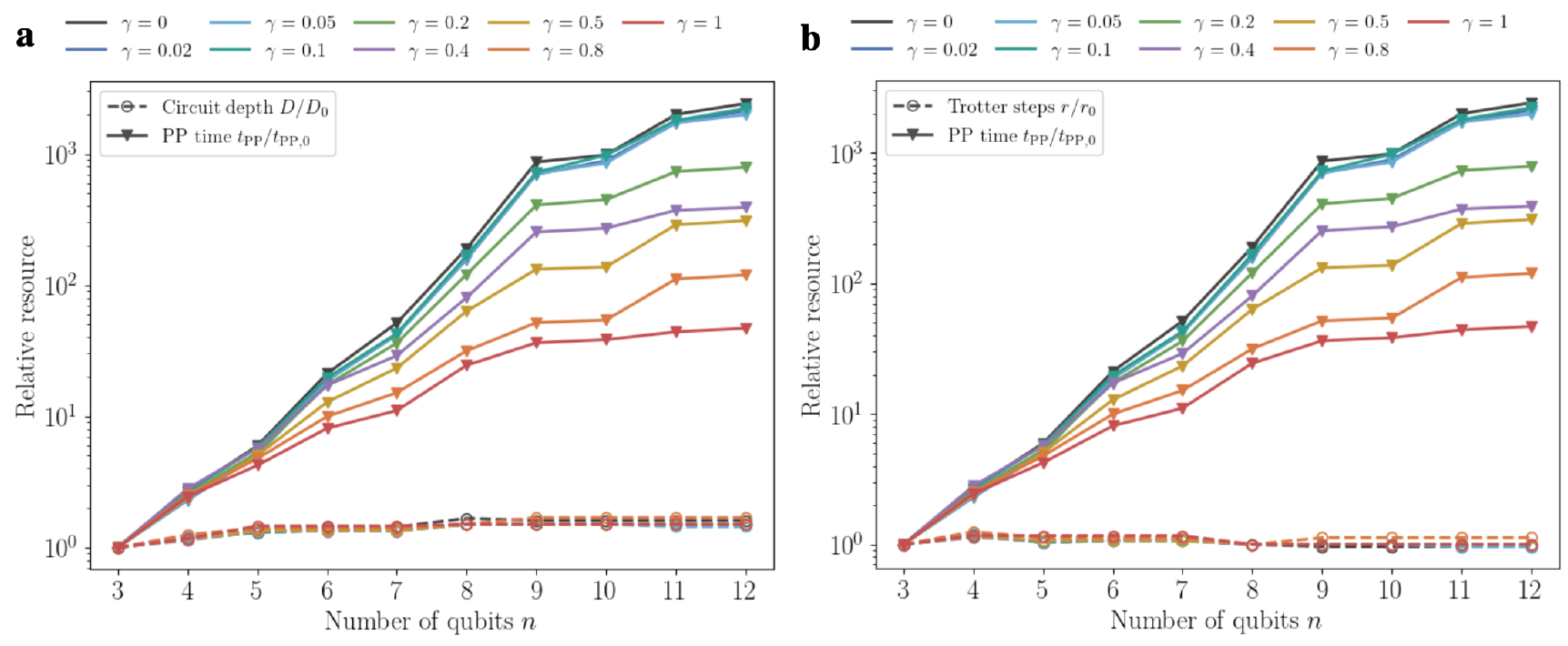}
    \caption{\textit{RMS for $1$-design ensemble.}
Relative resource growth for the two-dimensional TFIM at
$J=1$, $g_x=1.4$, and $T=1$, using a target RMS
expectation-value error $\varepsilon_{\mathrm{ens}}=0.01$
over independent uniform pure product states.
Dashed curves with open circles show the estimated
(a) circuit depth and (b) Trotter step count of the
first-order quantum circuit; solid curves with squares show
the fastest measured Pauli-propagation time meeting the
total-error target. Quantum step counts are obtained by
log--log interpolation of the untruncated error and rounding
to the nearest integer.
Depth includes noise:
$D=r[\Delta+1+\mathbf{1}_{\{\gamma>0\}}]$, where $\Delta$
is the maximum degree of the lattice patch, with one
additional layer of single-qubit rotations and, for positive
onsite noise rate $\gamma$, one parallel noise layer.
Each resource is normalized to its own value at $n=3$
for the same $\gamma$. 
}
    \label{fig:1design}
\end{figure}

\newpage
\suppsection{Comparison to other algorithms}
\label{supp:algorithm-comparison}
We devote this section to comparing the performance of different algorithms to the two algorithms proposed here. 

\subsection{Quantum algorithm}\label{sec:quantum_compar}

Comparing the performance of different quantum algorithms is difficult, as it depends on intricate details of the block-encoding oracles and the implementation of the Lindblad operators, among other things. Here we make our best effort to compare different methods fairly by estimating their gate counts. We denote by $G$ the maximum gate count among the circuits used by an algorithm and report $G/m$ as a normalized gate count, where $m$ is the number of local blocks. Compiling these gates into a fixed universal gate set can introduce additional logarithmic factors.

\begin{table}[ht]
\centering
\footnotesize
\setlength{\tabcolsep}{1.5pt}
\renewcommand{\arraystretch}{1.4}
\caption{\textbf{Total gate count divided by the number of local blocks $m$.} The entries give coarse upper bounds for one circuit at fixed $k,d,w,J,\gamma,\bar\gamma$ in the large-time regime.}
\begin{tabular}{*{8}{p{\dimexpr(\textwidth-24pt)/8\relax}}}
\hline\hline
 \raggedright \textbf{This work}\newline most states, local $\mathcal L$, Richardson extrapolation
 & \raggedright \textbf{Ref.~\cite{wang2026lindbladian}}\newline worst case, Richardson extrapolation, local $\mathcal L$
 & \raggedright \textbf{Ref.~\cite{barthel2012quasilocality}}\newline worst case, light cone, lattice $\mathcal L$, local observables
 & \raggedright \textbf{Ref.~\cite{wang2026query,chen2026query}}\newline worst case, block encoding
 & \raggedright \textbf{Ref.~\cite{chen2025randomized}}\newline averaged channel, qDRIFT
 & \raggedright \textbf{Ref.~\cite{kato2026exponentially}}\newline worst case, randomised Taylor series
 & \raggedright \textbf{Ref.~\cite{ding2024simulating,ding2024single}}\newline state error, Stinespring dilation
 & \raggedright \textbf{Ref.~\cite{cleve2016efficient}}\newline worst case, LCU
 \tabularnewline
\hline
 \raggedright $\mathrm{poly}(t,\allowbreak \log(1/\varepsilon))$
 & \raggedright $\mathrm{poly}(t,\allowbreak m,\allowbreak \log(1/\varepsilon))$
 & \raggedright $t^{O(\kappa)}\,\allowbreak O(1/\varepsilon)$
 & \raggedright $ O(mt\allowbreak+\log(1/\varepsilon))$
 & \raggedright $\mathrm{poly}(m,\allowbreak t,\allowbreak 1/\varepsilon)$
 & \raggedright $\mathrm{poly}(m,\allowbreak t)\,\allowbreak \times\log(1/\varepsilon)$
 & \raggedright $\mathrm{poly}(m,t,1/\varepsilon)$
 & \raggedright $\mathrm{poly}(m,\allowbreak t,\allowbreak \log(1/\varepsilon))$
 \tabularnewline
\hline\hline
\end{tabular}
\label{tab:gate-quant-comparison}
\end{table}

Table~\ref{tab:gate-quant-comparison} presents our best effort at an honest comparison across different methods, which may also require multiple runs, ancilla qubits, and other resources. We focus on the works that align most closely with our algorithm. However, several other algorithms simulate Lindbladian or Hamiltonian evolution~\cite{low2019well,rendon2024improved,watson2024randomly,watson2025exponentially}. The methods included in the table are described below. Ref.~\cite{wang2026lindbladian} uses a Trotter formula with Richardson extrapolation and analyses nested-commutator bounds. Ref.~\cite{barthel2012quasilocality} uses a Lieb--Robinson light cone: only the terms inside the light cone of the observable are Trotterised. This requires a lattice of finite dimension $\kappa$, meaning that the number of sites within distance $r$ grows at most as $r^{\kappa}$ (a chain has $\kappa=1$, a square lattice $\kappa=2$). Refs.~\cite{wang2026query,chen2026query} use transducers on block encodings. Ref.~\cite{chen2025randomized} uses qDRIFT: at each step, a single Lindbladian term is sampled at random, with probability proportional to its strength, and applied. Ref.~\cite{kato2026exponentially} uses randomised Taylor series. A truncated Taylor expansion of $e^{t\mathcal L}$ is compiled into randomly sampled circuits, whose outcomes are averaged to estimate expectation values. Refs.~\cite{ding2024simulating,ding2024single} use Stinespring dilation with reset. Each step consists of Hamiltonian evolution on the system and ancillas, followed by resetting the ancillas to reproduce the dissipation. Ref.~\cite{cleve2016efficient} uses linear combination of unitaries (LCU) for channels. The channel for a short time step is written as a linear combination of unitaries and implemented coherently using oblivious amplitude amplification.

The main difference between these results and the work of this manuscript lies in the scaling with system size. For fixed local parameters, our circuit depth and normalized gate count $G/m$ are independent of the number of qubits, while $m$ may still depend on it. This improvement comes from using dissipation to suppress high-weight Pauli components, which allows larger Trotter steps. We also do not require a finite-dimensional lattice, as in the light-cone approach of Ref.~\cite{barthel2012quasilocality}, but only bounded locality and degree for the interaction structure. These advantages, however, come with restrictions: we consider initially local observables and Pauli-diagonal noise that damps every nonidentity Pauli mode on each local block, whereas several of the other methods cover more general Lindbladians. Crucially, our expectation-value guarantees hold for most states, while several of the other methods provide guarantees for all input states.

\subsection{Classical simulation algorithm}\label{sec:classical_compar}
Classical algorithms are easier to compare, in the sense that they share a crude common metric: runtime. However, the dependence on different parameters and the assumptions behind each result can still make the comparison difficult.

\begin{table*}[ht]
\centering
\caption{\textbf{Classical runtime.} Classical algorithms for open-system and Hamiltonian dynamics, compared by their runtime for estimating a local expectation value $f$ with error $\varepsilon$. We fix $k,d,w_0,J$ whenever these parameters apply.}
\label{tab:classical-comparison}
\scriptsize
\setlength{\tabcolsep}{1.5pt}
\renewcommand{\arraystretch}{1.2}
\begin{tabular}{p{\dimexpr0.18\textwidth-2\tabcolsep\relax}p{\dimexpr0.18\textwidth-2\tabcolsep\relax}p{\dimexpr0.24\textwidth-2\tabcolsep\relax}p{\dimexpr0.16\textwidth-2\tabcolsep\relax}p{\dimexpr0.24\textwidth-2\tabcolsep\relax}}
\hline\hline
 \raggedright \textbf{This work}\newline continuous time, most states, Pauli propagation, weak damping
 & \raggedright \textbf{Ref.~\cite{schuster2024polynomial}}\newline noisy circuits, deterministic, most states, Pauli paths and propagation
 & \raggedright \textbf{Ref.~\cite{barthel2012quasilocality}}\newline continuous time, worst case, light cone, lattice $\mathcal L$, local observables
 & \raggedright \textbf{Ref.~\cite{wild2023classical}}\newline continuous time, worst case over product states, cluster expansion, noiseless $H$
 & \raggedright \textbf{Ref.~\cite{xu2026classical}}\newline continuous (short) time, entangled states, low-weight Pauli dynamics, noiseless $H$
 \tabularnewline
\hline
 \raggedright $\left({t}/{\varepsilon}\right)^{O_k(J/\gamma)}$
 & \raggedright $(e/\varepsilon)^{ O({\rm poly}(1/(\gamma\delta)))}$
 & \raggedright $e^{O({\rm poly}(\overline\gamma,t, \log^\kappa(1/\varepsilon)))}$
 & \raggedright $(e/\varepsilon)^{e^{O(t)}}$
 & \raggedright $O\bigl({\rm poly}(t,1/\varepsilon)n^{O(\log(t/\varepsilon)/\log1/t)}\bigr)$
 \tabularnewline
\hline\hline
\end{tabular}
\end{table*}

As in the previous section, Table~\ref{tab:classical-comparison} summarizes the main differences between our method and the works most directly comparable to it. Ref.~\cite{schuster2024polynomial} treats uniform noise in circuits of depth $L$. Its single-qubit depolarising factor is $e^{-\gamma\delta}$, where $\delta$ is the noise duration per layer. For our local observable, the number of initial Pauli terms is $O(1)$. Ref.~\cite{barthel2012quasilocality} restricts evolution to a Lieb--Robinson neighbourhood on a finite-range lattice with volume growth $R^\kappa$. Ref.~\cite{wild2023classical} uses a cluster expansion and analytic continuation for a bounded-degree Hamiltonian. Ref.~\cite{xu2026classical} truncates high-weight Pauli operators for $t<t_0$ and requires the quantitative entanglement condition. Its entry uses a first-order product formula.

One of the main differences is that our bounds justify coefficient truncation after merging contributions to the same Pauli operator. This allows us to discard small combined coefficients, alongside high-weight terms, while controlling the accumulated error. We can therefore exploit cancellations between contributions without tracking every Pauli path separately. Our analysis also applies as the Trotter step shrinks and the noise per step vanishes, since it compares the coherent and dissipative rates directly. In addition, we do not require noise on every qubit at every layer. Noise is paired with the active local blocks, and idle qubits can remain noiseless in the gate-noise setting. Ref.~\cite{schuster2024polynomial} also considers gate noise, but its general bound in this setting is quasipolynomial, while ours is polynomial for fixed gate and noise parameters.

Our continuous-time bound does not require a finite-dimensional lattice or a short-time restriction, and its dependence on $t$ and $1/\varepsilon$ is polynomial at fixed positive damping and other local parameters. These advantages rely on Pauli-diagonal noise that damps every nonidentity Pauli mode on each active local block, as well as efficient access to the required initial-state Pauli expectations. As in Ref.~\cite{schuster2024polynomial}, our expectation-value guarantees hold for most states.

Beyond the results in Table~\ref{tab:classical-comparison} several other works study the simulation of Lindblad dynamics under or noisy circuits~\cite{donvil2022quantum,sander2025large,barthel2012quasilocality,wild2023classical,gao2018efficient,noh2020efficient,fontana2023classical,aharonov2022polynomial,schuster2024polynomial,gonzalez2024pauli,angrisani2025simulating, martinez2025efficient}.

\end{document}